\pdfoutput=1
\documentclass[12pt]{article}
\usepackage[utf8]{inputenc}%
\usepackage[tracking=true, letterspace=50, expansion=false]{microtype}

\usepackage{amsmath} %
\usepackage{amsfonts} %
\usepackage{amsthm} %
\usepackage{amssymb} %

\usepackage[margin=0.9in]{geometry} %
\usepackage{setspace} %
\usepackage{verbatim} %
\usepackage{natbib} %
\usepackage[title]{appendix} %
\usepackage{rotating} %
\usepackage{multirow} %
\usepackage{booktabs} %
\usepackage[inline]{enumitem} %
\usepackage[normal, online, flushleft]{threeparttable}
\usepackage{pdflscape}

\newtheorem{theorem}{Theorem}[section] %
\newtheorem{assumption}{Assumption}[section] %

\theoremstyle{definition}%
\newtheorem{remark}{Remark}[section]

\DeclareMathOperator{\cv}{cv} 
\DeclareMathOperator{\sign}{sign} 
\DeclareMathOperator{\maxbias}{\overline{bias}} 
\DeclareMathOperator*{\argmin}{argmin} 
\DeclareMathOperator{\var}{var} 
\DeclareMathOperator{\sd}{sd} 
\DeclareMathOperator{\se}{se} 
\newcommand{\hatse}{\widehat{\se}} 

\newcommand{\dd}{\mathrm{d}} 

\newcommand{\po}{{q}} 
\newcommand{\SC}{{M}} 

\newcommand{\FSY}[1]{\mathcal{F}_{\textnormal{T},#1}} 
\newcommand{\FHol}[1]{\mathcal{F}_{\textnormal{Höl},#1}} 
\newcommand{\hrmse}{h^{*}_{\textsc{rmse}}} 

\usepackage{mathtools}
\DeclarePairedDelimiter\hor{[}{)}
\DeclarePairedDelimiter\abs{\lvert}{\rvert}

\renewcommand*{\arraystretch}{1.2}

\newcommand{\hpt}{{h}^*_{\textsc{pt}}} 
\newcommand{\hmsedpi}{\hat{h}^*_{\textsc{pt}}} %
\newcommand{\bmsedpi}{\hat{b}^*_{\textsc{pt}}} %
\newcommand{\hcedpi}{\hat{h}_{\textsc{ce}}} %
\newcommand{\bcedpi}{\hat{b}_{\textsc{ce}}} %
\newcommand{\hhrmse}[1]{\hat{h}^{*}_{\textsc{rmse},#1}} %
\newcommand{\hhrmsez}{\hat{h}^{*}_{\textsc{rmse}}} %
\newcommand{\Mrot}{\hat{M}_{\textsc{rot}}} %
\newcommand{\hfgrot}{\ensuremath\hat{h}^*_{\textsc{pt,rot}}} %

\newcommand{\FLCI}{\text{FLCI}} %
\newcommand{\RMSE}{\text{RMSE}} %
\newcommand{\OCI}{\text{OCI}} %

\usepackage{xcolor}%
\definecolor{webbrown}{rgb}{.6,0,0}

\DeclarePairedDelimiter\indicatorfence{\{}{\}}
\newcommand\1{\operatorname{I}\indicatorfence}

\usepackage{tikz} %
\usepackage{longtable} %
\usepackage{threeparttablex} %
\makeatletter
\g@addto@macro\TPT@defaults{\small}
\makeatother

\usepackage{xr}
\usepackage{hyperref}
\hypersetup{%
  breaklinks = true,%
  colorlinks = true,%
  anchorcolor = webbrown,%
  citecolor = webbrown,%
  filecolor = webbrown,%
  linkcolor = webbrown,%
  menucolor = webbrown,%
  urlcolor= webbrown,%
  citebordercolor= 1 0 0,%
  menubordercolor=1 0 0,%
  urlbordercolor=1 0 0,%
  runbordercolor=1 0 0,%
  pdftitle=Simple and Honest Confidence Intervals in Nonparametric Regression,%
  pdfauthor=Tim Armstrong \& Michal Kolesár}
\usepackage{cleveref}
\Crefname{equation}{Eq.}{Eqs.}
\crefname{assumption}{assumption}{assumptions}
\crefname{remark}{remark}{remarks}
\crefname{proposition}{proposition}{propositions}
\crefname{sappsec}{Supplemental Appendix}{Supplemental Appendices}
\crefname{sappsubsec}{Supplemental Appendix}{Supplemental Appendices}
\crefname{sappsubsubsec}{Supplemental Appendix}{Supplemental Appendices}
\crefname{appsec}{Appendix}{Appendices}

\title{Simple and Honest Confidence Intervals in Nonparametric
  Regression\thanks{We thank Don Andrews, Sebiastian Calonico, Matias Cattaneo,
    Max Farrell, Christoph Rothe and numerous seminar and conference
    participants for helpful comments and suggestions. We thank Kwok Hao Lee for
    research assistance. All remaining errors are our own. The research of the
    first author was supported by National Science Foundation Grant SES-1628939.
    The research of the second author was supported by National Science
    Foundation Grant SES-1628878.}}
\author{Timothy B. Armstrong\thanks{email: \texttt{timothy.armstrong@yale.edu}}\\
  Yale University \and
  Michal Kolesár\thanks{email: \texttt{mkolesar@princeton.edu}}\\
  Princeton University}%
\date{\today}

\begin{document}

\maketitle

\begin{abstract}
  We consider the problem of constructing honest confidence intervals (CIs) for
  a scalar parameter of interest, such as the regression discontinuity
  parameter, in nonparametric regression based on kernel or local polynomial
  estimators. To ensure that our CIs are honest, we use critical values that
  take into account the possible bias of the estimator upon which the CIs are
  based. We show that this approach leads to CIs that are more efficient than
  conventional CIs that achieve coverage by undersmoothing or subtracting an
  estimate of the bias. We give sharp efficiency bounds of using different
  kernels, and derive the optimal bandwidth for constructing honest CIs. We show
  that using the bandwidth that minimizes the maximum mean-squared error results
  in CIs that are nearly efficient and that in this case, the critical value
  depends only on the rate of convergence. For the common case in which the rate
  of convergence is $n^{-2/5}$, the appropriate critical value for 95\% CIs is
  2.18, rather than the usual 1.96 critical value. We illustrate our results in
  a Monte Carlo analysis and an empirical application.
\end{abstract}

\clearpage

\section{Introduction}

This paper considers the problem of constructing confidence intervals (CIs) for
a scalar parameter $T(f)$ of a function $f$, which can be a conditional mean or
a density. The scalar parameter may correspond, for example, to a conditional
mean, or its derivatives at a point, the regression discontinuity or the
regression kink parameter, or the value of a density or its derivatives at a
point. A popular approach to estimation of $T(f)$ is to use kernel or local
polynomial estimators. These estimators are both simple to implement, and highly
efficient in terms of their mean squared error (MSE) properties
\citep{fan93,cfm97}. CIs are typically formed by undersmoothing (choosing the
bandwidth to shrink more quickly than the MSE optimal bandwidth) or
bias-correction (subtracting an estimate of the estimator's bias).

In this paper, we propose a simple alternative approach to forming CIs based on
these estimators that is more efficient than both undersmoothing and
bias-correction in the sense that it leads to shorter CIs while maintaining
coverage over the same parameter space $\mathcal{F}$ for $f$ (which typically
places bounds on derivatives of $f$). In particular, one simply adds and
subtracts the estimator's standard error times a critical value that is larger
than the usual normal quantile $z_{1-\alpha/2}$, and takes into account the
possible bias of the estimator.\footnote{An R package implementing our CIs in
  regression discontinuity designs is available at
  \url{https://github.com/kolesarm/RDHonest}.} Asymptotically, these CIs
correspond to fixed-length CIs as defined in \citet{donoho94}, and so we refer
to them as fixed-length CIs. We show that the critical value depends only on (1)
the order of the derivative that one bounds to define the parameter space
$\mathcal{F}$; and (2) the criterion used to choose the bandwidth. In
particular, if the MSE optimal bandwidth is used with a local linear estimator,
computing our CI at the 95\% coverage level amounts to replacing the usual
critical value $z_{0.975}=1.96$ with 2.18.

When the criterion for bandwidth choice is the length of the resulting CI, we
show that the resulting bandwidth is in fact \emph{larger} than the MSE optimal
bandwidth. This contrasts with the work of \citet{hall_effect_1992} and
\citet{ccf15} on optimality of undersmoothing. Importantly, these papers
restrict attention to CIs that use the usual critical value $z_{1-\alpha/2}$. It
then becomes necessary to choose a small enough bandwidth so that the bias is
asymptotically negligible relative to the standard error, since this is the only
way to achieve correct coverage. Our results imply that rather than choosing a
smaller bandwidth, it is better to use a larger critical value that takes into
account the potential bias; this also ensures correct coverage regardless of the
bandwidth sequence. While the fixed-length CIs shrink at the optimal rate,
undersmoothed CIs shrink more slowly. We also show that under smoothness
assumptions needed to implement bias-correction, our CIs shrink at a faster rate
than bias-corrected CIs, once the standard error is adjusted to take into
account the variability of the bias estimate (\citet{cct14} show that doing so
is important for maintaining coverage). The oversmoothing relative to the MSE
optimal bandwidth is relatively modest: under a range of conditions most
commonly used in practice, a fixed-length CI centered at the MSE optimal
bandwidth is 99\% efficient relative to using the CI optimal bandwidth.
Therefore, a practically attractive implementation of our CIs is to simply
center them around an estimator with MSE optimal bandwidth, rather than
reoptimizing the bandwidth for length and coverage of the CI\@.

A key requirement that underlies our results is the notion of honesty: as in
\citet{li89honest}, we require that the CIs cover the true parameter
asymptotically at the nominal level uniformly over the parameter space
$\mathcal{F}$. Furthermore, we allow this parameter space to grow with the
sample size. The notion of honesty is closely related to the use of the minimax
criterion used to derive the MSE efficiency results: in both cases, one requires
good performance uniformly over the parameter space $\mathcal{F}$. The
requirement that the CIs be honest is necessary for good finite-sample
performance. In contrast, approaches to inference based on pointwise-in-$f$
asymptotics, such as using bandwidths that optimize the pointwise-in-$f$
asymptotic MSE can lead to arbitrarily poor finite-sample behavior, as we
discuss further in \Cref{sec:rmse-pointw-optim}. To illustrate the
practical importance of this point, we conduct a Monte Carlo study in which we
show that commonly used CIs based on plug-in bandwidths that attempt to estimate
this pointwise-in-$f$ optimal bandwidth exhibit severe undercoverage, even when
combined with undersmoothing or bias-correction.

When the parameter space places a bound $\SC$ on a derivative of $f$, our CIs
require this bound to be specified explicitly. While this may appear to be a
disadvantage of our particular approach, due to impossibility results of
\citet{low97}, \citet{CaLo04}, and \citet{ArKo18optimal}, this cannot be
avoided, regardless of how one forms the CI, without making further restrictions
on the function $f$. In particular, these papers show that, without additional
assumptions on the parameter space, one cannot use a data-driven method to
estimate $\SC$ and maintain coverage over the whole parameter space---any other
method that appears to avoid making this choice must do so \emph{implicitly}.
For example, an apparent advantage of undersmoothing is that it leads to correct
coverage for any fixed smoothness constant $\SC$. However, as we discuss in
detail in \Cref{sec:effic-comp}, a more accurate description of
undersmoothing is that for each sample size $n$, it implicitly chooses a
constant $\SC_n$ under which coverage is controlled. Given a sequence of
undersmoothed bandwidths, we show how $\SC_n$ can be calculated explicitly. One
can then obtain a shorter CI with the same coverage properties by computing a
fixed-length CI for the corresponding $\SC_n$. Regardless of how one chooses
$\SC$, the fixed-length CIs we propose are more efficient than undersmoothed or
bias-corrected CIs that use the same (implicit or explicit) choice of $\SC$. In
fact, it follows from the calculations in \citet{donoho94} and
\citet{ArKo18optimal} that our CIs, when constructed using a length-optimal or
MSE-optimal bandwidth, are highly efficient among \emph{all} honest CIs: no
other approach to inference can substantively improve on their length, while
still maintaining coverage.

As an alternative to choosing $\SC$ a priori, one can place additional
conditions on the function $f$ that allow for an upper bound on $\SC$ to be
estimated. To maintain efficiency of the resulting CI, however, care must be
taken in doing so: if $\SC$ is a bound on the $p$th derivative, and one imposes
a bound $\tilde{\SC}$ on the $(p+1)$th derivative in order to estimate $\SC$,
then the optimal CI will be based on a different estimator and will depend on
the new bound $\tilde\SC$. To avoid such issues, we propose a regularity class
that relates a global polynomial approximation to smoothness of the function $f$
near the point of interest, and we show formally that, for this class, one can
obtain a valid and highly efficient CI using a global polynomial rule of thumb
suggested by \citet{fg96}. However, given the additional assumptions required by
this (or any) data driven choice of $\SC$, we recommend that this approach be
used as a starting point for sensitivity analysis allowing for other choices of
$\SC$.

Another approach to data-driven choices of $\SC$ is to use ``self-similarity''
conditions, as suggested by \citet{GiNi10}, which relate the maximum and minimum
bias at different bandwidths. \citet{bull_honest_2012} and \citet{cck14as} have
obtained rate optimal confidence bands under such conditions, which, like the
CIs considered here, use a critical value based on an upper bound on the bias.
While these results for confidence bands could be extended to cover the problem
of constructing CIs for a scalar parameter, obtaining sharp critical values
appears to be very difficult. Indeed, the results of
\citet{armstrong_adaptation_2018} show that the sharp form of such CIs must
depend to first order on auxiliary constants used to define self-similarity.
Nonetheless, our approach of bounding local smoothness using a global polynomial
approximation is inspired by the self-similarity approach taken by this
literature, and we see it as being in the same spirit. \citet{schennach15} also
uses an upper bound on the bias based on an estimated smoothness constant. While
the coverage of the resulting CIs is pointwise-in-$f$, it is plausible that the
CIs are honest under additional auxiliary conditions, similar in spirit to
self-similarity.

In addition to calculating the relative efficiency of CIs constructed using
different bandwidths, our results allow us to calculate the relative efficiency
of CIs constructed using different kernels. In particular, we show that the
relative efficiency of kernels for the CIs we propose is \emph{the same} as the
relative efficiency of the estimates in terms of MSE\@. Thus, relative
efficiency calculations for MSE, such as the ones in \citet{fan93},
\citet{cfm97}, and \citet{fggbe97} for estimation of a nonparametric mean at a
point (estimation of $f(x_{0})$ for some $x_{0}$) that motivate much of
empirical practice in the applied regression discontinuity literature, translate
directly to CI construction. Despite their importance in motivating empirical
practice, however, such results are subject to a technical critique about how
the parameter space is specified: rather than placing a bound on a derivative of
$f$ (a Hölder condition), currently available relative efficiency results place
assumptions directly on the error of a Taylor approximation at a particular
point, so that some ``nonsmooth'' functions are in fact not ruled
out.\footnote{See \citet{ImWa17}, as well as our discussion in
  \Cref{sec:inference-point-theory} for an elaboration of this critique.} To address
this, we derive the minimax performance of local polynomial estimators under
Hölder restrictions on $f$. These results confirm that the local polynomial
estimators used in empirical practice are also highly efficient under Hölder
restrictions on $f$. Furthermore, while we focus on asymptotic CIs and relative
efficiency, these results include a derivation of the finite-sample worst-case
bias of local polynomial estimators under Hölder restrictions, which was used by
\citet{KoRo16} to form finite-sample valid CIs in a fixed-design regression
setting. These findings may be of independent interest.

The requirement of honesty is also important to ensure that our concept of
optimality is well-defined and consistent. As discussed above, it allows us to
consider bandwidth or kernel efficiency for constructing CIs. In addition, it
also allows us to formally show that using local polynomial regression of an
order that's too high given the amount of smoothness imposed is suboptimal. In
contrast, under pointwise-in-$f$ asymptotics, high-order local polynomial
estimates are superefficient at every point in the parameter space (see Chapter
1.2.4 in \citealp{tsybakov09}, and \citealp{BrLoZh97}).

To illustrate the implementation of the honest CIs, we reanalyze the data from
\citet{LuMi07}, who, using a regression discontinuity design, find a large and
significant effect of receiving technical assistance to apply for Head Start
funding on child mortality at a county level. However, this result is based on
CIs that ignore the possible bias of the local linear estimator around which
they are built, and an ad hoc bandwidth choice. We find that, if one bounds the
second derivative globally by a constant $\SC$ using a Hölder class, the
uncertainty associated with the effect size is much larger than originally
reported, unless one is very optimistic about the constant $\SC$, allowing $f$
to only be linear or nearly-linear.

Our results build on the literature on estimation of linear functionals in
normal models with convex parameter spaces, as developed by \citet{donoho94},
\citet{IbKh85} and many others. As with the results in that literature, our
setup gives asymptotic results for problems that are asymptotically equivalent
to the Gaussian white noise model, including nonparametric regression
\citep{BrLo96} and density estimation \citep{nussbaum_asymptotic_1996}. Our main
results build on the ``renormalization heuristics'' of \citet{DoLo92}, who show
that many nonparametric estimation problems have renormalization properties that
allow easy computation of minimax MSE optimal kernels and rates of convergence.
Our results hold under essentially the same conditions, which apply in many
classical nonparametric settings.

The CIs we consider in this paper are applications of the fixed-length CIs
proposed in the context of inference on linear functionals $T(f)$ in Gaussian
nonparametric regression by \citet{donoho94}, which have also been studied
recently in \citet{ArKo18optimal}, and in contemporaneous and subsequent work by
\citet{KoRo16} and \citet{ImWa17}. In contrast to the finite-sample approach
taken in these papers, we focus on asymptotic results, and we also allow $T(f)$
to be non-linear. Instead of imposing the nonparametric regression model, we
require a renormalization condition (see \Cref{performance_approx_eq} below)
that allows us to apply the ``renormalization heuristics'' of \citet{DoLo92}; we
are thus able to cover settings such as density estimation or estimation of a
bidder valuation in first-price auctions
(see \begin{NoHyper}\Cref{sec:addit-appl}\end{NoHyper}). Our asymptotic approach
allows for simplifications that deliver our main relative efficiency results.
These efficiency results are different from and complementary to the asymptotic
form of the efficiency bounds given in \citet{donoho94} and
\citet{ArKo18optimal}: whereas we consider relative efficiency of estimators and
fixed-length CIs based on different kernels and bandwidths, \citet{donoho94} and
\citet{ArKo18optimal} bound the scope for efficiency gains from CIs that do not
fall into this class. \citet{donoho94} and \citet{ArKo18optimal} find that the
scope for further improvement is small, which motivates our focus on this class
of estimators and CIs. See \Cref{remark:nonlinear_lower_bounds} for further
discussion.

The rest of this paper is organized as follows. \Cref{results_sec} gives the
main results. \Cref{sec:applications} applies our results to inference at a
point, sharp and fuzzy RD, and it discusses practical implementation issues,
including a rule of thumb for choosing $\SC$. \Cref{sec:comp-with-other} gives a
theoretical comparison of our fixed-length CIs to other approaches, and
\Cref{monte_carlo_sec} compares them in a Monte Carlo study. Finally,
\Cref{application_sec} presents an empirical application based on
\citet{LuMi07}. \Cref{proofs_sec} gives proofs of the results in
\Cref{results_sec}. Additional results are collected in Supplemental Appendices.

\section{General results}\label{results_sec}

We are interested in a scalar parameter $T(f)$ of a function $f$, which is
typically a conditional mean or a density. The function $f$ is assumed to lie in
a function class $\mathcal{F}=\mathcal{F}(\SC)$, which places ``smoothness''
conditions on $f$, where $\SC$ indexes the level of smoothness. We focus on
classical nonparametric function classes, in which $\SC$ corresponds to a bound
on a derivative of $f$ of a given order. We allow $\SC=\SC_{n}$ to grow with the
sample size $n$.

We have available a class of estimators $\hat T(h;k)$, indexed by a bandwidth
$h=h_{n}>0$ and a kernel $k$. Let $\hatse(h;k)$ denote the standard error of
$\hat{T}(h;k)$, an estimate of its standard deviation $\sd_f(\hat T(h;k))$. We
assume that a central limit theorem applies to $\hat{T}(h;k)$, so that in large
samples, the $t$-statistic $[\hat T(h;k)-T(f)]/\hatse(h;k)$ will be
approximately normal with variance 1 and mean given by the ratio of bias to
standard deviation, $t_{f}=(E_{f}[\hat{T}(h;k)-T(f)])/\sd_{f}(\hat{T}(h;k))$.
Since $t_{f}$ depends on the unknown function $f$, this ratio is unknown. Note,
however, that we can bound $\abs{t_{f}}$ by the worst-case ratio of bias to
standard deviation (bias-sd ratio),
$t_{\mathcal{F}}=\sup_{f\in\mathcal{F}}\abs{E_{f}[\hat{T}(h;k)-T(f)]}/\sd_{f}(\hat{T}(h;k))$.
Therefore, if this bias-sd ratio can be computed up to asymptotically negligible
terms, we can construct an honest CI as
\begin{equation}\label{eq:honest-two-sided-ci}
  \hat T(h;k)\pm \cv_{1-\alpha}(t)\cdot\hatse(h;k),
\end{equation}
where the approximate bias-sd ratio $t$ satisfies $t=t_{\mathcal{F}}(1+o(1))$,
and $\cv_{1-\alpha}(t)$ is the $1-\alpha$ quantile of the folded normal
distribution $\abs{N(t,1)}$, or, equivalently, the square root of the $1-\alpha$
quantile of a $\chi^{2}$ distribution with $1$ degree of freedom, and
non-centrality parameter $t^{2}$, which is readily available in statistical
software. For easy reference, we list these critical values in
\Cref{tab:cvs} for selected values of $t$. Because the quantiles of a
$\chi^{2}$ distribution are increasing in its non-centrality parameter,
replacing $t_{f}$ with an upper bound that is valid for all $f\in\mathcal{F}$
yields a CI that is honest over $\mathcal{F}$. The CI
in~\eqref{eq:honest-two-sided-ci} is an approximate version of a fixed-length
confidence interval (FLCI) studied in \citet{donoho94}, who replaces
$\hatse(h;k)$ with $\sd_f(\hat T(h;k))$ in the definition of this CI, and
assumes $\sd_f(\hat T(h;k))$ is constant over $f$, in which case its length will
be fixed. We thus refer to CIs of this form as ``fixed-length'', even though
$\hatse(h;k)$ is random.

To motivate our main regularity condition~\eqref{performance_approx_eq} below
that will facilitate studying the performance of these FLCIs and allow for an
easy computation of the bias-sd ratio $t$, suppose that the standard deviation
and the worst-case bias of the estimator $\hat{T}(h;k)$,
\begin{equation*}
  \maxbias(\hat{T}(h;k))=\sup_{f\in\mathcal{F}}\abs{E_f\hat{T}(h;k)-T(f)},
\end{equation*}
scale as powers of $h$. In particular, suppose that, for some $\gamma_b>0$,
$\gamma_s<0$, $B(k)>0$ and $S(k)>0$,
\begin{align}\label{bias_var_scale_eq}
  \maxbias(\hat T(h;k))&=h^{\gamma_b}\SC B(k)(1+o(1)), &
  \sd_{f}(\hat{T}(h;k))&=h^{\gamma_s}n^{-1/2}S(k)(1+o(1)),
\end{align}
where the $o(1)$ term in the second equality is uniform over $f\in\mathcal{F}$.
We show in
\begin{NoHyper}\Cref{verification_sec}\end{NoHyper} that
this condition will hold whenever the renormalization heuristics of
\citet{DoLo92} can be formalized. This includes most classical nonparametric
problems, such as estimation of a density or a conditional mean, or its
derivative, evaluated at a point (which may be a boundary point). In
\Cref{sec:inference-point-theory}, we show that~\eqref{bias_var_scale_eq} holds with
$\gamma_{b}=p$, and $\gamma_{s}=-1/2$ under mild regularity conditions when
$\hat{T}(h;k)$ is a local polynomial estimator of a conditional mean at a point,
and $\mathcal{F}(\SC)$ consists of functions with $p$th derivative bounded by
$\SC$.

\begin{remark}\label{remark:notation}
  The second condition in~\eqref{bias_var_scale_eq} implies that the standard
  deviation does not depend on the underlying function $f$ asymptotically. In
  certain settings, such as density estimation (see
  \begin{NoHyper}\Cref{density_sec}\end{NoHyper}), this may
  require choosing a localized sequence of parameter spaces $\mathcal{F}_n$,
  similar to local asymptotic minimax results in parametric settings
  \citep[e.g., Section 8.7 in][]{van_der_vaart_asymptotic_1998}. While we allow
  for such dependence, we keep any dependence of $\mathcal{F}$ on $n$ implicit
  in our notation in the main text. Similarly, the quantities $B(k)$ and $S(k)$
  generally depend on $\mathcal{F}$ (which if the parameter space is localized
  includes the localization point), as well as on other nuisance parameters,
  such as the variance of the regression errors. To prevent notational clutter,
  we keep this dependence implicit.
\end{remark}

Under~\eqref{bias_var_scale_eq}, we can use the ratio
$t=h^{\gamma_b-\gamma_s}\SC B(k)/(n^{-1/2}S(k))$ of the leading worst-case bias
and standard deviation terms to compute the critical value $\cv_{1-\alpha}(t)$
in~\eqref{eq:honest-two-sided-ci}. Analogously to the two-sided case, honest
one-sided $1-\alpha$ CIs based on $\hat{T}(h;k)$ can be constructed by
subtracting the standard error times a $1-\alpha$ quantile of the distribution
$\mathcal{N}(t,1)$. This is asymptotically equivalent to the CI
\begin{equation}\label{eq:oci}
  \hor{\hat{T}(h;k)-h^{\gamma_b}\SC B(k)-
    z_{1-\alpha}h^{\gamma_s}n^{-1/2}S(k)\;, \;\infty},
\end{equation}
which subtracts the maximum bias, in addition to subtracting $z_{1-\alpha}$
times the standard deviation, from $\hat{T}(h;k)$.

\begin{remark}\label{remark:two-sided-cis}
  One could also form honest two-sided CIs by simply adding and subtracting the
  worst case bias, in addition to adding and subtracting the standard error
  times $z_{1-\alpha/2}=\cv_{1-\alpha}(0)$, the $1-\alpha/2$ quantile of a
  standard normal distribution, forming the CI as
  $\hat{T}(h;k)\pm (h^{\gamma_b}\SC B(k)+z_{1-\alpha/2}\cdot\hatse(h;k))$.
  However, since the estimator $\hat{T}(h;k)$ cannot simultaneously have a large
  positive and a large negative bias, such CI will be conservative, and longer
  than the CI given in \Cref{eq:honest-two-sided-ci}.
\end{remark}

To discuss the optimal choice of bandwidth $h$ and compare efficiency of
different kernels $k$ in forming one- and two-sided CIs, and compare the results
to the bandwidth and kernel efficiency results for estimation, it will be useful
to introduce notation for a generic performance criterion. Let $R(\hat T)$
denote the worst-case (over $\mathcal{F}$) performance of $\hat{T}$ according to
a given criterion, and let $\tilde{R}(b, s)$ denote the value of this criterion
when $\hat T-T(f)\sim N(b,s^2)$. For FLCIs, we can take their half-length as the
criterion, which leads to
\begin{align*}
  R_{\FLCI, \alpha}(\hat T(h;k))& =\inf \big\{\chi :
    P_{f}(\abs{\hat{T}(h;k)-T(f)}\le \chi)\ge 1-\alpha\;\text{for all
      $f\in\mathcal{F}$} \big\}, \\
  \tilde{R}_{\FLCI, \alpha}(b, s) &= \inf \left\{\chi : P_{Z\sim
      N(0,1)}\left(\abs{s Z+b}\le \chi\right)\ge 1-\alpha \right\} =s\cdot
  \cv_{1-\alpha}(b/s).
\end{align*}
To evaluate one-sided CIs, one needs a criterion other than length, which is
infinite. A natural criterion is expected excess length, or quantiles of excess
length. We focus here on the quantiles of excess length. For CI of the
form~\eqref{eq:oci}, its worst-case $\beta$ quantile of excess length is given
by
$R_{\OCI, \alpha, \beta}(\hat{T}(h;k)) = \sup_{f\in\mathcal{F}}q_{f,
  \beta}(T(f)-\hat{T}(h;k)+h^{\gamma_b}\SC B(k)+
z_{1-\alpha}h^{\gamma_s}n^{-1/2}S(k))$, where $q_{f, \beta}(Z)$ is the $\beta$
quantile of a random variable $Z$. The worst-case $\beta$ quantile of excess
length based on an estimator $\hat{T}$ when $\hat T-T(f)$ is normal with
variance $s^2$ and bias ranging between $-b$ and $b$ is
$\tilde{R}_{\OCI, \alpha, \beta}(b, s)= 2b+(z_{1-\alpha}+z_\beta)s$. Finally, to
evaluate $\hat{T}(h;k)$ as an estimator we use the maximum root mean squared
error (RMSE) under $\mathcal{F}$ as the performance criterion:
\begin{align*}
  R_{\RMSE}(\hat{T})&=\sup_{f\in\mathcal{F}}\sqrt{E_f[\hat{T}-T(f)]^2},
  & \tilde{R}_{\RMSE}(b, s)&=\sqrt{b^2+s^2}.
\end{align*}

The key regularity condition that we impose on the class of estimators
$\hat{T}(h; k)$ is that their performance can be approximated in large samples
by the performance of a normally distributed estimator with bias and standard
deviation that scale as powers of $h$,
\begin{equation}\label{performance_approx_eq}
  R(\hat T(h;k))= \tilde{R}(h^{\gamma_b}\SC B(k), h^{\gamma_s}n^{-1/2}S(k))(1+o(1)).
\end{equation}
For the performance criteria above, if the estimator $\hat{T}(h;k)$ satisfies an
appropriate central limit theorem, and \Cref{bias_var_scale_eq} holds,
condition~\eqref{performance_approx_eq} will hold so long as the estimator is
centered, so that, up to asymptotically negligible terms, its maximum and
minimum bias over $\mathcal{F}$ sum to zero,
$\sup_{f\in\mathcal{F}}E_f(\hat{T}(h;k)-T(f))=-\inf_{f\in\mathcal{F}}
E_f(\hat{T}(h;k)-T(f))(1+o(1))$.\footnote{\label{fn:centering}This centering
  condition holds automatically by a symmetry argument for kernel or local
  polynomial estimators if $f$ is a conditional mean or a density, $T(f)$ is its
  value or its derivative at a point, or a regression discontinuity parameter,
  and $\mathcal{F}$ bounds its derivatives. In other cases,
  \Cref{performance_approx_eq} will hold when the estimator is recentered by
  subtracting
  $\mathfrak{B}=(\sup_{f\in\mathcal{F}}E_f(\hat{T}(h;k)-T(f))+\inf_{f\in\mathcal{F}}
  E_f(\hat{T}(h;k)-T(f)))/2$, or an estimate $\hat{\mathfrak{B}}$ of
  $\mathfrak{B}$ that is consistent in the sense that
  $(\hat{\mathfrak{B}}-\mathfrak{B})/\hatse(h;k)$ converges in probability to
  zero, uniformly over $\mathcal{F}$. Recentering the estimator in this way
  improves the estimator's performance under the criteria that we consider.}
Heuristically, this follows because if
$(\hat{T}(h;k)-E_{f}\hat{T}(h;k))/\sd_{f}(\hat{T})$ is asymptotically $N(0, 1)$,
then under~\eqref{bias_var_scale_eq}, $\hat{T}(h;k)-T(f)$ will be in large
samples approximately normal, with standard deviation
$h^{\gamma_{s}}n^{-1/2}S(k)$, and mean bounded above and below by
$h^{\gamma_{b}}\SC B(k)$. In \Cref{sec:inference-point-theory}, we
verify~\eqref{performance_approx_eq} for the problem of estimation of a
conditional mean at a point. For estimation of certain smooth non-linear
functionals of the regression function or non-parametric density, including
fuzzy regression discontinuity discussed in \Cref{sec:applications}, and
estimating a bidder valuation in first price auctions discussed in
\begin{NoHyper}\Cref{auctions_sec}\end{NoHyper}, moments of
the estimator may not exist. In these cases, one can use
\begin{NoHyper}\Cref{general_sufficient_condition_thm,delta_method_thm} in \Cref{verification_sec}\end{NoHyper} to verify
\eqref{performance_approx_eq}, which only require a weaker version
of~\eqref{bias_var_scale_eq} stated in terms of convergence in distribution
rather than moments, so long as one truncates unbounded loss
functions.\footnote{\label{fn:truncation}For evaluating estimators in these
  cases, we focus on minimizing the limit of the scaled truncated RMSE
  $\lim_{c\to\infty}\lim_{n\to\infty}n^{r/2}M^{r-1}R_{\ell_{c}}(\hat{T}(h;k))$,
  where $R_{\ell_{c}}$ denotes the worst-case risk under a version of the RMSE
  that truncates the squared error loss at $c^{2}$. This is equivalent to
  minimizing the (untruncated) asymptotic RMSE (see
  \begin{NoHyper}\Cref{sufficient_conditions_sec}\end{NoHyper}
  for details). Under this criterion, the RMSE optimal bandwidth defined below
  and \Cref{two_R_thm} below are not affected by the truncation.}

We also assume that $\tilde{R}$ is homogeneous of degree one,
\begin{equation}\label{eq:homo-1}
      \tilde{R}(tb, ts)=t \tilde{R}(b, s)\quad\text{for all $t>0$.}
\end{equation}
This condition holds for all three criteria considered above. This allows us to
simplify the right-hand side of~\eqref{performance_approx_eq}. In particular,
using the bias-sd ratio $t=h^{\gamma_b-\gamma_s}\SC B(k)/(n^{-1/2}S(k))$, write
the bandwidth as
$h=\left(tn^{-1/2}S(k)/(\SC B(k))\right)^{1/(\gamma_b-\gamma_s)}$. Substituting
this expression in~\eqref{performance_approx_eq} and using~\eqref{eq:homo-1}
gives
\begin{equation}\label{Rt_eq}
  \begin{split}
    R(\hat{T}(h;k))& = \tilde{R}(t^r n^{-r/2}\SC^{1-r}S(k)^{r}B(k)^{1-r},
    t^{r-1}n^{-r/2}\SC^{1-r}S(k)^{r} B(k)^{1-r})(1+o(1))
    \\
    &= n^{-r/2}\SC^{1-r} S(k)^{r}B(k)^{1-r} t^{r-1} \tilde{R}(t,1)(1+o(1)),
  \end{split}
\end{equation}
where $r=\gamma_b/(\gamma_b-\gamma_s)$. Since the performance criterion
converges at the rate $n^{r/2}$ when $\SC$ is fixed, we refer to $r$ as the rate
exponent (this matches the definition in, e.g., \citealt{DoLo92}). We denote the
bandwidth choice that minimizes the right-hand side of~\eqref{Rt_eq} for a given
performance criterion $R$ by
$h^*_R=(n^{-1/2}S(k)t^*_R/(\SC B(k)))^{1/(\gamma_b-\gamma_s)}$, with
$t^*_R=\argmin_{t} t^{r-1}\tilde{R}(t,1)$, and assume that $t^{*}_{R}$ is finite
and strictly greater than zero, which is the case for the performance criteria
we consider.

The bandwidth choice $h^{*}_{R}$ will be asymptotically optimal so long as it is
suboptimal to choose a bandwidth sequence $h_{n}$ such that such that the bias
or the variance dominates asymptotically, which is the case in the settings
considered here. For our main results, we assume this directly by assuming that
\begin{equation}\label{large_small_h_eq}
  \SC^{r-1}n^{\frac{r}{2}}R(\hat T(h_n;k))\to\infty\;
  \text{for any $h_n$ with}\;
  h_n(n\SC^2)^{\frac{1}{2(\gamma_b-\gamma_s)}}\to\infty\;\text{or}\;
  h_n(n\SC^2)^{\frac{1}{2(\gamma_b-\gamma_s)}}\to 0.
\end{equation}
Under this condition, we only need~\eqref{performance_approx_eq} to hold for
bandwidth sequences that are of the same order
$(n\SC^2)^{-1/[2(\gamma_b-\gamma_s)]}$ as the optimal bandwidth
$h^{*}_{R}$.\footnote{In typical settings, a necessary condition for
  \Cref{performance_approx_eq} to hold is that the optimal bandwidth $h^*_R$
  shrinks at a rate such that $(h^{*}_R)^{-2\gamma_{s}} n\to\infty$ and
  $h^*_R\to 0$. If $\SC$ is fixed, this simply requires that
  $\gamma_b-\gamma_s>1/2$, which basically amounts to a requirement that
  $\mathcal{F}(\SC)$ imposes enough smoothness so that the problem is not
  degenerate in large samples. If $\SC=\SC_{n}\to \infty$, then the condition
  also requires $n^{r/2}\SC^{r-1}\to \infty$, so that $\SC$ does not increase
  too quickly.} Note that optimal bandwidth is of the same order regardless of
the performance criterion---the performance criterion only determines the
optimal bandwidth constant through $t^{*}_{R}$.

The next theorem collects implications of these derivations for the performance
of different kernels. In particular, we consider minimax performance over
bandwidth sequences, that is, bandwidth sequences $h_{n}$ that achieve the
asymptotically best possible worst-case performance in large samples in the
sense that $\SC^{r-1}n^{r/2}(R(\hat T(h_{n};k))-\inf_{h>0}R(\hat T(h;k)))=o(1)$.

\begin{theorem}\label{single_R_thm}
  Let $R$ be a performance criterion with $\tilde R(b, s)>0$ for all
  $(b, s)\ne 0$. Suppose that \Cref{performance_approx_eq} holds for
  any bandwidth sequence $h_{n}$ with
  $\liminf_{n\to\infty} h_{n}(n\SC^2)^{1/[2(\gamma_b-\gamma_s)]}>0$ and
  $\limsup_{n\to\infty} h_{n}(n\SC^2)^{1/[2(\gamma_b-\gamma_s)]}<\infty$, and
  suppose that \Cref{eq:homo-1,large_small_h_eq} hold.
  Define $h^*_R$ and $t^*_R$ as above, and assume that $t^*_R>0$ is unique and
  well-defined. Then:
  \begin{enumerate}[label= ({\roman*})]
  \item\label{item:theorem-1} The asymptotic minimax performance under the
    kernel $k$ is given by
    \begin{equation*}
      \begin{split}
        \SC^{r-1}n^{r/2}\inf_{h>0}R(\hat T(h;k))
        &=\SC^{r-1}n^{r/2}R(\hat T(h^*_R;k))+o(1) \\
        &=S(k)^{r}B(k)^{1-r}(t^{*}_{R})^{r-1}\tilde R(t^{*}_{R},1)+o(1).
      \end{split}
    \end{equation*}
  \item\label{item:theorem-2} The asymptotic relative efficiency of two kernels
    $k_1$ and $k_2$ is given by
    \begin{equation*}
      \lim_{n\to\infty}\frac{\inf_{h>0}R(\hat{T}(h;k_1))}{\inf_{h>0}R(\hat{T}(h;k_2))}
      =\frac{S(k_1)^{r}B(k_1)^{1-r}}{S(k_2)^{r}B(k_2)^{1-r}}.
    \end{equation*}
    It depends on the rate $r$ but not on the performance criterion $R$.
  \item\label{item:theorem-3} If
    we consider two performance criteria $R_1$ and $R_2$ satisfying the
    conditions above, then the limit of the ratio of optimal bandwidths for
    these criteria is
    \begin{equation*}
      \lim_{n\to\infty} \frac{h^*_{R_1}}{h^*_{R_2}}
      =\left(\frac{t^*_{R_1}}{t^*_{R_2}}\right)^{1/(\gamma_b-\gamma_s)}.
    \end{equation*}
    It depends only on $\gamma_b$ and $\gamma_s$ and the performance criteria.
    If~\eqref{bias_var_scale_eq} holds, the asymptotically optimal bias-sd ratio is
    given by
    \begin{equation*}
      \lim_{n\to\infty}
      \frac{\maxbias(\hat{T}(h^*_R;k))}{\sd_f(\hat{T}(h^*_R;k))} = \argmin_{t}
      t^{r-1}\tilde R(t,1)=t^*_R.
    \end{equation*}
    It depends only on the performance criterion $R$ and rate exponent $r$.
  \end{enumerate}
\end{theorem}
Part~\ref{item:theorem-1} gives the optimal bandwidth formula for a given
performance criterion. The performance criterion only determines the optimal
bandwidth constant (the optimal bias-sd ratio) $t^{*}_{R}$.

Part~\ref{item:theorem-2} shows that relative kernel efficiency results do not
depend on the performance criterion. In particular, known kernel efficiency
results under the RMSE criterion such as those in \citet{fan93}, \citet{cfm97},
and \citet{fggbe97} apply unchanged to other performance criteria such as length
of FLCIs, excess length of one-sided CIs, or expected absolute error.

Part~\ref{item:theorem-3} shows that the optimal bias-sd ratio for a given
performance criterion depends on $\mathcal{F}$ only through the rate exponent
$r$, and does not depend on the kernel. The optimal bias-sd ratio for RMSE, FLCI
and OCI, respectively, are
{\allowdisplaybreaks
\begin{align*}
  t^*_{\RMSE}&=\argmin_{t>0}t^{r-1}\tilde{R}_{\RMSE}(t,1)
              =\argmin_{t>0}t^{r-1}\sqrt{t^2+1}=\sqrt{1/r-1}, \\
  t^*_{\FLCI}&=\argmin_{t>0}t^{r-1}\tilde{R}_{\FLCI, \alpha}(t,1)
              =\argmin_{t>0}t^{r-1}\cv_{1-\alpha}(t), \qquad\text{and}\\
  t^*_{\OCI}&=\argmin_{t>0}t^{r-1}\tilde{R}_{\OCI, \alpha, \beta}(t,1)
             =\argmin_{t>0}t^{r-1}[2t+(z_{1-\alpha}+z_\beta)]
             =(1/r-1)\frac{z_{1-\alpha}+z_\beta}{2}.
\end{align*}}
\Cref{fig:bias-variance-flci,fig:bias-variance-oci} plot these quantities as a
function of $r$. Note that the optimal bias-sd ratio is larger for FLCIs (at
levels $\alpha=.05$ and $\alpha=.01$) than for RMSE\@. Since $h$ is increasing
in $t$, it follows that, for FLCI, the optimal bandwidth \emph{oversmooths}
relative to the RMSE optimal bandwidth.

\begin{remark}\label{remark:nonlinear_lower_bounds}
  \Cref{single_R_thm} does not address whether further efficiency
  improvements are possible by using estimators that do not fall into the class
  $\hat T(h;k)$, or by using variable length CIs. However, it follows from
  \citet{donoho94} and \citet{ArKo18optimal} that, in typical settings where our
  results hold, little further improvement is possible. In particular, these
  papers give efficiency bounds that, applied to our setting, yield asymptotic
  lower bounds for $R(\hat T^*)/R(\hat T(h^*;k^*))$, where $\hat T^*$ is the
  optimal estimator or CI among all procedures (for CIs, this includes variable
  length CIs, with performance measured in terms of expected length), and $h^*$
  and $k^*$ are the optimal bandwidth and kernel. These asymptotic lower bounds
  depend only on the rate exponent $r$, and so can be used along with the bounds
  in \Cref{single_R_thm} to obtain the efficiency of a particular kernel
  and bandwidth relative to the fully optimal procedure.
\end{remark}

One can also form FLCIs centered at the estimator that is optimal for different
performance criterion $R$ as
$\hat T(h^*_{R};k)\pm \hatse(h^*_{R};k)\cdot\cv_{1-\alpha}(t^*_{R})$. The
critical value $\cv_{1-\alpha}(t^*_{R})$ depends only on the rate exponent $r$
and the performance criterion $R$. In particular, the CI centered at the RMSE
optimal estimator takes this form with $t^*_{\RMSE}=\sqrt{1/r-1}$, which yields
the CI
\begin{equation}\label{eq:honest-two-sided-ci-rmse}
  \hat T(h^{*}_{\RMSE};k)\pm \cv_{1-\alpha}(\sqrt{1/r-1})\cdot\hatse(h^{*}_{\RMSE};k),
\end{equation}

\Cref{tab:cvs} reports this critical value $\cv_{1-\alpha}(\sqrt{1/r-1})$
for rate exponents $r$ commonly encountered in practice. By~\eqref{Rt_eq}, the
resulting CI is wider than the one computed using the FLCI optimal bandwidth by
a factor of
\begin{equation}\label{mse_bandwidth_ci_inefficiency_eq}
  \frac{(t^*_{\FLCI})^{r-1}\cdot \cv_{1-\alpha}(t^*_{\FLCI})}
     {(t^*_{\RMSE})^{r-1}\cdot \cv_{1-\alpha}(t^*_{\RMSE})}.
\end{equation}
\Cref{fig:eff-mse-flci} plots this quantity as a function of $r$. It can
be seen from the figure that if $r\geq 4/5$, CIs constructed around the RMSE
optimal bandwidth are highly efficient. For example, if $r=4/5$, to construct an
honest 95\% FLCI based on an estimator with bandwidth chosen to optimize RMSE,
one simply adds and subtracts the standard error multiplied by 2.18 (rather than
the usual 1.96 critical value), and the corresponding CI is less than 1\% longer
than the one with bandwidth chosen to optimize CI length. The next theorem gives
a formal statement.

\begin{theorem}\label{two_R_thm}
  Suppose that the conditions of \Cref{single_R_thm} hold for
  $R_{\RMSE}$ and for $R_{\FLCI, \tilde\alpha}$ for all
  $\tilde\alpha$ in a neighborhood of $\alpha$. Let
  $\hatse(\hrmse;k)$ be such that
  $\hatse(\hrmse;k)/[(\hrmse)^{\gamma_s}n^{-1/2}S(k)]$ converges in
  probability to 1 uniformly over $f\in\mathcal{F}$. Then
  \begin{equation*}
    \lim_{n\to\infty}\inf_{f\in\mathcal{F}} P_f\left(T(f)\in
    \left\{\hat{T}(\hrmse;k)\pm \hatse(\hrmse;k)\cdot
        \cv_{1-\alpha}(\sqrt{1/r-1})\right\}\right)=1-\alpha.
  \end{equation*}
  The asymptotic efficiency of this CI relative to the one centered at the FLCI
  optimal bandwidth, defined as $\lim_{n\to\infty} \frac{\inf_{h>0}
    R_{\FLCI, \alpha}(\hat{T}(h;k))}{R_{\FLCI, \alpha}
    (\hat{T}(\hrmse;k))}$, is given
  by~\eqref{mse_bandwidth_ci_inefficiency_eq}. It depends only on $r$.
\end{theorem}

\section{Applications}\label{sec:applications}
In this section, we apply the general results from \Cref{results_sec} to the
problem of inference about a nonparametric regression function at a point, and
to regression discontinuity (RD). Readers who are interested only in implementing our
CIs in these applications can skip
\Cref{sec:theoretical-results}. \begin{NoHyper}\Cref{sec:addit-appl}\end{NoHyper}
discusses two additional applications: estimation of a density at a point, and
estimation of a bidder valuation in first-price auctions.

\subsection{Setup and Estimators}\label{sec:setup-estimators}

\paragraph{Inference at a point}
We are interested in inference about a nonparametric regression function $f$ at
a point, which we normalize to be zero, so that the parameter of interest is
given by $T(f)=f(0)$. We write the nonparametric regression model as
\begin{equation}\label{eq:np-regression}
  y_{i}=f(x_{i})+u_{i}, \quad i=1,\dotsc, n,\qquad E u_{i}=0,\quad \var(u_{i})=\sigma(x_{i}).
\end{equation}
where the design points $x_{i}$ are non-random. We allow the point of interest
$0$ to lie on the boundary of the support of the design points. We focus on
estimating $f(0)$ using a local polynomial estimator of order $\po$ with kernel
$k(\cdot)$,
\begin{equation*}
  \hat{T}_{\po}(h;k)=\sum_{i=1}^{n}w_{\po}^{n}(x_{i}; h, k)y_{i},
\end{equation*}
where the weights $w_{\po}^{n}(x_{i}; h, k)$ are given by
\begin{equation}\label{eq:lp-weights}
  w_{\po}^{n}(x; h, k)=e_{1}'Q_{n}^{-1} m_{\po}(x)k(x/h), \qquad
  Q_{n}=\sum_{i=1}^{n}k(x_{i}/h)m_{\po}(x_{i})m_{\po}(x_{i})'.
\end{equation}
Here $m_{\po}(t)=(1,t, \dotsc, t^{\po})'$, $e_{1}$ is a vector of zeros with
1 in the first position, and $h$ is a bandwidth. Thus, $\hat{T}_{\po}(h;k)$
corresponds to the intercept in a weighted least squares regression of $y_{i}$
on $(1,x_{i}, \dotsc, x_{i}^{\po})$ with weights $k(x_{i}/h)$. Local linear
estimators correspond to $q=1$, and Nadaraya-Watson (local constant) estimators
to $q=0$.

\paragraph{Sharp RD} In a sharp RD design, using data from the nonparametric
regression model~\eqref{eq:np-regression}, the goal is to to estimate the jump
in the regression function $f$ at a known cutoff, which we normalize to $0$, so
that $T(f)=\lim_{x\downarrow 0}f(x)-\lim_{x\uparrow 0}f(x)$. The cutoff
determines participation in a binary treatment: units with $x_{i}\geq 0$ are
treated; units with $x_{i}<0$ are controls. If the regression functions of
potential outcomes are continuous at zero, then $T(f)$ measures the average
effect of the treatment for units with $x_{i}=0$ \citep{htv01}. For brevity, we
focus on estimating $T(f)$ based only on local linear regressions: the estimator
$\hat{T}(h;k)$ is given by a difference between estimates from two local linear
regressions with bandwidth $h$ and kernel $k$ at a boundary point, one for units
with non-negative values running variable $x_{i}$, and one for units with
negative values of the running variable. The estimator can be written as
\begin{equation}\label{eq:sharp-rd-estimate}
  \hat{T}(h;k)=\sum_{i=1}^{n}(w_{+}^{n}(x;h, k)-w_{-}^{n}(x;h, k))y_{i},
\end{equation}
with the weight $w_{+}^{n}$ given by
\begin{equation*}
  w^{n}_{+}(x;h, k)=e_{1}'Q_{n, +}^{-1}m_{1}(x)k_{+}(x/h)
 ,\quad k_{+}(u)=k(u)\1{u\geq 0},
\end{equation*}
and $Q_{n, +}=\sum_{i=1}^{n}k_{+}(x_{i}/h)m_{1}(x_{i}) m_{1}(x_{i})'$. The
weights $w_{-}^{n}$, Gram matrix $Q_{n,-}$ and kernel $k_{-}$ are defined
similarly. Let $\sigma^{2}_{+}(x)=\sigma^{2}(x)\1{x\geq 0}$, and
$\sigma^{2}_{-}(x)=\sigma^{2}(x)\1{x <0}$.

\paragraph{Fuzzy RD}
In a fuzzy RD design, the treatment $d_{i}$ is not entirely determined by
whether the running variable $x_{i}$ exceeds a cutoff. Instead, the cutoff
induces a jump in the treatment probability. This fits into our framework if we
let $f=(f_{1}, f_{2})$ comprise two regression functions, corresponding to the
reduced-form regression of the outcome on the running variable, and the
first-stage regression of the treatment on the running variable:
\begin{equation}\label{eq:frd-regression}
  \begin{matrix}
    y_{i}&=f_{1}(x_{i})+u_{i1},\\
    d_{i}&=f_{2}(x_{i})+u_{i2},\\
  \end{matrix}\qquad i=1,\dotsc,n,\qquad E u_{i}=0,\quad \var(u_{i})=\Omega(x_{i}),
\end{equation}
with $u_{i}=(u_{i1},u_{i2})'$. The parameter of interest is given by the ratio
$T(f)=L_{1}(f)/L_{2}(f)$ of sharp RD parameters
$L_{j}(f)=\lim_{x\downarrow 0}f_{j}(x)-\lim_{x\uparrow 0}f_{j}(x)$ in the
reduced-form ($j=1$) and first-stage regression ($j=2$). If the regression
functions of the potential outcomes and potential treatments are continuous at
zero, and a monotonicity condition holds, then $T(f)$ measures the average
treatment effect for individuals with $x_{i}=0$ who are compliers
\citep[see][]{htv01}. We consider estimating $T(f)$ by its sample analog,
replacing $L_{1}$ and $L_{2}$ with sharp RD local linear estimates, which are
for simplicity assumed to be based on the same bandwidth,
$\hat{T}(h;k)=\hat{L}_{1}(h;k)/\hat{L}_{2}(h;k)$, where
\begin{equation*}
  \hat{L}(h;k)=
\begin{pmatrix}\hat{L}_{1}(h;k)\\\hat{L}_{1}(h;k)\end{pmatrix}
=\sum_{i}(w_{+}^{n}(x;h, k)-w_{-}^{n}(x;h, k))
\begin{pmatrix}y_{i}\\d_{i}
\end{pmatrix},
\end{equation*}
with the weights $w_{+}^{n}$ and $w_{-}^{n}$ defined as
in~\eqref{eq:sharp-rd-estimate}.

\subsection{Theoretical results}\label{sec:theoretical-results}
We now discuss the conditions under which the key regularity
condition~\eqref{performance_approx_eq} holds in each application. We also
discuss kernel efficiency results, and gains from imposing global, rather than
just local, smoothness on $f$.

\subsubsection{Inference at a point}\label{sec:inference-point-theory}
To state the results, it will be convenient to define the equivalent kernel
\begin{equation}\label{equivalent_kernel_eq}
  k^{*}_{\po}(u)=
  e_{1}'\left(\int_{\mathcal{X}}m_{\po}(t)m_{\po}(t)'k(t)\,
    \dd t\right)^{-1}m_{\po}(u)k(u),
\end{equation}
where the integral is over $\mathcal{X}=\mathbb{R}$ if $0$ is an interior point,
and over $\mathcal{X}=\hor{0,\infty}$ if $0$ is a (left) boundary point.

We assume the following conditions on the design points and regression errors
$u_{i}$:
\begin{assumption}\label{x_assump}
  For some $d>0$, the sequence $\{x_i\}_{i=1}^{n}$ satisfies
  $\frac{1}{nh_n}\sum_{i=1}^n g(x_i/h_n)\to d\cdot\int_{\mathcal{X}} g(u)\, du$ for
  any bounded function $g$ with finite support and any sequence $h_n$ with
  $0<\liminf_n h_n (n\SC^2)^{1/(2p+1)}<\limsup_n h_n
  (n\SC^2)^{1/(2p+1)}<\infty$.
\end{assumption}

\begin{assumption}\label{sigma_assump}
  The random variables $\{u_i\}_{i=1}^n$ are independent with $E u_{i}=0$,
  $E u_i^{2+\eta}\le 1/\eta$ for some $\eta>0$, and $\var(u_{i})=\sigma^2(x_i)$
  for some variance function $\sigma^2(x)$ that is continuous at $x=0$ with
  $\sigma^2(0)>0$.
\end{assumption}
\Cref{x_assump} requires that the empirical distribution of the design points is
smooth around $0$. When the support points are treated as random, the constant
$d$ typically corresponds to their density at $0$.

Because the estimator is linear in $y_{i}$, its variance doesn't depend on $f$,
\begin{equation}
  \label{eq:sd-rescaling-lp}
  \sd(\hat{T}_{\po}(h;k))^{2}
  =\sum_{i=1}^{n}w_{\po}^{n}(x_{i})^{2}\sigma^{2}(x_{i})
  =\frac{S(k)^{2}}{nh} (1+o(1)),
  \quad S(k)=\sqrt{\frac{\sigma^{2}(0)
      {\int_{\mathcal{X}}k^{*}_{\po}(u)^{2}\, \dd u}}{
      d}},
\end{equation}
where the second equality holds under \Cref{x_assump,sigma_assump}, as we show
in
\begin{NoHyper}\Cref{sec:local-polyn-estim}\end{NoHyper}.
The condition on the standard deviation in \Cref{bias_var_scale_eq}
thus holds with $\gamma_{s}=-1/2$, and $S(k)$ given in the preceding display.
\begin{NoHyper}\Cref{sec:kernel-constants}\end{NoHyper}
gives the constant $\int_{\mathcal{X}}k^{*}_{\po}(u)^{2}\, \dd u$ for selected
kernels.

On the other hand, the worst-case bias will be driven primarily by the function
class $\mathcal{F}$. We consider inference under two popular function classes.
First, the Taylor class of order $p$,
\begin{equation*}
  \FSY{p}(\SC)
  =\left\{f\colon \abs*{
      f(x)-\textstyle\sum_{j=0}^{p-1}f^{(j)}(0)x^{j}/j!}
    \le \SC \abs{x}^{p} /p! \;\; x\in\mathcal{X}\right\}.
\end{equation*}
This class consists of all functions for which the approximation error from a
($p-1$)-th order Taylor approximation around $0$ can be bounded by
$\frac{1}{p!}\SC \abs{x}^{p}$. It formalizes the idea that the $p$th derivative
of $f$ at zero should be bounded by some constant $\SC $. Using this class of
functions to derive optimal estimators goes back at least to
\citet{legostaeva_minimax_1971}, and it underlies much of existing minimax
theory concerning local polynomial estimators \citep[see][Chapter
3.4--3.5]{fg96}.

While analytically convenient, the Taylor class may not be attractive in some
empirical settings because it allows $f$ to be non-smooth and discontinuous away
from $0$. We therefore also consider inference under Hölder classes (for
  simplicity, we focus on Hölder classes of integer order)
\begin{equation*}
  \FHol{p}(\SC)
  =\left\{f\colon \abs{f^{(p-1)}(x)-f^{(p-1)}(x')}\leq \SC \abs{x-x'},\;
    x, x'\in\mathcal{X}\right\}.
\end{equation*}
This class is the closure of the family of $p$ times differentiable functions
with the $p$th derivative bounded by $\SC$, uniformly over $\mathcal{X}$, not
just at $0$. It formalizes the intuitive notion that $f$ should be $p$-times
differentiable with a bound on the $p$th derivative. The case $p=1$ corresponds
to the Lipschitz class of functions.

\begin{theorem}\label{theorem:maximum-bias-lp}
  Suppose that \Cref{x_assump} holds and that $k(\cdot)$ is bounded with bounded
  support and $\po\geq p-1$. Then, for any
  bandwidth sequence $h_{n}$ with $n h_n\to\infty$ and
  $0<\liminf_n h_n (n\SC^2)^{1/(2p+1)}<\limsup_n h_n
  (n\SC^2)^{1/(2p+1)}<\infty$,
  \begin{equation*}
        \maxbias_{\FSY{p}(\SC)}(\hat{T}_{\po}(h_{n};k)) =
    \frac{\SC h_{n}^p}{p!}\mathcal{B}_{p, \po}^{\text{T}}(k)(1+o(1)), \qquad
    \mathcal{B}_{p, \po}^{\text{T}}(k)=\int_{\mathcal{X}} \abs{u^p k^*_q(u)}\,
    du
  \end{equation*}
  and
  \begin{multline*}
    \maxbias_{\FHol{p}(\SC)}(\hat{T}_q(h_{n};k)) =
    \frac{\SC h_{n}^p}{p!}\mathcal{B}^{\textnormal{Höl}}_{p, \po}(k)(1+o(1)),\\
    \mathcal{B}^{\textnormal{Höl}}_{p, \po}(k) = p \int_{t=0}^\infty
    \abs*{\int_{u\in\mathcal{X}, \abs{u}\ge t} k^*_q(u)(\abs{u}-t)^{p-1}\, du}\,
    dt.
  \end{multline*}
  Thus, the first part of \Cref{bias_var_scale_eq} holds with
  $\gamma_b=p$ and $B(k)= \mathcal{B}_{p, \po}(k)/{p!}$, where
  $\mathcal{B}_{p, \po}(k)=\mathcal{B}^{\textnormal{Höl}}_{p, \po}(k)$ for
  $\FHol{p}(\SC)$, and
  $\mathcal{B}_{p, \po}(k)=\mathcal{B}^{\textnormal{T}}_{p, \po}(k)$ for
  $\FSY{p}(\SC)$.

  If, in addition, \Cref{sigma_assump} holds, then \Cref{performance_approx_eq}
  holds for the RMSE, FLCI and OCI performance criteria, with $\gamma_b$ and
  $B(k)$ given above and $\gamma_s$ and $S(k)$ given in
  \Cref{eq:sd-rescaling-lp}.
\end{theorem}
The theorem verifies the regularity conditions needed for the results in
\Cref{results_sec}, and implies that $r=2p/(2p+1)$ for $\FSY{p}(\SC)$ and
$\FHol{p}(\SC)$. If $p=2$, then we obtain $r=4/5$. By
\Cref{single_R_thm}\ref{item:theorem-1}, the optimal rate of convergence
of a criterion $R$ is $R(\hat{T}(h^{*}_{R};k))=O((n/\SC^{1/p})^{-p/(2p+1)})$. As
we will see from the relative efficiency calculation below, the optimal order of
the local polynomial regression is $\po=p-1$ for the kernels considered here.
The theorem allows $\po\geq p-1$, so that we can examine the efficiency of local
polynomial regressions that are of order that's too high relative to the
smoothness class. Allowing for $\po<p-1$ is not meaningful, as in this case, the
maximum bias is infinite.\footnote{\label{fn:not-nested} The smoothness classes
  $\FSY{p}(\SC)$ and $\FHol{p}(\SC)$ do not restrict derivatives of order $p-1$
  and lower, so that, in order to achieve a finite worst-case bias, the
  estimator needs to be unbiased for polynomials of order $p-1$, which requires
  $\po\geq p-1$.}

Under the Taylor class $\FSY{p}(\SC)$, the least favorable (bias-maximizing)
function is given by $f(x)=\SC/p!\cdot\sign(w^{n}_{\po}(x))\abs{x}^{p}$. In
particular, if the weights are not all positive, it will be discontinuous away
from the boundary. The first part of \Cref{theorem:maximum-bias-lp} then follows
by taking the limit of the bias under this function. \Cref{x_assump} ensures
that this limit is well-defined. Under the Hölder class $\FHol{p}(\SC)$, the
least favorable function takes the form of a $p$th order spline. See
\begin{NoHyper}\Cref{sec:local-polyn-estim}\end{NoHyper}
for details.

These results imply that given a kernel $k$ and order of a local polynomial
$\po$, the RMSE-optimal bandwidth for $\FSY{p}(\SC)$ and $\FHol{p}(\SC)$ is
given by
\begin{equation}\label{minimax_rmse_optimal_bw_eq}
  \hrmse=
  \left(\frac{1}{2pn}\frac{S(k)^{2}}{\SC^2 B(k)^{2}}\right)^{\frac{1}{2p+1}}
  =\left(\frac{\sigma^{2}(0)p!^{2}}{2pn
      d\SC ^{2}}\frac{\int_{\mathcal{X}}k^{*}_{\po}(u)^{2}\, \dd u}{
      \mathcal{B}_{p, \po}(k)^{2}}\right)^{\frac{1}{2p+1}},
\end{equation}
where $\mathcal{B}_{p, \po}(k)=\mathcal{B}^{\textnormal{Höl}}_{p, \po}(k)$ for
$\FHol{p}(\SC)$, and
$\mathcal{B}_{p, \po}(k)=\mathcal{B}^{\textnormal{T}}_{p, \po}(k)$ for
$\FSY{p}(\SC)$. For kernels given by polynomial functions over their support,
$k^{*}_{\po}$ also has the form of a polynomial, and
$\mathcal{B}^{\textnormal{T}}_{p, \po}$ and
$\mathcal{B}^{\textnormal{Höl}}_{p, \po}$ can be computed analytically.
\begin{NoHyper}\Cref{sec:kernel-constants}\end{NoHyper} gives these constants for
selected kernels.

\paragraph{Kernel efficiency}%

It follows from \Cref{single_R_thm}\ref{item:theorem-2} that the optimal
equivalent kernel minimizes $S(k)^{r}B(k)^{1-r}$, independently of the
performance criterion. Under the Taylor class $\FSY{p}(\SC)$, this is equivalent
to minimizing
\begin{equation}\label{eq:minimax-problem}
  \Big(\int_{\mathcal{X}}k^{*}(u)^{2}\, \dd u \Big)^{p}
  \cdot \int_{\mathcal{X}}\abs{u^{p}k^{*}(u)}\, \dd u,
\end{equation}
The solution to this problem follows from \citet[Theorem 1]{SaYl78} (see also
\citet{cfm97}). We give details of the solution in
\begin{NoHyper}\Cref{sec:optim-kern-details}\end{NoHyper}.
\Cref{tab:taylor-eff} compares the asymptotic relative efficiency of local
polynomial estimators based on the uniform, triangular, and Epanechnikov kernels
to the optimal Sacks-Ylvisaker kernels. \citet{fggbe97} and \citet{cfm97},
conjecture that minimizing~\eqref{eq:minimax-problem} yields a sharp bound on
kernel efficiency. It follows from
\Cref{single_R_thm}\ref{item:theorem-2} that this conjecture is correct,
and \Cref{tab:taylor-eff} matches the kernel efficiency bounds in these
papers. \Cref{tab:taylor-eff} shows that the choice of the kernel doesn't
matter very much, so long as the local polynomial is of the right order.
However, if the order is too high, $\po>p-1$, the efficiency can be quite low,
even if the bandwidth used was optimal for the function class or the right
order, $\FSY{p}(\SC)$, especially on the boundary. If the bandwidth picked is
optimal for $\FSY{\po-1}(\SC)$, it will shrink at a lower rate than optimal
under $\FSY{p}(\SC)$, and the resulting rate of convergence will be lower than
$r$. Consequently, the relative asymptotic efficiency will be zero. A similar
point in the context of pointwise asymptotics was made in \citet[Remark 5, page
8]{sun05}.

The solution to minimizing $S(k)^{r}B(k)^{1-r}$ under $\FHol{p}(\SC)$ is only
known in special cases. When $p=1$, the optimal estimator is a local constant
estimator based on the triangular kernel. When $p=2$, the solution is given in
\citet{fuller61} and \citet{zhao97} for the interior point problem, and in
\citet{gao_2018} for the boundary point problem. See
\begin{NoHyper}\Cref{sec:optim-kern-details}\end{NoHyper}
for details. When $p\geq 3$, the solution is unknown. Therefore, for $p=3$, we
compute efficiencies relative to a local quadratic estimator with a triangular
kernel. \Cref{tab:holder-eff} calculates the resulting efficiencies for
local polynomial estimators based on the uniform, triangular, and Epanechnikov
kernels. Relative to the class $\FSY{p}(\SC)$, the bias constants are smaller:
imposing smoothness away from the point of interest helps to reduce the
worst-case bias. Furthermore, the loss of efficiency from using a local
polynomial estimator of order that's too high is smaller. Finally, local linear
regression with a triangular kernel achieves high asymptotic efficiency under
both $\FSY{2}(\SC)$ and $\FHol{2}(\SC)$, both at the interior and at a boundary,
with efficiency at least 97\%, giving a theoretical justification to this
popular choice in empirical work.

\paragraph{Gains from imposing smoothness globally}
The Taylor class $\FSY{p}(\SC)$, only restricts the $p$th derivative locally to
the point of interest, while the Hölder class $\FHol{p}(\SC)$ restricts the
$p$th derivative globally. How much can one tighten a confidence interval or
reduce the RMSE due to this additional smoothness?

It follows from \Cref{theorem:maximum-bias-lp} and from arguments
underlying \Cref{single_R_thm} that the performance of using a local
polynomial estimator of order $p-1$ with kernel $k_{H}$ and optimal bandwidth
under $\FHol{p}(\SC)$ relative to using a local polynomial estimator of order
$p-1$ with kernel $k_{T}$ and optimal bandwidth under $\FSY{p}(\SC)$ is given by
\begin{equation}\label{eq:global-smooothness}
  \frac{\inf_{h>0}R_{\FHol{p}(\SC)}(\hat{T}(h;k_{H}))}{
    \inf_{h>0}R_{\FSY{p}(\SC)}(\hat{T}(h;k_{T}))}=
  \left(\frac{
      \int_{\mathcal{X}}k_{H,p-1}^{*}(u)^{2}\, \dd u
    }{\int_{\mathcal{X}}k_{T,p-1}^{*}(u)^{2}\, \dd u}\right)^{\frac{p}{2p+1}}
  \left(\frac{\mathcal{B}_{p,p-1}^{\textnormal{Höl}}(k_{H})}{
      \mathcal{B}_{p,p-1}^{T}(k_{T})
    }\right)^{\frac{1}{2p+1}}
  (1+o(1)),
\end{equation}
where $R_{\mathcal{F}}(\hat{T})$ denotes the worst-case performance of $\hat{T}$
over $\mathcal{F}$. If the same kernel is used, the first term equals 1, and the
efficiency ratio is determined by the ratio of the bias constants
$\mathcal{B}_{p,p-1}(k)$. \Cref{tab:global-smoothness-gains} computes the
resulting efficiency gain for common kernels. In general, the gains are greater
for larger $p$, and greater at the boundary. For estimation at a boundary point
with $p=2$, for example, imposing global smoothness of $f$ reduces CI length by
about 13--15\%, depending on the kernel, and about 10\% if the optimal kernel is
used.

\subsubsection{Sharp regression discontinuity}\label{sec:sharp-regr-disc}

We focus on the most empirically relevant case in which the
regression function $f$ is assumed to lie in the class $\FHol{2}(\SC)$ on either
side of the cutoff:
\begin{equation*}
  f\in\mathcal{F}_{\text{SRD}}(\SC)=\{f_{+}(x)\1{x\geq 0}-f_{-}(x)\1{x <0}\colon
  f_{+}, f_{-}\in\FHol{2}(\SC)\}.
\end{equation*}
Inference on $T(f)$ is then equivalent to inference on the difference between
two regression functions evaluated at boundary points, and the results follow by
a slight extension of the results for estimation at a boundary point in
\Cref{sec:inference-point-theory}.

It follows from the results in \Cref{sec:inference-point-theory} that if
\Cref{x_assump,sigma_assump} hold (with the requirement that $\sigma^{2}(x)$ is
continuous $0$ replaced by right- and left-continuity of $\sigma^{2}_{+}(x)$ and
$\sigma^{2}_{-}(x)$), then the variance of the estimator doesn't depend on $f$
and satisfies
\begin{equation*}
  \sd(\hat{T}(h;k))^{2}=\sum_{i=1}^{n}
  \tilde{w}^{n}(x_{i})^{2}
  \sigma^{2}(x_{i})
  =\frac{S(k)^{2}}{nh}(1+o(1)),
  \quad S(k)^{2}=\frac{{\int_{0}^{\infty}k^{*}_{1}(u)^{2}\, \dd
      u\left(
        \sigma^{2}_{+}(0)
        + \sigma^{2}_{-}(0)\right)}}{d},
\end{equation*}
with $d$ defined in \Cref{x_assump}, and
$\tilde{w}^{n}(x_{i})=w^{n}_{+}(x_{i})+w^{n}_{-}(x_{i})$.
\Cref{theorem:maximum-bias-lp} and arguments in
\begin{NoHyper}\Cref{sec:local-polyn-estim}\end{NoHyper}
imply that the bias of $\hat{T}(h;k)$ is maximized at
$f(x)=-\SC x^{2}/2\cdot (\1{x\geq 0}-\1{x <0})$, so long as the kernel
$k(\cdot)$ takes on nonnegative values. The worst-case bias therefore satisfies
\begin{equation*}
  \maxbias(\hat{T}(h;k))=
  -\frac{\SC}{2}
  \sum_{i=1}^{n}\tilde{w}^{n}(x_{i})x_{i}^{2}
  =\SC h^{2}B(k) (1+o(1)), \; B(k)=
  -\int_{0}^{\infty}u^{2}k_{1}^{*}(u)\, \dd u.
\end{equation*}
It follows that for the RMSE, FLCI, and OCI criteria,
\Cref{performance_approx_eq} holds with $\gamma_{b}=2$,
$\gamma_{s}=-1/2$, and $B(k)$ and $S(k)$ given in the displays above. Thus, the
RMSE-optimal bandwidth is given by
\begin{equation}\label{eq:rmse-rd}
  \hrmse=\left(\frac{\int_{0}^{\infty}k^{*}_{1}(u)^{2}\, \dd
      u}{(\int_{0}^{\infty}u^{2}k^{*}_{1}(u)\, \dd u)^{2}}
    \cdot\frac{\sigma^{2}_{+}(0)+\sigma^{2}_{-}(0)}{4dn\SC^{2}}
  \right)^{1/5}.
\end{equation}
The kernel efficiency results are analogous to those in
\Cref{sec:inference-point-theory}.

In principle, one could allow the bandwidths on either side of the cutoff to be
different. We show in
\begin{NoHyper}\Cref{sec:rd-with-different}\end{NoHyper},
however, that the loss in efficiency resulting from constraining the bandwidths
to be the same is quite small unless the ratio of variances on either side of
the cutoff, $\sigma_{+}^{2}(0)/\sigma_{-}^{2}(0)$, is quite large.

\subsubsection{Fuzzy regression discontinuity}\label{sec:fuzzy-regr-disc}

We assume that $f=(f_{1},f_{2})$ lies in the class
$\mathcal{F}_{\text{FRD}}(\SC_{1},
\SC_{2})=\mathcal{F}_{\text{SRD}}(\SC_{1})\times
\mathcal{F}_{\text{SRD}}(\SC_{2})$, so that both the reduced-form and the
first-stage regression functions are assumed to have a bounded second derivative
on either side of the cutoff.\footnote{While we allow the bounds $\SC_{1}$ and
  $\SC_{2}$ to change with sample size, we assume that their ratio
  $\SC_{1}/\SC_{2}$ is fixed for simplicity.}

Since the estimator is non-linear, to ensure that~\eqref{performance_approx_eq}
holds, it will be necessary to consider a sequence of parameter spaces
$\mathcal{F}_{\text{FRD}, n}(\SC_{1}, \SC_{2})$ localized around a particular
value $L^{*}$ of $L(f)=(L_{1}(f), L_{2}(f))'$ with a non-zero jump in the
first-stage regression $L^{*}_{2}\neq 0$. This allows us to apply a version of
the delta method to $\hat{L}(h; k)$. We defer details to
\begin{NoHyper}\Cref{sec:fuzzy-rd-suppl}\end{NoHyper}, where we show
that under \Cref{x_assump} and a version of \Cref{sigma_assump}, the
distribution of $\hat{T}(h;k)-T(f)$ can in large samples be approximated by a
normal distribution with variance
\begin{equation*}
  \operatorname{avar}(\hat{T}(h;k))=
  \frac{S(k)^{2}}{nh}=
  \sum_{i=1}^{n}\frac{\varsigma^{2}(x_{i};T(f))}{L_{2}(f)^{2}}
  \tilde{w}^{n}(x_{i};h, k)^{2}(1+o(1)),
\end{equation*}
and mean bounded by
\begin{equation*}
    \operatorname{\overline{abias}}(\hat{T}(h; k))
    =\SC_{1}h^{2}B(k)
     = - \frac{\SC_{1}+\abs{T(f)}\SC_{2}}{2\abs{L_{2}(f)}}
    \sum_{i=1}^{n}\tilde{w}^{n}(x_{i};h, k)x_{i}^{2}(1+o(1)),
\end{equation*}
where $\tilde{w}^{n}(x_{i};h, k)=w_{+}^{n}(x_{i})+w_{-}^{n}(x_{i})$,
$\varsigma^{2}(x_{i};T)=(1,-T)\Omega(x_{i})(1,-T)'$,
\begin{equation*}
  B(k)=-
  \frac{\int_{0}^{\infty}u^{2}k_{1}^{*}(u)\, \dd u (1+\abs{T(f)}\SC_{2}/\SC_{1})
  }{\abs{L_{2}(f)}}, \;
  S(k)^{2}=\frac{\int_{0}^{\infty}k_{1}^{*}(u)^{2}\, \dd u}{d}
  \frac{\varsigma^{2}_{+}(0;T(f))+\varsigma^{2}_{-}(0;T(f))}{L_{2}(f)^{2}},
\end{equation*}
$\varsigma_{+}^{2}(0;T)=\lim_{x\downarrow 0}\varsigma^{2}(x;T)$, and
$\varsigma_{-}^{2}(0;T)=\lim_{x\uparrow 0}\varsigma^{2}(x;T)$.

It then follows that for the FLCI, OCI, and a truncated version of the RMSE
criterion, \Cref{performance_approx_eq} holds with $\SC=\SC_{1}$,
$\gamma_{b}=2$, $\gamma_{s}=-1/2$, and $B(k)$ and $S(k)$ given in the preceding
display. The RMSE-optimal bandwidth is therefore given by
\begin{equation}\label{eq:rmse-frd}
  \hrmse=\left(\frac{\int_{0}^{\infty}k^{*}_{1}(u)^{2}\, \dd
      u}{(\int_{0}^{\infty}u^{2}k^{*}_{1}(u)\, \dd u)^{2}}
    \cdot\frac{\varsigma^{2}(T(f))}{4dn(\SC_{1}+\abs{T(f)}\SC_{2})}
  \right)^{1/5}.
\end{equation}
Since $S(k)$ and $B(k)$ depend on the kernel $k$ through the same quantities as
for inference at a boundary point, the kernel efficiency results are analogous
to those in \Cref{sec:inference-point-theory}.

Because the optimal bandwidth depends on $T(f)$, implementing a feasible version
of it requires replacing it with an initial estimate. An alternative approach to
the construction of two-sided CIs for $T(f)$ that doesn't require localization
or the use of initial estimates is an \citet{anderson_estimation_1949} style
construction studied by \citet{NoRo19}. In particular, \citet{NoRo19} propose
constructing, for each $T_{0}$, an auxiliary CI for the jump in the mean of
$y_{i}-d_{i}T_{0}$ at the cutoff, using an approach similar to that we use for
inference in sharp RD\@. The CI for $T(f)$ is then constructed by collecting all
$T_{0}$'s for which the auxiliary CI contains zero. This approach also has the
additional advantage that it can allow for weak identification
while it yields asymptotically
equivalent CIs under strong identification.\footnote{Because we require that the
  sequence of parameter spaces $\mathcal{F}_{\text{FRD}, n}(\SC_{1}, \SC_{2})$
  be localized around a value of $L^*$ with $L^*_2\neq 0$, we rule out
  sequences in which the jump in the first-stage regression is arbitrarily close
  to zero (the term ``weak identification'' refers to such sequences).
  As a result, the CI we propose, unlike the CI proposed by
  \citet{NoRo19}, is not honest over the original parameter space
  $\mathcal{F}_{\text{FRD}}(\SC_{1}, \SC_{2})$.}
See \citet{NoRo19} for a more
detailed discussion.

\subsection{Practical implementation}\label{sec:pract-impl}

We now discuss some practical issues that arise when implementing our CIs for
inference at a point, and in sharp and fuzzy RD studied in the previous
subsections. To focus the discussion, we consider smoothness classes
$\FHol{2}(\SC)$, $\mathcal{F}_{\text{SRD}}(\SC)$, and
$\mathcal{F}_{\text{FRD}}(\SC_{1}, \SC_{2})$ that constrain the second
derivative globally, so that, in the discussion below, $p=2$. In other words,
for inference at a point, we assume that the conditional mean
$f$~\Cref{eq:np-regression} is (almost everywhere) twice differentiable with the
second derivative bounded by $\SC$; for sharp RD, we assume that that $f$ is
twice differentiable on either side of the cutoff, with the second derivative
bounded by $\SC$; and for fuzzy RD, we assume that $f_{1}$ and $f_{2}$ in
in~\Cref{eq:frd-regression} are twice differentiable on either side of the
cutoff, with the second derivative bounded by $\SC_{1}$ and $\SC_{2}$,
respectively. These assumptions imply optimality of the estimators defined
in~\Cref{sec:setup-estimators} based on local linear regression ($\po=1$), which
is the most popular method in practice; they also imply that both the
Epanechnikov and the triangular kernel are nearly optimal.

\subsubsection{Choice of \texorpdfstring{$\SC$}{M}}\label{sec:choice-M}

Appropriate choice of the smoothness constant is key to implementing our method.
Since the smoothness classes we consider are convex, the results of
\citet{low97}, \citet{CaLo04} and \citet{ArKo18optimal} imply that, to maintain
honesty over the whole function class, a researcher must choose $\SC$ a priori,
rather than attempting to use a data-driven method.\footnote{These negative
  results contrast with more positive results for estimation. See, for example,
  \citet{lepski90} who, in the context of estimating the value of the
  regression function at a point, proposes a data-driven method that automates
  the choice of both $p$ and $\SC$.} We therefore recommend that, whenever
possible, problem-specific knowledge be used to decide what choice of $\SC$ is
reasonable a priori, and that one consider a range of plausible values by way of
sensitivity analysis.\footnote{\label{fn:snooping}As is well-known, if the final
  bandwidth choice is influenced by such sensitivity analysis, the resulting CI
  may undercover, even if the estimator is unbiased. In this case, one can
  combine our method with the bandwidth snooping adjustment of
  \citet{ArKo18snooping}.}

If one imposes additional restrictions on $f$ that make the parameter space for
$f$ non-convex, a data-driven method for choosing $\SC$ may be
feasible.\footnote{\label{fn:hall-horowitz}An alternative to restricting the
  parameter space is to change the notion of coverage. For example, in the
  context of constructing confidence bands for a regression function $f(x)$,
  \citet{hall_simple_2013} propose bands that have an average coverage property
  in that the bands achieve coverage of $f(x)$ for a random subset of values of
  $x$. This subset may vary with the unknown regression function and the
  realized sample.}
In \begin{NoHyper}\Cref{data-driven_bw_sec_append}\end{NoHyper}, we consider a
restriction which relates $\SC$ to a global polynomial approximation to the
regression function. In particular, the restriction formalizes the notion that
the second derivative in a neighborhood of zero is bounded by the maximum second
derivative of a $\tilde{p}$th order global polynomial approximation.
Heuristically, such restriction will hold if the local smoothness of $f$ is no
smaller than its smoothness at large scales.

This restriction allows us to calibrate $\SC$ based on the following rule of
thumb. For inference at a point, let $\breve{f}(x)$ be an estimate of $f$ based
on a global polynomial regression of order $\tilde{p}$, and let
$[x_{\min}, x_{\max}]$ denote the support of $x_i$. Put
$\Mrot=\sup_{x\in[x_{\min}, x_{\max}]}\abs{\breve{f}^{(p)}(x)}$. This rule of
thumb is similar to the suggestion of \citet[Chapter 4.2]{fg96}, with the
important distinction that their rule of thumb was designed to estimate the
pointwise-in-$f$ optimal bandwidth. We discuss the difference between this
bandwidth and $\hrmse$ in \Cref{sec:comp-with-other}. In sharp RD, the rule of
thumb is analogous, except we define $\breve{f}^{(p)}(x)$ to be the global
polynomial estimate of order $\tilde{p}$ in which the intercept and all
coefficients are allowed to be different on either side of the discontinuity
(that is, as regressors, we use $1,x_{i}, \dotsc, x_{i}^{\tilde{p}}$, and their
interactions with the indicator $\1{x_i\geq 0}$). For fuzzy RD, we use an
analogous approach to separately calibrate the reduced-form and first-stage
smoothness parameters $\SC_{1}$ and $\SC_{2}$ based on the reduced-form and
first-stage regressions.

As a default choice, we set $\tilde{p}=p+2=4$.
In \begin{NoHyper}\Cref{data-driven_bw_sec_append}\end{NoHyper}, we give a
formal analysis of this rule, showing that the resulting CIs are honest and
nearly optimal (over a regularity class that imposes the additional restriction
$f$ discussed above). In contrast, we expect that calibrating $\SC$ based on
local smoothness estimates may be difficult to justify, since estimating a local
derivative of $f$ is a harder problem than the initial problem of estimating its
value at a point. We investigate the finite-sample performance of FLCIs based on
$\Mrot$ in a Monte Carlo exercise in \Cref{monte_carlo_sec}.

\subsubsection{Computation of RMSE-optimal bandwidth}

Given a choice of $\SC$, one can compute a feasible version
$\hat{h}^{*}_{\textsc{rmse}}$ of the RMSE-optimal bandwidth by plugging this
choice into the expressions~\eqref{minimax_rmse_optimal_bw_eq},
\eqref{eq:rmse-rd}, and \eqref{eq:rmse-frd}, along with consistent estimates of
$d$, and of the variance at $0$ (for fuzzy RD, one also needs a preliminary
estimate of $T(f)$). In the simulation exercise and empirical application below,
we use an alternative approach based on directly minimizing the finite-sample
RMSE over the bandwidth $h$. To describe it, let $\tilde{w}^{n}(x_{i};h, k)$
denote the weights $w^{n}_{1}(x_{i};h, k)$ given in~\eqref{eq:lp-weights} if the
parameter of interest is the conditional mean at a point, and let
$\tilde{w}^{n}(x_{i};h, k)=w_{+}^{n}(x_{i})+w_{-}^{n}(x_{i})$ if the parameter
of interest is the sharp or fuzzy RD parameter.

For inference at a point, or for sharp RD, the finite-sample RMSE takes the form
\begin{equation}\label{eq:fs-mse-point}
  {\operatorname{RMSE}}(h; \SC)^{2}= \frac{\SC^{2}}{4}
  \left(\sum_{i=1}^{n}\tilde{w}^{n}(x_{i};h, k) x_{i}^{2}
  \right)^{2}+
  \sum_{i=1}^{n}\tilde{w}^{n}(x_{i};h, k)
  {\sigma}^{2}(x_{i}),
\end{equation}
Since $\sigma^{2}(x_{i})$ is typically unknown, one needs to replace it by an
estimate. For inference at a point, the simplest choice is to use some estimate
$\hat{\sigma}^{2}(x_{i})=\hat{\sigma}^{2}$ that assumes homoskedasticity of the
variance function. For sharp RD, one can use the estimate
$\hat{\sigma}^{2}(x_{i})=\hat{\sigma}^{2}_{+}(0)\1{x\geq
  0}+\hat{\sigma}^{2}_{-}(0)\1{x< 0}$, where $\hat{\sigma}^{2}_{+}(0)$ and
$\hat{\sigma}^{2}_{-}(0)$ are some preliminary variance estimates based on
observations above and below the cutoff. We use the bandwidth
$\hhrmse{\tilde{\SC}}$ that minimizes \Cref{eq:fs-mse-point} for
$M=\tilde{\SC}$, the chosen smoothness constant. This method was considered
previously in \citet{ArKo18optimal}. %

Since the estimate in fuzzy RD is non-linear, its moments, and hence the
finite-sample RMSE do not exist. However, one can still employ an analogous
approach minimizing the finite-sample analog of the asymptotic RMSE\@. As the
asymptotic bias and the asymptotic standard deviation both scale with the jump
in the first-stage regression at the cutoff, $L_{2}(f)$, this scaling doesn't
affect the optimum, we can equivalently minimize the asymptotic RMSE times
$L_{2}(f)$,
\begin{equation*}
  \operatorname{ARMSE}(h;\SC_{1}, \SC_{2})^{2}=
  \frac{(\SC_{1}+\abs{T(f)}\SC_{2})^{2}}{4}
  \left(\sum_{i=1}^{n}\tilde{w}^{n}(x_{i};h, k)x_{i}^{2 }\right)^{2}+
  \sum_{i=1}^{n}w_{\po}^{n}(x_{i};h;k)^{2}\varsigma^{2}(x_{i};T(f)),
\end{equation*}
with $\varsigma^{2}(x;T)=(1,-T)\Omega(x)(1,-T)'$. Since $\Omega(x_{i})$ is
unknown, one can again replace it with
$\hat{\Omega}^{2}(x_{i})=\hat{\Omega}^{2}_{+}(0)\1{x\geq
  0}+\hat{\Omega}^{2}_{-}(0)\1{x< 0}$, where $\hat{\Omega}^{2}_{+}(0)$ and
$\hat{\Omega}^{2}_{-}(0)$ are some preliminary variance estimates for
observations above and below the cutoff. As a preliminary estimate of $T(f)$,
one can take the estimate $\hat{T}(\hat{h}_0;k) $, where $\hat{h}_0$ minimizes
the above expression at $T(f)=0$. One can also use $\hat{h}_0$ directly as a
simple bandwidth selector, which, while not RMSE optimal, has the advantage that
it doesn't depend on the choice of $\SC_{2}$.

\subsubsection{Construction of FLCIs}
Given an estimate $\hhrmsez$ of $\hrmse$, such as the estimate
$\hhrmse{\tilde{\SC}}$ discussed above, an honest FLCI can be constructed as
\begin{equation}\label{eq:infeasible-CI}
  \hat{T}(\hhrmsez;k)\pm \cv_{1-\alpha}(t)\cdot
  \hatse(\hhrmsez;k),
\end{equation}
where $t$ is an estimate of the bias-sd ratio, and
$\hatse(\hhrmsez;k)$ is an estimate of the standard
error. For the standard error, many choices are available in the literature. For
inference at a point and sharp RD, the estimator
$\hat{T}(\hhrmsez;k)$ is a weighted least squares estimator,
and one can directly estimate its finite-sample conditional variance by the
nearest neighbor variance estimator considered in \citet{AbIm06match} and
\citet{AbImZh14}. Given a bandwidth $h$, the estimator takes the form
\begin{equation}\label{eq:nn-estimator}
  \hatse(h, k)^{2}=\sum_{i=1}^{n}\tilde{w}^{n}(x_{i};h, k)^{2}\hat{\sigma}^{2}(x_{i}),
  \qquad \hat{\sigma}^{2}(x_{i})=\frac{J}{J+1}\left(y_{i}
    -\frac{1}{J}\sum_{j=1}^{J}y_{j(i)}\right)^{2},
\end{equation}
for some fixed (small) $J\geq 1$, where $j(i)$ denotes the $j$th closest
observation to $i$ (for sharp RD $j(i)$ is only taken among units with the same
sign of the running variable.). In contrast, the usual Eicker-Huber-White
estimator sets $\hat{\sigma}^{2}(x_{i})=\hat{u}_{i}^{2}$, where $\hat{u}_{i}$ is
the regression residual, and it can be shown that this estimator will generally
overestimate the conditional variance. For $t$, one can either use the
asymptotic bias-sd ratio $t=1/2$, or else an estimate of the finite-sample
bias-sd ratio
$t=-\SC \sum_{i=1}^{n}\tilde{w}^{n}(x_{i};\hhrmsez, k)
x_{i}^{2}/2\hatse(\hhrmsez, k)$. We use the latter approach in the Monte Carlo
and empirical application below. While both approaches are asymptotically
equivalent when $x_{i}$ is continuous, the latter approach has the advantage
that it remains valid even when the covariates are discrete.\footnote{See
  \citet{ArKo18optimal}, \citet{KoRo16} and \citet{ImWa17} for a more thorough discussion of the case
  with discrete covariates.}

For fuzzy RD, one can use an analogous approach to estimate the standard error
as
\begin{equation*}
  \hatse(h, k)^{2} =\frac{1}{\hat{L}_{2}(h;k)^{2}}
  \sum_{i=1}^{n}\tilde{w}^{n}(x_{i};h, k)^{2}\hat{\varsigma}^{2}(x_{i}, \hat{T}
  (h;k)),
\end{equation*}
where
$\hat{\varsigma}^{2}(x_{i};T)=\frac{J}{J+1}(1,-T)(z_{i}
-\frac{1}{J}\sum_{j=1}^{J}z_{j(i)})(z_{i}
-\frac{1}{J}\sum_{j=1}^{J}z_{j(i)})'(1,-T)'$, $z_{i}=(y_{i}, d_{i})'$, and $j(i)$
denotes that $j$th closest observation with the same sign of the running
variable. For $t$, one can use $t=1/2$, or else the finite-sample analog of the
asymptotic bias-sd ratio,\footnote{For inference based on
  $\hat{T}(\hat{h}_{0};k)$, it is necessary to use the finite-sample analog of
  the bias-sd ratio, since the bandwidth $\hat{h}_{0}$ is not RMSE optimal.}
$t=-(\tilde{\SC}_{1}+\abs{\hat{T}}\tilde{\SC}_{2})\cdot
\sum_{i=1}^{n}\tilde{w}^{n}(x_{i};\hhrmsez, k)x_{i}^{2} /\allowbreak
2\sqrt{\sum_{i=1}^{n}\hat{\varsigma}^{2}(x_{i};\hat{T})\tilde{w}^{n}(x_{i};\hhrmsez, k)^{2}}$.

\section{Comparison with other approaches}\label{sec:comp-with-other}

In this section, we compare our approach to inference about the parameter $T(f)$
to three other approaches to inference. To make the comparison concrete, we make
the comparison in the context of inference about a nonparametric regression
function at a point, discussed in \Cref{sec:applications}. The first approach,
which we term ``conventional,'' ignores the potential bias of the estimator and
constructs the CI as $\hat{T}_{\po}(h, k)\pm z_{1-\alpha/2}\hatse(h;k)$. The
bandwidth $h$ is typically chosen to minimize the asymptotic mean squared error
(MSE) of $\hat{T}_{\po}(h;k)$ under pointwise-in-$f$ (or ``pointwise'', for
short) asymptotics. We refer to this bandwidth as $\hpt$. We discuss the
distinction between $\hpt$ and the bandwidth $\hrmse$ in
\Cref{sec:rmse-pointw-optim}. Under the second approach, undersmoothing, one
chooses a sequence of smaller bandwidths, so that in large samples, the bias of
the estimator is dominated by its standard error. Finally, in bias correction,
one re-centers the conventional CI by subtracting an estimate of the leading
bias term from $\hat{T}_{\po}(h;k)$. In \Cref{sec:effic-comp}, we compare the
coverage and length properties of these CIs to the fixed-length CI (FLCI) based
on $\hat{T}_{\po}(\hrmse;k)$.

Implementing any of these CIs in practice requires feasible bandwidth and tuning
parameter choices. This may require auxiliary assumptions (such as assumptions
relating local and global smoothness of $f$ if one picks $\SC$ using the rule of
thumb discussed in~\Cref{sec:choice-M}), which may differ across the methods.
For clarity of comparison, we keep implementation issues separate, and focus in
this section on a theoretical comparison, assuming any tuning parameters
(including the smoothness parameter $\SC$) are known. The Monte Carlo exercise
in \Cref{monte_carlo_sec} below considers their finite-sample performance when
the tuning parameters need to be chosen.

\subsection{RMSE and pointwise optimal bandwidth}\label{sec:rmse-pointw-optim}

The RMSE optimal bandwidth given in \Cref{minimax_rmse_optimal_bw_eq} seeks to
minimize the asymptotic approximation to the maximum RMSE (or, equivalently,
MSE) over $f\in\FSY{p}(\SC)$ or $f\in\FHol{p}(\SC)$. In contrast, the bandwidth
$\hpt$ is intended to optimize the MSE at the function $f$ itself. In
particular, it minimizes the sum of the leading squared bias and variance terms
under pointwise asymptotics for the case $\po=p-1$. It is given by \citep[see,
for example,][Eq.~(3.20)]{fg96}
\begin{equation}\label{pointwise_rmse_optimal_bw_eq}
  \hpt= \left(
    \frac{\sigma^{2}(0)p!^{2}}{2p n d f^{(p)}(0)^{2}}
    \frac{\int_{\mathcal{X}}k^{*}_{\po}(u)^{2}\, \dd u}{
      (\int_{\mathcal{X}} t^{p}k_{\po}^{*}(t)\, \dd t)^{2}}
  \right)^{\frac{1}{2p+1}}.
\end{equation}
Comparing this expression with that for $\hrmse$ in
\Cref{minimax_rmse_optimal_bw_eq}, we see that the pointwise optimal bandwidth
replaces $\SC$ with the $p$th derivative at zero, $f^{(p)}(0)$, and it replaces
$\mathcal{B}_{p, \po}(k)$ with $\int_{\mathcal{X}} t^{p}k_{\po}^{*}(t)\, \dd t$.
Note that
$\mathcal{B}_{p, \po}(k)\ge \abs{\int_{\mathcal{X}} t^{p}k_{\po}^{*}(t)\, \dd
  t}$ (this can be seen by noting that the right-hand side corresponds to the
bias at the function $f(x)=\pm x^p/p!{}$, while the left-hand side is the
supremum of the bias over functions with $p$th derivative bounded by $1$). Thus,
assuming that $f^{(p)}(0)\le \SC$ (this holds by definition for any
$f\in\mathcal{F}$ when $\mathcal{F}=\FHol{p}(\SC)$), we will have
$\hpt/\hrmse
\geq\left(\SC/\abs{f^{(p)}(0)}
\right)^{\frac{2}{2p+1}} \ge 1$.

Even though the bandwidth $\hpt$ is intended to optimize the RMSE at the
function $f$ itself, its performance may be arbitrarily bad relative to $\hrmse$
at functions for which $f^{(p)}(0)$ is close to zero. For example, consider the
function $f(x)=x^{p+1}$ if $p$ is odd, or $f(x)=x^{p+2}$ if $p$ is even. This is
a smooth function with all derivatives bounded on the support of $x_i$. Since
$f^{(p)}(0)=0$, $\hpt$ is infinite, and the resulting estimator is a global
$p$th order polynomial least squares estimator. Its RMSE will be poor, since the
estimator is not even consistent.\footnote{To ensure consistency and finiteness
  of $\hpt$, it is standard to assume that $f^{(p)}\neq 0$. However, the RMSE
  can still be arbitrarily poor whenever the $p$th derivative is locally small,
  but non-zero, and large globally, such as when $f(x)=x^{p+1}+\eta x^{p}$ for
  $p$ odd and $f(x)=x^{p+2}+\eta x^{p}$ if $p$ is even, provided $\eta$ is
  sufficiently small.}

To address this problem, plug-in bandwidths that estimate $\hpt$ include tuning
parameters to prevent them from approaching infinity. The RMSE of the resulting
estimator at such functions is then determined almost entirely by these tuning
parameters. Furthermore, if one uses such a bandwidth as an input to an
undersmoothed or bias-corrected CI, the coverage will be determined by these
tuning parameters, and can be arbitrarily bad if the tuning parameters allow the
bandwidth to be large. Indeed, we find in our Monte Carlo analysis in
\Cref{monte_carlo_sec} that plug-in estimates of $\hpt$ used in practice
can lead to very poor coverage even when used as a starting point for a
bias-corrected or undersmoothed estimator.

\subsection{Efficiency and coverage comparison}\label{sec:effic-comp}

Let us now consider the efficiency and coverage properties of conventional,
undersmoothed, and bias-corrected CIs relative to the FLCI based on
$\hat{T}_{p-1}(\hrmse, k)$. To keep the comparison meaningful, and avoid the
issues discussed in the previous subsection, we assume these CIs are also based
on $\hrmse$, rather than $\hpt$ (in case of undersmoothing, we assume that the
bandwidth is undersmoothed relative to $\hrmse$). Suppose that the smoothness
class is either $\FSY{p}(\SC)$ or $\FHol{p}(\SC)$ and denote it by
$\mathcal{F}_{p}(\SC)$. For concreteness, let $p=2$, and $\po=1$.

Consider first conventional CIs, given by
$\hat{T}_{1}(h;k)\pm z_{1-\alpha/2}\hatse(h;k)$. If the bandwidth $h$ equals
$\hrmse$, then these CIs are shorter than the 95\% FLCIs by a factor of
$z_{0.975}/\cv_{0.95}(1/2)=0.90$. Consequently, their coverage is $92.1\%$
rather than the nominal $95\%$ coverage.
At the RMSE-optimal bandwidth, the bias-sd ratio equals
$1/2$, so disregarding the bias doesn't result in severe undercoverage. If one
uses a larger bandwidth, however, the bias-sd ratio will be larger, and the
undercoverage problem more severe: for example, if the bandwidth is 50\% larger
than $\hrmse$, so that the bias-sd ratio equals $1/2\cdot
(1.5)^{(5/2)}$, the coverage is only $71.9\%$.

Second, consider undersmoothing. This amounts to choosing a bandwidth sequence
$h_n$ such that $h_n/\hrmse\to 0$, so that for any fixed $\SC$, the bias-sd
ratio $t_n=h_n^{\gamma_b-\gamma_s}\SC B(k)/(n^{-1/2}S(k))$ approaches zero, and
the CI
$\hat T(h^n;k)\pm \cv_{1-\alpha}(0)\hatse(h_n;k)=\hat T(h^n;k)\pm
z_{1-\alpha/2}\hatse(h_n;k)$ will consequently have proper coverage in large
samples. However, the CIs shrink at a slower rate than $n^{r/2}=n^{4/5}$, and
thus the asymptotic efficiency of the undersmoothed CI relative to the optimal
FLCI is zero.

On the other hand, an apparent advantage of the undersmoothed CI is that it
appears to avoid specifying the smoothness constant $\SC$. However, a more
accurate description of undersmoothing is that the bandwidth sequence $h_{n}$
implicitly chooses a sequence of smoothness constants $\SC_n\to\infty$ such that
coverage is controlled under the sequence of parameter spaces
$\mathcal{F}_{p}(\SC_n)$. We can improve on the coverage and length of the
resulting CI by making this sequence explicit and computing an optimal (or
near-optimal) FLCI for $\mathcal{F}_{p}(\SC_n)$.

To this end, given a sequence $h_n$, a better approximation to the finite-sample
coverage of the CI $\hat T(h_n;k)\pm z_{1-\alpha/2}\hatse(h_n;k)$ over the
parameter space $\mathcal{F}_{p}(\SC)$ is
$P_{Z\sim N(0,1)}(\abs{Z+t_n(\SC)}\ge z_{1-\alpha/2})$ where
$t_n(\SC)=h_n^{\gamma_b-\gamma_s}\SC B(k)/(n^{-1/2}S(k))$ is the bias-sd ratio
for the given choice of $\SC$. This approximation is exact in idealized
settings, such as the white noise model discussed in
\begin{NoHyper}\Cref{verification_sec}\end{NoHyper}. For a
given level of undercoverage $\eta=\eta_n$, one can then compute $\SC_n$ as the
greatest value of $\SC$ such that this approximation to the coverage is at least
$1-\alpha-\eta$. In order to trust the undersmoothed CI, one must be convinced
of the plausibility of the assumption $f\in\mathcal{F}_{p}(\SC_n)$: otherwise
the coverage will be worse than $1-\alpha-\eta$. This suggests that, in the
interest of transparency, one should make this smoothness constant explicit by
reporting $\SC_n$ along with the undersmoothed CI\@. However, once the sequence
$\SC_{n}$ is made explicit, a more efficient approach is to simply report an
optimal or near-optimal CI for this sequence, either at the coverage level
$1-\alpha-\eta$ (in which case the CI will be strictly smaller than the
undersmoothed CI while maintaining the same coverage) or at level $1-\alpha$ (in
which case the CI will have better finite-sample coverage and may also be
shorter than the undersmoothed CI).

Finally, let us consider bias correction. It is known that re-centering
conventional CIs by an estimate of the leading bias term often leads to poor
coverage \citep{hall_effect_1992}. In an important paper, \citet[CCT
hereafter]{cct14} show that the coverage properties of this bias-corrected CI
are much better if one adjusts the standard error estimate to account for the
variability of the bias estimate, which they call robust bias correction (RBC).
For simplicity, consider the case in which the main bandwidth and the pilot
bandwidth (used to estimate the bias) are the same, and that the main bandwidth
is chosen optimally in that it equals $\hrmse$. In this case, the bias-corrected
local linear estimator coincides with a local quadratic estimator. As a result,
the RBC procedure in this case amounts to using a local quadratic estimator, but
with a bandwidth $\hrmse$, optimal for a local linear estimator. The resulting
CI obtains by adding and subtracting $z_{1-\alpha/2}$ times the standard
deviation of the estimator.

To ensure that the bias is estimable, the theory of bias correction requires
that the conditional mean function is sufficiently smooth, which requires
$\po<p-1$ (thus, assuming that $f$ is sufficiently smooth to ensure that the
bias of $\hat{T}_{1}(h;k)$ can be estimated implies that the polynomial order
$\po=1$ of the original estimator is not optimal). Suppose, therefore, that the
smoothness class is given by $\mathcal{F}_{3}(\SC)$ (with $\po=1$, and
$h=\hrmse$ still chosen to be MSE optimal for $\mathcal{F}_{2}(\SC)$). In this
case the RBC interval can be considered an undersmoothed CI based on a second
order local polynomial estimator. Following the discussion of undersmoothed CIs
above, the limiting coverage is $1-\alpha$ when $\SC$ is fixed (this matches the
pointwise-in-$f$ coverage statements in CCT, which assume the existence of a
continuous third derivative in the present context). Due to this undersmoothing,
however, the RBC CI shrinks at a slower rate than the optimal CI\@.

It is also interesting to consider the case when the order $\po=1$ of the local
polynomial of the estimator $\hat{T}_{1}(\hrmse;k)$ is optimal under the
maintained smoothness assumption, so that the smoothness class is given by
$\mathcal{F}_{2}(\SC)$. In this case, the smoothness of the conditional mean
function is too low for the bias to be estimable: the bias of the bias-corrected
estimator will be of the same order as the bias of the original estimator.
Consequently, the estimator will remain asymptotically biased, even after the
bias correction. In particular, bias-sd ratio of the estimator is given by
\begin{equation}\label{eq:bias-sd-cct}
  t_{\text{RBC}}=(\hrmse)^{5/2}
  \frac{\SC \mathcal{B}_{2,2}(k) /2}{\sigma(0)
    (\int k^{*}_{2}(u)^{2}\, \dd u/dn)^{1/2}}=
  \frac{1}{2}
  \frac{\mathcal{B}_{2,2}(k)}{\mathcal{B}_{2,1}(k)}
  \left(\frac{\int_{\mathcal{X}}k^{*}_{1}(u)^{2}\, \dd u}{
      {\int_{\mathcal{X}}k^{*}_{2}(u)^{2}\, \dd u}}\right)^{1/2}.
\end{equation}
The resulting coverage is given by
$\Phi(t_{\text{RBC}}+z_{1-\alpha/2})-\Phi(t_{\text{RBC}}-z_{1-\alpha/2})$. The
RBC interval length relative to the $1-\alpha$ FLCI around a local linear
estimator with the same kernel and minimax MSE bandwidth is the same under both
$\FSY{p}(\SC)$, and $\FHol{p}(\SC)$, and given by
\begin{equation}\label{eq:length-ratio-cct}
  \frac{z_{1-\alpha/2}\left(\int_{\mathcal{X}}k^{*}_{2}(u)^{2}\, \dd u\right)^{1/2}
  }{\cv_{1-\alpha}(1/2)\left(\int_{\mathcal{X}}k^{*}_{1}(u)^{2}\, \dd u\right)^{1/2}}
  (1+o(1)).
\end{equation}

The resulting coverage and relative length is given in
\Cref{tab:cct-minimax-mse}. One can see that although the undercoverage is
very mild, (since $t_{\text{RBC}}$ is quite low in all cases), the intervals are
about 30\% longer than the FLCIs around the RMSE bandwidth.

Under the class $\FHol{2}(\SC)$, the RBC intervals are also reasonably robust to
using a larger bandwidth: if the bandwidth used is 50\% larger than $\hrmse$, so
that the bias-sd ratio in \Cref{eq:bias-sd-cct} is larger by a factor of
$(1.5)^{5/2}$, the resulting coverage is still at least 93.0\% for the kernels
considered in \Cref{tab:cct-minimax-mse}. Under $\FSY{2}(\SC)$, using a
bandwidth 50\% larger than $\hrmse$ yields coverage of about $80\%$ on the
boundary and 87\% in the interior. Thus, depending on the smoothness class, the
95\% RBC CI has close to 95\% coverage and efficiency loss of about 30\%, or
exactly 95\% coverage at the cost of shrinking at a slower than optimal rate.

Our asymptotic efficiency comparisons focus on minimizing length among CIs with
coverage at least $1-\alpha$ for all $f\in\mathcal{F}$, which follows the usual
definition of coverage. One may also consider a criterion that also penalizes
CIs that cover ``too much,'' by placing an upper bound $1-\underline\alpha$ on
coverage. For the CIs considered in this paper, the maximum coverage occurs when
the bias is zero, and is given by
$P_{Z\sim N(0,1)}(|Z|\le \cv_{1-\alpha}(t))=1-2\Phi(-\cv_{1-\alpha}(t))$ where
$t$ is the asymptotic bias-sd ratio. In particular, when
$\mathcal{F}=\FSY{2}(\SC)$ or $\mathcal{F}=\FHol{2}(\SC)$ and the RMSE optimal
bandwidth is used, the maximum coverage of a FLCI with $95\%$ (minimum) coverage
is $1-2\Phi(-2.18)=.971$. If one wants the maximum coverage to be smaller, then
undersmoothing (or subtracting an estimate of the bias) will be necessary, and
Edgeworth expansions may be needed to deal with higher order approximation terms
if one wants $\alpha-\underline\alpha\to 0$ quickly enough with the sample size
\citep[see][]{calonico_coverage_2019}. Because, as we discuss above,
undersmoothing or bias correction yields longer CIs than the ones we propose,
the resulting CIs will be longer than the CIs we propose, which do not penalize
``overcoverage.''

\section{Monte Carlo}\label{monte_carlo_sec}

To study the finite-sample performance of the FLCI that we propose, and compare
its performance to other approaches, this section conducts a Monte Carlo
analysis of the conditional mean estimation problem considered in
\Cref{sec:applications}.

We consider Monte Carlo designs with conditional mean functions
\begin{align*}
  f_{1}(x)&=\frac{\SC}{2}(x^{2}- 2\textsf{s}(\abs{x}-0.25)), \\
  f_{2}(x)&=\frac{\SC}{2}(x^{2}-2\textsf{s}(\abs{x}-0.2)^{2}
            +2\textsf{s}(\abs{x}-0.5)
            -2\textsf{s}(\abs{x}-0.65)), \\
  f_{3}(x)&=\frac{\SC}{2}((x+1)^{2}-2\textsf{s}(x+0.2)+2\textsf{s}(x-0.2)
            -2\textsf{s}(x-0.4) +2\textsf{s}(x-0.7)-0.92),
\end{align*}
where $\textsf{s}(x)=(x)_{+}^{2}=\max\{x, 0\}^{2}$ is the square of the plus
function, and $\SC\in\{2,6\}$, giving a total of 6 designs. In all cases,
$x_{i}$ is drawn from a uniform distribution with support $[-1,1]$ (so that the
design is random), $u_i\sim N(0,1/4)$, and the sample size is $n=500$.
\Cref{fig:mc-fkt} plots these designs. The regression function for each
design lies in $\FHol{2}(\SC)$ for the corresponding $\SC$. To ensure that our
results, discussed below, are not sensitive to the choice of the error
distribution or the distribution for the running variable,
in \begin{NoHyper}\Cref{sec:addit-monte-carlo}\end{NoHyper}, we also consider
designs with $x_{i}$ drawn from a beta distribution, designs with log-normal and
heteroskedastic errors, and designs with different error variance. Finally, we
also show in the appendix that the results remain effectively the same when the
function $\textsf{s}(\cdot)$ is replaced by a smooth approximating
function.\footnote{\label{fn:rbc-simulation}The RBC method considered below
  assumes that the conditional mean function be at least three times
  continuously differentiable in the neighborhood of $0$. Since the functions
  $f_{1}, f_{2}$ and $f_{3}$ are not globally three times continuously
  differentiable, depending on the neighborhood definition, this assumption is
  arguably violated. The results in the appendix are nearly identical to those
  reported here, implying that the performance of the RBC method is not driven
  by this lack of smoothness.}

For each design, we implement the optimal FLCI centered at a local linear
estimate with a triangular kernel and MSE optimal bandwidth, as described in
\Cref{sec:pract-impl}, for each choice of $\SC\in\{2,6\}$, and with $\SC$
calibrated using the rule-of-thumb (ROT) described in
\Cref{sec:pract-impl}. The implementations with $\SC\in\{2,6\}$ allow us
to gauge the effect of using an appropriately calibrated $\SC$, compared to a
choice of $\SC$ that is either too conservative or too liberal by a factor of 3.
The ROT calibration chooses $\SC$ automatically, but requires additional
conditions in order to have correct coverage (see \Cref{sec:pract-impl}).

In addition to these FLCIs, we consider seven other CIs
(\begin{NoHyper}\Cref{sec:addit-monte-carlo}\end{NoHyper} considers one more
method). The first five are different implementations of the robust
bias-corrected (RBC) CIs proposed by CCT (discussed in
\Cref{sec:comp-with-other}). Implementing these CIs requires two bandwidth
choices: a bandwidth for the local linear estimator, and a pilot bandwidth that
is used to construct an estimate of its bias. The first two CIs use bandwidth
choices justified by pointwise-in-$f$ asymptotics. The first CI uses a plug-in
estimate of $\hpt$ defined in~\eqref{pointwise_rmse_optimal_bw_eq}, as
implemented by \citet{ccf15}, and an analogous estimate for the pilot bandwidth.
The second CI, also implemented by \citet{ccf15}, uses bandwidth estimates for
both bandwidths that optimize the pointwise asymptotic coverage error (CE) among
CIs that use usual $z_{1-\alpha/2}$ critical value. This CI can be considered a
particular form of undersmoothing. The third CI sets both the pilot bandwidth
and the main bandwidth to the plug-in estimate of $\hpt$. For the next three
CIs, we consider bandwidths justified by uniform-in-$f$ asymptotics. For the
fourth and fifth CIs, we set both the main and the pilot bandwidth to $\hrmse$
with $\SC=2$, and $\SC=6$, respectively. For the sixth CI, we set both
bandwidths to $\hhrmse{\Mrot}$. Finally, we consider a conventional CI centered
at a plug-in bandwidth estimate of $\hpt$, using the rule-of-thumb estimator of
\citet[][Chapter 4.2]{fg96}. All CIs are computed at the nominal $95\%$ coverage
level.

\Cref{tab:mc} reports the results. The FLCIs perform well when the correct
$\SC$ is used. As expected, they suffer from undercoverage if $\SC$ is chosen
too small, or suboptimal length when $\SC$ is chosen too large. The ROT choice
of $\SC$ appears to do a reasonable job of having good coverage and length in
these designs without requiring knowledge of the true smoothness constant.
However, as discussed in \Cref{sec:pract-impl}, this ROT choice
imposes additional restrictions on the parameter space, so one must take care in
extrapolating these results to other designs.

As predicted by the theory in \Cref{sec:comp-with-other}, the RBC CIs
also have good coverage when implemented using the $\hrmse$ bandwidth, and they
are less sensitive to the choice of $\SC$ than the corresponding FLCIs, at the
expense of being on average about 25\% longer. RBC CIs with bandwidth given by
$\hhrmse{\Mrot}$ also achieve good coverage, but they are again about 25\%
longer than the corresponding FLCIs.

The CIs based on bandwidths justified by pointwise-in-$f$ asymptotics (rows 1,
2, 3, and 7 for each design in the table) all have very poor coverage for at least
one of the designs. Our analysis in \Cref{sec:comp-with-other} suggests that
this is due to the tuning parameter choices required by these bandwidths.
Indeed, looking at the average of the bandwidth over the Monte Carlo draws (also
reported in \Cref{tab:mc}), it can be seen that the bandwidths tend to be
much larger than those that estimate $\hrmse$. This is even the case for the CE
bandwidth, which is intended to minimize coverage errors.

Overall, the Monte Carlo analysis suggests that default approaches to
nonparametric CI construction (bias-correction or undersmoothing relative to
plug-in bandwidths) can lead to severe undercoverage when implemented using
bandwidths justified by pointwise-in-$f$ asymptotics. Bias-corrected CIs such as
the one proposed by CCT can have good coverage if one starts from the minimax
RMSE bandwidth, although they will be wider than FLCIs proposed in this paper.

\section{Empirical illustration}\label{application_sec}
To illustrate the implementation of feasible versions of the
CIs~\eqref{eq:infeasible-CI}, we use a subset of the dataset from
\citet{LuMi07}.

In 1965, when the Head Start federal program launched, the Office of Economic
Opportunity provided technical assistance to the 300 poorest counties in the
United States to develop Head Start funding proposals. \citet{LuMi07} use this
cutoff in technical assistance to look at intent-to-treat effects of the Head
Start program on a variety of outcomes using as a running variable the county's
poverty rate relative to the poverty rate of the 300th poorest county (which had
poverty rate equal to approximately 59.2\%). We focus here on their main
finding, the effect on child mortality due to causes addressed as part of Head
Start's health services. See \citet{LuMi07} for a detailed description of this
variable. Relative to the dataset used in \citet{LuMi07}, we remove one
duplicate entry and one outlier, which after discarding counties with partially
missing data leaves us with 3,103 observations, with 294 of them above the
poverty cutoff.

\Cref{fig:lm-rawdata} plots the data (to reduce the noise in the outcome
variable, we plot bin averages of size 25). To estimate the discontinuity in
mortality rates, \citet{LuMi07} use a uniform kernel\footnote{\citet{LuMi07}
  state that the estimates were obtained using a triangular kernel. However, due
  to a bug in the code, the results reported in the paper were actually obtained
  using a uniform kernel.} and consider bandwidths equal to 9, 18, and 36. This
yields point estimates equal to $-1.90$, $-1.20$ and $-1.11$ respectively, which
are large effects given that the average mortality rate for counties not
receiving technical assistance was $2.15$ per 100,000. The $p$-values reported
in the paper, based on bootstrapping the $t$-statistic (which ignores any
potential bias in the estimates), are 0.036, 0.081, and 0.027. The standard
errors for these estimates, obtained using the nearest neighbor method (with
$J=3$) are $1.04$, $0.70$, and $0.52$.

These bandwidth choices are optimal in the sense that they minimize the RMSE
expression~\eqref{eq:fs-mse-point} if $\SC=0.040$, $0.0074$, and $0.0014$,
respectively. Thus, for these bandwidths to be optimal, one has to be very
optimistic about the smoothness of the regression function. In comparison, the
rule of thumb method for estimating $\SC$ discussed in \Cref{sec:pract-impl}
yields $\Mrot=0.299$, implying $\hrmse$ estimate $4.0$, and the point estimate
$-3.17$. For these smoothness parameters, the critical values based on the
finite-sample bias-sd ratio are given by $2.165$, $2.187$, $2.107$ and $2.202$
respectively, which is very close to the asymptotic value
$\cv_{.95}(1/2)=2.181$. The resulting 95\% confidence intervals are given by
\begin{align*}
  (-4.143, 0.353), && (-2.720, 0.323), &&(-2.215, -0.013),
  &&\text{and}&&(-6.352, 0.010),
\end{align*}
respectively. The $p$-values based on these estimates are given by $0.100$,
$0.125$, $0.047$, and $0.051$. These $p$-values are larger than those reported in
the paper, as they take into account the potential bias of the estimates.

Using a triangular kernel helps to tighten the confidence intervals by a few
percentage points in length, as predicted by the relative asymptotic efficiency
results from \Cref{tab:holder-eff}, yielding
\begin{align*}
  (-4.138, 0.187), &&(-2.927, 0.052), &&(-2.268, -0.095), &&\text{and}
  &&(-5.980, -0.322)
\end{align*}
The underlying optimal bandwidths are given by $11.6$, $23.1$, $45.8$, and $4.9$
respectively. The $p$-values associated with these estimates are $0.074$,
$0.059$, $0.033$, and $0.028$, tightening the $p$-values based on the uniform
kernel.

These results indicate that unless one is very optimistic about the smoothness
of the regression function, the uncertainty associated with the magnitude of the
effect of Head Start assistance on child mortality is much higher than
originally reported. This is due mainly to the relatively large bandwidths used
by \citet{LuMi07}, which imply an optimistic bound on the smoothness of the
regression function if we assume that such bandwidths are close to optimal for
MSE\@. Interestingly, while the more conservative smoothness bound in our
benchmark specification leads to much wider CIs, the point estimate is larger in
magnitude, so that one still finds a statistically significant effect at a 5 or
10\% level, depending on the kernel.

\clearpage

\begin{appendices}
\crefalias{section}{appsec}
\section{Proofs of theorems in Section~\ref{results_sec}}\label{proofs_sec}

\subsection{Proof of Theorem~\ref{single_R_thm}}
Parts~\ref{item:theorem-2} and~\ref{item:theorem-3} follow from
part~\ref{item:theorem-1} and simple calculations. To prove
part~\ref{item:theorem-1}, note that, if it did not hold, there would be a
bandwidth sequence $h_{n}$ such that
\begin{equation*}
  \liminf_{n\to\infty} \SC^{r-1} n^{r/2} R(\hat T(h_{n};k)) < S(k)^{r}B(k)^{1-r}\inf_t
  t^{r-1}\tilde{R}(t,1).
\end{equation*}
By \Cref{large_small_h_eq}, the bandwidth sequence $h_{n}$ must
satisfy $\liminf_{n\to\infty} h_n(n\SC^2)^{1/[2(\gamma_b-\gamma_s)]}>0$ and
$\limsup_{n\to\infty} h_n(n\SC^2)^{1/[2(\gamma_b-\gamma_s)]}<\infty$. Thus, by
\Cref{Rt_eq},
\begin{equation*}
  \SC^{r-1}n^{r/2}
  R(\hat{T}(h_{n};k)) =S(k)^{r}B(k)^{1-r}t_n^{r-1}\tilde R(t_n,1)+o(1),
\end{equation*}
where $t_n=h_n^{\gamma_b-\gamma_s}B(k)/\allowbreak (n^{-1/2}S(k))$. This
contradicts the display above.

\subsection{Proof of Theorem~\ref{two_R_thm}}

The second statement (relative efficiency) is immediate from~\eqref{Rt_eq}. For
the first statement (coverage), fix $\varepsilon>0$ and let
$\sd_n=n^{-1/2}(\hrmse)^{\gamma_s}S(k)$ so that
$\sd_n/\hatse(\hrmse;k)\stackrel{p}{\to} 1$ uniformly over $f\in\mathcal{F}$.
Note that, by \Cref{single_R_thm} and the fact that $t^*_{\RMSE}=\sqrt{1/r-1}$,
\begin{equation*}
  \tilde R_{\FLCI, \alpha+\varepsilon}(\hat T(\hrmse;k)) =
  \sd_n \cdot \cv_{1-\alpha-\varepsilon}(\sqrt{1/r-1})(1+o(1)),
\end{equation*}
and similarly for
$\tilde{R}_{\FLCI, \alpha-\varepsilon}(\hat{T}(\hrmse;k))$.
Since $\cv_{1-\alpha}(\sqrt{1/r-1})$ is strictly decreasing in $\alpha$, it
follows that there exists $\eta>0$ such that, with probability approaching 1
uniformly over $f\in\mathcal{F}$,
\begin{multline*}
  R_{\FLCI, \alpha+\varepsilon}(\hat T(\hrmse;k))
  <\hatse(\hat T(\hrmse;k)) \cdot
  \cv_{1-\alpha}(\sqrt{1/r-1})\\
  <(1-\eta)R_{\FLCI, \alpha-\varepsilon}(\hat{T}
  (\hrmse;k)).
\end{multline*}
Thus,
\begin{multline*}
  \liminf_n\inf_{f\in\mathcal{F}}P\left(T(f)\in\left\{\hat{T} (\hrmse;k) \pm
      \hatse(\hat{T}(\hrmse;k)) \cdot \cv_{1-\alpha}(\sqrt{1/r-1})\right\}
  \right) \\
  \ge \liminf_n\inf_{f\in\mathcal{F}}P\left(T(f)\in\left\{ \hat{T}(\hrmse;k)\pm
      R_{\FLCI, \alpha+\varepsilon}(\hat{T}(\hrmse;k))\right\} \right) \ge
  1-\alpha-\varepsilon,
\end{multline*}
and
\begin{multline*}
  \limsup_n\inf_{f\in\mathcal{F}}P\left(T(f)\in\left\{\hat{T}(\hrmse;k)\pm
      \hatse(\hat T(\hrmse;k)) \cdot \cv_{1-\alpha}(\sqrt{1/r-1})\right\}
  \right) \\
  \le \limsup_n\inf_{f\in\mathcal{F}}P\left(T(f)\in\left\{\hat{T} (\hrmse;k)\pm
      R_{\FLCI, \alpha-\varepsilon}(\hat{T}(\hrmse;k))(1-\eta)\right\}
  \right) \le 1-\alpha+\varepsilon,
\end{multline*}
where the last inequality follows by definition of
$R_{\FLCI, \alpha-\varepsilon}(\hat T(\hrmse;k))$. Taking
$\varepsilon\to 0$ gives the result.

\end{appendices}

\onehalfspacing
\bibliography{../np-testing-library}

\begin{table}[p]
  \centering
  \begin{threeparttable}
    \caption{Critical values $\cv_{1-\alpha}(\cdot)$}\label{tab:cvs}
    \begin{tabular*}{0.9\linewidth}{@{\extracolsep{\fill}}llrrr@{}}
      &    & \multicolumn{3}{c}{$\alpha$}\\
      \cmidrule(rl){3-5}
      $r$&    $t$& $0.01$ & 0.05 & 0.1\\
      \midrule
      & 0.0 & 2.576 & 1.960 & 1.645 \\
      6/7 &0.408& 2.764 & 2.113 & 1.777 \\
      4/5 & 0.5 & 2.842 & 2.181 & 1.839 \\
      2/3 & 0.707& 3.037 & 2.362 & 2.008 \\
      1/2 &  1.0 & 3.327 & 2.646 & 2.284 \\
      &  1.5 & 3.826 & 3.145 & 2.782 \\
      &  2.0 & 4.326 & 3.645 & 3.282\\[1ex]
    \end{tabular*}
    \begin{tablenotes}
    \item \emph{Notes:} Critical values $\cv_{1-\alpha}(t)$ and
      $\cv_{1-\alpha}(\sqrt{1/r-1})$, for the FLCIs
      in~\eqref{eq:honest-two-sided-ci} and \eqref{eq:honest-two-sided-ci-rmse},
      corresponding to the $1-\alpha$ quantiles of the $\abs{N(t,1)}$ and
      $\abs{N(\sqrt{1/r-1},1)}$ distributions, where $t$ is the bias-sd ratio,
      and $r$ is the rate exponent. For $t\geq 2$,
      $\cv_{1-\alpha}(t)\approx t+z_{1-\alpha/2}$ up to 3 decimal places for
      these values of $\alpha$.
    \end{tablenotes}
  \end{threeparttable}
\end{table}

\begin{table}[p]
  \centering
  \begin{threeparttable}
    \caption{Relative efficiency of local polynomial estimators for the function
      class $\FSY{p}(\SC)$.}\label{tab:taylor-eff}
    \begin{tabular*}{0.9\linewidth}{@{\extracolsep{\fill}}@{}lcllllll@{}}
      & & \multicolumn{3}{c}{Boundary Point} &\multicolumn{3}{c}{Interior point}\\
      \cmidrule(rl){3-5}\cmidrule(rl){6-8}
      Kernel &  Order & $p=1$ & $p=2$ & $p=3$& $p=1$ & $p=2$ & $p=3$\\
      \midrule
      \multirow{3}{*}{\begin{tabular}{l}Uniform\\ $\1{\abs{u}\leq 1}$\end{tabular}}
      &0 & 0.9615 & & &               0.9615 &        &       \\
      &1 & 0.5724 & 0.9163 & &        0.9615 & 0.9712 &       \\
      &2 & 0.4121 & 0.6387 & 0.8671 & 0.7400 & 0.7277 & 0.9267\\[1ex]
      \multirow{3}{*}{\begin{tabular}{l}Triangular\\ $(1-\abs{u})_{+}$\end{tabular}}
      &0 & 1 & & &                  1      &        &       \\
      &1 & 0.6274 & 0.9728& &       1      & 0.9943 &       \\
      &2 & 0.4652 & 0.6981 & 0.9254&0.8126 & 0.7814 & 0.9741\\[1ex]
      \multirow{3}{*}{\begin{tabular}{l}Epanechnikov\\ $\frac{3}{4}(1-u^{2})_{+}$\end{tabular}}
      &0 & 0.9959 &  &             &0.9959 &        &       \\
      &1 & 0.6087 & 0.9593 &       &0.9959 & 1      &       \\
      &2 & 0.4467 & 0.6813 & 0.9124&0.7902 & 0.7686 & 0.9672 \\[1ex]
    \end{tabular*}
    \begin{tablenotes}
    \item \emph{Notes:} Efficiency is relative to the optimal equivalent kernel
      $k^{*}_{SY}$. The functional $T(f)$ corresponds to the value of $f$ at a
      point.
    \end{tablenotes}
  \end{threeparttable}
\end{table}

\begin{table}[p]
  \centering
  \begin{threeparttable}
    \caption{Relative efficiency of local polynomial estimators for the function
      class $\FHol{p}(\SC)$.}\label{tab:holder-eff}
    \begin{tabular*}{0.9\linewidth}{@{\extracolsep{\fill}}@{}lclll lll@{}}
      & & \multicolumn{3}{c}{Boundary Point} &\multicolumn{3}{c}{Interior point}\\
      \cmidrule(rl){3-5}\cmidrule(rl){6-8}
      Kernel &  Order & $p=1$ & $p=2$ & $p=3$& $p=1$ & $p=2$ & $p=3$\\
      \midrule
      \multirow{3}{*}{\begin{tabular}{l}Uniform\\ $\1{\abs{u}\leq 1}$\end{tabular}}
      &0 & 0.9615 &        &        &0.9615 &        &       \\
      &1 & 0.7211 & 0.9711 &        &0.9615 & 0.9662 &       \\
      &2 & 0.5944 & 0.8372 & 0.9775 &0.8800 & 0.9162 & 0.9790\\[1ex]
      \multirow{3}{*}{\begin{tabular}{l}Triangular\\ $(1-\abs{u})_{+}$\end{tabular}}
      &0 & 1      &        &  &1\\
      &1 & 0.7600 & 0.9999 & & 1      &0.9892\\
      &2 & 0.6336 & 0.8691 &1& 0.9263 & 0.9487 & 1\\[1ex]
      \multirow{3}{*}{\begin{tabular}{l}Epanechnikov\\ $\frac{3}{4}(1-u^{2})_{+}$\end{tabular}}
      &0 & 0.9959 &        &      &0.9959 &        & \\
      &1 & 0.7471 & 0.9966 &      &0.9959 & 0.9949 & \\
      &2 & 0.6186 & 0.8602 &0.9974&0.9116 & 0.9425 & 1\\[1ex]
    \end{tabular*}
    \begin{tablenotes}
    \item \emph{Notes:} For $p=1,2$, efficiency is relative to the optimal
      kernel, for $p=3$, efficiency is relative to the local quadratic estimator
      with triangular kernel. The functional $T(f)$ corresponds to the value of
      $f$ at a point.
    \end{tablenotes}
  \end{threeparttable}
\end{table}

\begin{table}[p]
  \centering
  \begin{threeparttable}
    \caption{Gains from imposing global
      smoothness}\label{tab:global-smoothness-gains}
    \begin{tabular*}{0.9\linewidth}{@{\extracolsep{\fill}}@{}lllllll@{}}
      & \multicolumn{3}{c}{Boundary Point} &\multicolumn{3}{c}{Interior point}\\
      \cmidrule(rl){2-4}\cmidrule(rl){5-7}
      Kernel & $p=1$ & $p=2$ & $p=3$& $p=1$ & $p=2$ & $p=3$\\
      \midrule
      Uniform     &  1&   0.855& 0.764&   1&      1   & 0.848\\
      Triangular  &  1&   0.882& 0.797&   1&      1   & 0.873\\
      Epanechnikov&  1&   0.872& 0.788&   1&      1   & 0.866\\
      Optimal     &  1&   0.906&      &   1&    0.995 &\\[1ex]
    \end{tabular*}
    \begin{tablenotes}
    \item \emph{Notes:} The table gives the relative asymptotic risk of local
      polynomial estimators of order $p-1$ and a given kernel under the class
      $\FHol{p}(\SC)$ relative to the risk under $\FSY{p}(\SC)$ given
      in~\Cref{eq:global-smooothness}. ``Optimal'' refers to using the optimal
      kernel under a given smoothness class.
    \end{tablenotes}
  \end{threeparttable}
\end{table}

\begin{table}[p]
  \centering
  \begin{threeparttable}
    \caption{Performance of RBC CIs based on $\hrmse$ bandwidth for local linear
      regression under $\FSY{2}$ and $\FHol{2}$.}\label{tab:cct-minimax-mse}
    \begin{tabular*}{0.9\linewidth}{@{\extracolsep{\fill}}@{}llccc ccc@{}}
      &&\multicolumn{3}{c}{$\FSY{2}$} & \multicolumn{3}{@{}c@{}}{$\FHol{2}$} \\
      \cmidrule(lr){3-5}    \cmidrule(lr){6-8}
      \phantom{q}& Kernel & Length & Coverage & $t_{\text{RBC}}$
                                   & Length & Coverage & $t_{\text{RBC}}$ \\
      \cmidrule(l){1-8}
      \multicolumn{2}{l}{Boundary}\\
      \cmidrule(rl){1-2}
      &Uniform        &1.35&0.931& 0.400&1.35&0.948& 0.138\\
      &Triangular     &1.32&0.932& 0.391&1.32&0.947& 0.150\\
      &Epanechnikov   &1.33&0.932& 0.393&1.33&0.947& 0.148\\[1ex]
      \multicolumn{2}{l}{Interior}\\
      \cmidrule(rl){1-2}
      &Uniform        &1.35&0.941& 0.279&1.35&0.949& 0.086\\
      &Triangular     &1.27&0.940& 0.297&1.27&0.949& 0.110\\
      &Epanechnikov   &1.30&0.940& 0.298&1.30&0.949& 0.105
    \end{tabular*}
    \begin{tablenotes}
    \item \emph{Legend:} Length---CI length relative to 95\% FLCI based on a
      local linear estimator and the same kernel and bandwidth $\hrmse$;
      $t_{\text{RBC}}$---ratio of the worst-case bias to standard deviation;
    \end{tablenotes}
  \end{threeparttable}
\end{table}

\begin{landscape}
  \renewcommand{\arraystretch}{1.05} %
  \footnotesize
  \begin{ThreePartTable}
    \begin{TableNotes}
    \item \emph{Legend:} SE---average standard error; $E[h]$---average (over
      Monte Carlo draws) bandwidth; Cov---coverage of CIs (in \%); RL---relative
      (to optimal FLCI) length.
    \item \emph{Bandwidth descriptions:} $\hmsedpi$---plugin estimate of
      pointwise MSE optimal bandwidth (bw); $\bmsedpi$---analog for estimate of
      the bias; $\hcedpi$---plugin estimate of coverage error optimal bw;
      $\bcedpi$---analog for estimate of the bias; The implementation of
      \citet{ccf15} is used for all four bws. $\hhrmse{2}$, $\hhrmse{6}$---RMSE
      optimal bw, assuming $\SC=2$, and $\SC=6$, respectively.
      $\hfgrot$---\citet{fg96} rule of thumb; $\hhrmse{\Mrot}$---RMSE optimal
      bw, using rule-of-thumb for $\SC$. 50,000 Monte Carlo draws.
    \end{TableNotes}
    \begin{longtable}{@{}lll rllrl c rllrl ll@{}}
      \caption{Monte Carlo simulation: Inference at a point.}\label{tab:mc}\\
      &&&\multicolumn{5}{c@{}}{$\SC=2$}&&\multicolumn{5}{c}{$\SC=6$}\\
      \cmidrule(rl){4-8}\cmidrule(rl){10-14}
      &Method & Bandwidth & Bias& SE &   $E[h]$&   Cov&   RL&&Bias& SE &   $E_{m}[h]$&   Cov&   RL\\
      \midrule
      \endfirsthead%
      \caption*{Monte Carlo simulation: baseline DGP (continued)}\\
      &&&\multicolumn{5}{c@{}}{$\SC=2$}&&\multicolumn{5}{c}{$\SC=6$}\\
      \cmidrule(rl){4-8}\cmidrule(rl){10-14}
      &Method & Bandwidth & Bias& SE &   $E[h]$&   Cov&   RL&&Bias& SE &   $E_{m}[h]$&   Cov&   RL\\
      \midrule
      \endhead%
      \endfoot%
      \insertTableNotes%
      \endlastfoot%

      \multicolumn{2}{@{}l}{Design 1}\\
      \cmidrule(r){1-2} \phantom{a}
      &RBC         &$h=\hmsedpi$, $b=\bmsedpi$  & 0.063& 0.035& 0.75& 55.6& 0.73&&  0.157& 0.036& 0.62&  0.1& 0.61\\
      &RBC         &$h=\hcedpi$, $b=\bcedpi$    & 0.030& 0.041& 0.45& 85.8& 0.85&&  0.059& 0.045& 0.34& 72.4& 0.76\\
      &RBC         &$h=b=\hmsedpi$              & 0.025& 0.042& 0.75& 93.1& 0.88&&  0.042& 0.047& 0.62& 89.1& 0.78\\
      &RBC         &$h=b=\hhrmse{2}$            & 0.001& 0.061& 0.36& 94.5& 1.27&&  0.002& 0.061& 0.36& 94.5& 1.01\\
      &RBC         &$h=b=\hhrmse{6}$            & 0.000& 0.076& 0.23& 94.2& 1.58&&  0.000& 0.075& 0.23& 94.2& 1.26\\
      &RBC         &$h=b=\hhrmse{\Mrot}$        & 0.000& 0.078& 0.22& 93.9& 1.64&&  0.000& 0.097& 0.14& 93.4& 1.63\\
      &Conventional&$\hfgrot$                   & 0.032& 0.036& 0.56& 76.6& 0.76&&  0.049& 0.046& 0.31& 77.4& 0.77\\
      &FLCI, $\SC=2$ &$\hhrmse{2}$              & 0.021& 0.043& 0.36& 94.9& 1.00&&  0.065& 0.043& 0.36& 75.2& 0.80\\
      &FLCI, $\SC=6$ &$\hhrmse{6}$              & 0.009& 0.054& 0.23& 96.6& 1.25&&  0.028& 0.053& 0.23& 94.7& 1.00\\
      &FLCI, $\SC=\Mrot$&$\hhrmse{\Mrot}$       & 0.008& 0.056& 0.22& 95.6& 1.29&&  0.010& 0.069& 0.14& 96.3& 1.30\\
      \multicolumn{2}{@{}l}{Design 2}\\
      \cmidrule(r){1-2}
      &RBC         &$h=\hmsedpi$, $b=\bmsedpi$  & 0.043& 0.035& 0.77& 75.9& 0.72&&  0.129& 0.035& 0.77&  4.6& 0.58\\
      &RBC         &$h=\hcedpi$, $b=\bcedpi$    & 0.028& 0.040& 0.49& 87.4& 0.83&&  0.074& 0.041& 0.44& 54.1& 0.69\\
      &RBC         &$h=b=\hmsedpi$              & 0.026& 0.041& 0.77& 90.9& 0.87&&  0.077& 0.042& 0.77& 53.0& 0.70\\
      &RBC         &$h=b=\hhrmse{2}$            & 0.002& 0.061& 0.36& 94.5& 1.27&&  0.006& 0.061& 0.36& 94.4& 1.01\\
      &RBC         &$h=b=\hhrmse{6}$            & 0.000& 0.076& 0.23& 94.2& 1.58&&  0.000& 0.075& 0.23& 94.2& 1.26\\
      &RBC         &$h=b=\hhrmse{\Mrot}$        & 0.001& 0.068& 0.30& 94.0& 1.43&&  0.000& 0.083& 0.20& 93.8& 1.38\\
      &Conventional&$\hfgrot$                   & 0.032& 0.032& 0.78& 74.4& 0.67&&  0.073& 0.040& 0.44& 53.0& 0.66\\
      &FLCI, $\SC=2$ &$\hhrmse{2}$              & 0.020& 0.043& 0.36& 95.1& 1.00&&  0.061& 0.043& 0.36& 78.1& 0.80\\
      &FLCI, $\SC=6$ &$\hhrmse{6}$              & 0.009& 0.054& 0.23& 96.6& 1.25&&  0.028& 0.053& 0.23& 94.7& 1.00\\
      &FLCI, $\SC=\Mrot$&$\hhrmse{\Mrot}$       & 0.013& 0.048& 0.30& 94.3& 1.13&&  0.020& 0.059& 0.20& 94.3& 1.10\\
      \multicolumn{2}{@{}l}{Design 3}\\
      \cmidrule(r){1-2}
      &RBC         &$h=\hmsedpi$, $b=\bmsedpi$  & -0.043& 0.035& 0.77& 75.7& 0.72&& -0.123& 0.035& 0.74&  9.9& 0.59\\
      &RBC         &$h=\hcedpi$, $b=\bcedpi$    & -0.026& 0.040& 0.49& 88.1& 0.83&& -0.063& 0.043& 0.43& 64.2& 0.71\\
      &RBC         &$h=b=\hmsedpi$              & -0.024& 0.042& 0.77& 90.8& 0.87&& -0.066& 0.043& 0.74& 60.3& 0.71\\
      &RBC         &$h=b=\hhrmse{2}$            & -0.002& 0.061& 0.36& 94.5& 1.27&& -0.007& 0.061& 0.36& 94.4& 1.01\\
      &RBC         &$h=b=\hhrmse{6}$            &  0.000& 0.076& 0.23& 94.2& 1.58&&  0.000& 0.075& 0.23& 94.2& 1.26\\
      &RBC         &$h=b=\hhrmse{\Mrot}$        &  0.000& 0.074& 0.25& 94.2& 1.54&&  0.000& 0.092& 0.16& 93.6& 1.54\\
      &Conventional&$\hfgrot$                   & -0.032& 0.033& 0.72& 74.7& 0.69&& -0.065& 0.042& 0.39& 62.0& 0.70\\
      &FLCI, $\SC=2$ &$\hhrmse{2}$              & -0.020& 0.043& 0.36& 95.0& 1.00&& -0.060& 0.043& 0.36& 78.1& 0.80\\
      &FLCI, $\SC=6$ &$\hhrmse{6}$              & -0.009& 0.054& 0.23& 96.5& 1.25&& -0.027& 0.053& 0.23& 94.7& 1.00\\
      &FLCI, $\SC=\Mrot$&$\hhrmse{\Mrot}$       & -0.010& 0.052& 0.25& 95.6& 1.22&& -0.013& 0.065& 0.16& 96.1& 1.22\\
    \end{longtable}
  \end{ThreePartTable}
\end{landscape}

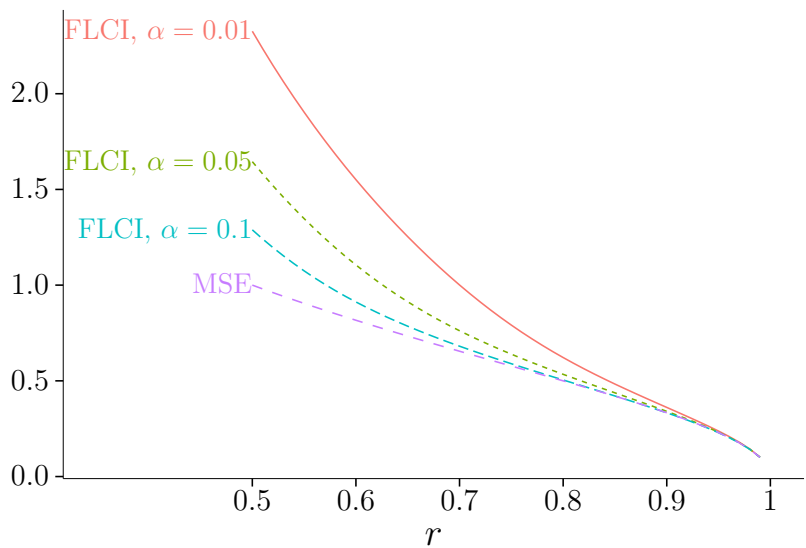
\begin{figure}[p]
  \centering%
\begin{tikzpicture}[x=1pt,y=1pt]
\definecolor{fillColor}{RGB}{255,255,255}
\path[use as bounding box,fill=fillColor,fill opacity=0.00] (0,0) rectangle (325.21,216.81);
\begin{scope}
\path[clip] (  0.00,  0.00) rectangle (325.21,216.81);
\definecolor{drawColor}{RGB}{255,255,255}
\definecolor{fillColor}{RGB}{255,255,255}

\path[draw=drawColor,line width= 0.6pt,line join=round,line cap=round,fill=fillColor] (  0.00,  0.00) rectangle (325.21,216.81);
\end{scope}
\begin{scope}
\path[clip] ( 34.16, 30.69) rectangle (319.71,211.31);
\definecolor{fillColor}{RGB}{255,255,255}

\path[fill=fillColor] ( 34.16, 30.69) rectangle (319.71,211.31);
\definecolor{drawColor}{RGB}{0,0,0}

\path[draw=drawColor,line width= 0.6pt,line join=round] (107.04,203.10) --
	(111.04,196.37) --
	(115.03,189.90) --
	(119.02,183.67) --
	(123.02,177.68) --
	(127.01,171.90) --
	(131.01,166.33) --
	(135.00,160.95) --
	(138.99,155.76) --
	(142.99,150.74) --
	(146.98,145.90) --
	(150.98,141.21) --
	(154.97,136.68) --
	(158.96,132.29) --
	(162.96,128.04) --
	(166.95,123.93) --
	(170.94,119.95) --
	(174.94,116.10) --
	(178.93,112.38) --
	(182.93,108.78) --
	(186.92,105.31) --
	(190.91,101.97) --
	(194.91, 98.75) --
	(198.90, 95.66) --
	(202.90, 92.69) --
	(206.89, 89.84) --
	(210.88, 87.12) --
	(214.88, 84.52) --
	(218.87, 82.02) --
	(222.86, 79.64) --
	(226.86, 77.35) --
	(230.85, 75.15) --
	(234.85, 73.04) --
	(238.84, 71.01) --
	(242.83, 69.04) --
	(246.83, 67.13) --
	(250.82, 65.27) --
	(254.82, 63.44) --
	(258.81, 61.65) --
	(262.80, 59.89) --
	(266.80, 58.13) --
	(270.79, 56.37) --
	(274.78, 54.59) --
	(278.78, 52.79) --
	(282.77, 50.94) --
	(286.77, 49.01) --
	(290.76, 46.96) --
	(294.75, 44.72) --
	(298.75, 42.16) --
	(302.74, 38.95);

\path[draw=drawColor,line width= 0.6pt,dash pattern=on 2pt off 2pt ,line join=round] (107.04,152.84) --
	(111.04,148.09) --
	(115.03,143.52) --
	(119.02,139.14) --
	(123.02,134.93) --
	(127.01,130.88) --
	(131.01,127.00) --
	(135.00,123.28) --
	(138.99,119.71) --
	(142.99,116.29) --
	(146.98,113.03) --
	(150.98,109.91) --
	(154.97,106.94) --
	(158.96,104.10) --
	(162.96,101.41) --
	(166.95, 98.84) --
	(170.94, 96.39) --
	(174.94, 94.06) --
	(178.93, 91.84) --
	(182.93, 89.71) --
	(186.92, 87.68) --
	(190.91, 85.73) --
	(194.91, 83.86) --
	(198.90, 82.06) --
	(202.90, 80.32) --
	(206.89, 78.64) --
	(210.88, 77.01) --
	(214.88, 75.42) --
	(218.87, 73.88) --
	(222.86, 72.36) --
	(226.86, 70.88) --
	(230.85, 69.42) --
	(234.85, 67.98) --
	(238.84, 66.56) --
	(242.83, 65.15) --
	(246.83, 63.74) --
	(250.82, 62.33) --
	(254.82, 60.93) --
	(258.81, 59.51) --
	(262.80, 58.07) --
	(266.80, 56.61) --
	(270.79, 55.12) --
	(274.78, 53.59) --
	(278.78, 52.00) --
	(282.77, 50.33) --
	(286.77, 48.56) --
	(290.76, 46.64) --
	(294.75, 44.52) --
	(298.75, 42.06) --
	(302.74, 38.91);

\path[draw=drawColor,line width= 0.6pt,dash pattern=on 4pt off 2pt ,line join=round] (107.04,126.55) --
	(111.04,123.04) --
	(115.03,119.70) --
	(119.02,116.55) --
	(123.02,113.56) --
	(127.01,110.75) --
	(131.01,108.08) --
	(135.00,105.56) --
	(138.99,103.17) --
	(142.99,100.91) --
	(146.98, 98.77) --
	(150.98, 96.72) --
	(154.97, 94.77) --
	(158.96, 92.91) --
	(162.96, 91.13) --
	(166.95, 89.42) --
	(170.94, 87.78) --
	(174.94, 86.19) --
	(178.93, 84.65) --
	(182.93, 83.17) --
	(186.92, 81.72) --
	(190.91, 80.31) --
	(194.91, 78.93) --
	(198.90, 77.58) --
	(202.90, 76.26) --
	(206.89, 74.97) --
	(210.88, 73.69) --
	(214.88, 72.42) --
	(218.87, 71.17) --
	(222.86, 69.93) --
	(226.86, 68.70) --
	(230.85, 67.47) --
	(234.85, 66.25) --
	(238.84, 65.02) --
	(242.83, 63.78) --
	(246.83, 62.54) --
	(250.82, 61.29) --
	(254.82, 60.01) --
	(258.81, 58.72) --
	(262.80, 57.40) --
	(266.80, 56.05) --
	(270.79, 54.65) --
	(274.78, 53.20) --
	(278.78, 51.69) --
	(282.77, 50.09) --
	(286.77, 48.38) --
	(290.76, 46.52) --
	(294.75, 44.44) --
	(298.75, 42.01) --
	(302.74, 38.90);

\path[draw=drawColor,line width= 0.6pt,dash pattern=on 4pt off 4pt ,line join=round] (107.04,105.26) --
	(111.04,103.79) --
	(115.03,102.36) --
	(119.02,100.95) --
	(123.02, 99.57) --
	(127.01, 98.21) --
	(131.01, 96.88) --
	(135.00, 95.56) --
	(138.99, 94.26) --
	(142.99, 92.98) --
	(146.98, 91.72) --
	(150.98, 90.47) --
	(154.97, 89.24) --
	(158.96, 88.02) --
	(162.96, 86.81) --
	(166.95, 85.62) --
	(170.94, 84.43) --
	(174.94, 83.26) --
	(178.93, 82.09) --
	(182.93, 80.93) --
	(186.92, 79.78) --
	(190.91, 78.63) --
	(194.91, 77.49) --
	(198.90, 76.35) --
	(202.90, 75.21) --
	(206.89, 74.08) --
	(210.88, 72.94) --
	(214.88, 71.80) --
	(218.87, 70.66) --
	(222.86, 69.52) --
	(226.86, 68.37) --
	(230.85, 67.21) --
	(234.85, 66.05) --
	(238.84, 64.87) --
	(242.83, 63.68) --
	(246.83, 62.47) --
	(250.82, 61.25) --
	(254.82, 60.00) --
	(258.81, 58.73) --
	(262.80, 57.42) --
	(266.80, 56.08) --
	(270.79, 54.68) --
	(274.78, 53.24) --
	(278.78, 51.72) --
	(282.77, 50.12) --
	(286.77, 48.41) --
	(290.76, 46.54) --
	(294.75, 44.46) --
	(298.75, 42.02) --
	(302.74, 38.90);
\end{scope}
\begin{scope}
\path[clip] ( 34.16, 30.69) rectangle (319.71,211.31);
\definecolor{drawColor}{RGB}{0,0,0}

\node[text=drawColor,anchor=base east,inner sep=0pt, outer sep=0pt, scale=  0.90] at (107.04,102.16) {MSE};

\node[text=drawColor,anchor=base east,inner sep=0pt, outer sep=0pt, scale=  0.90] at (107.04,123.45) {FLCI, $\alpha=0.1$};

\node[text=drawColor,anchor=base east,inner sep=0pt, outer sep=0pt, scale=  0.90] at (107.04,149.74) {FLCI, $\alpha=0.05$};

\node[text=drawColor,anchor=base east,inner sep=0pt, outer sep=0pt, scale=  0.90] at (107.04,200.00) {FLCI, $\alpha=0.01$};
\end{scope}
\begin{scope}
\path[clip] (  0.00,  0.00) rectangle (325.21,216.81);
\definecolor{drawColor}{RGB}{0,0,0}

\path[draw=drawColor,line width= 0.2pt,line join=round] ( 34.16, 30.69) --
	( 34.16,211.31);
\end{scope}
\begin{scope}
\path[clip] (  0.00,  0.00) rectangle (325.21,216.81);
\definecolor{drawColor}{gray}{0.30}

\node[text=drawColor,anchor=base east,inner sep=0pt, outer sep=0pt, scale=  0.88] at ( 29.21, 28.45) {0.0};

\node[text=drawColor,anchor=base east,inner sep=0pt, outer sep=0pt, scale=  0.88] at ( 29.21, 65.34) {0.5};

\node[text=drawColor,anchor=base east,inner sep=0pt, outer sep=0pt, scale=  0.88] at ( 29.21,102.22) {1.0};

\node[text=drawColor,anchor=base east,inner sep=0pt, outer sep=0pt, scale=  0.88] at ( 29.21,139.11) {1.5};

\node[text=drawColor,anchor=base east,inner sep=0pt, outer sep=0pt, scale=  0.88] at ( 29.21,176.00) {2.0};
\end{scope}
\begin{scope}
\path[clip] (  0.00,  0.00) rectangle (325.21,216.81);
\definecolor{drawColor}{gray}{0.20}

\path[draw=drawColor,line width= 0.6pt,line join=round] ( 31.41, 31.49) --
	( 34.16, 31.49);

\path[draw=drawColor,line width= 0.6pt,line join=round] ( 31.41, 68.37) --
	( 34.16, 68.37);

\path[draw=drawColor,line width= 0.6pt,line join=round] ( 31.41,105.26) --
	( 34.16,105.26);

\path[draw=drawColor,line width= 0.6pt,line join=round] ( 31.41,142.14) --
	( 34.16,142.14);

\path[draw=drawColor,line width= 0.6pt,line join=round] ( 31.41,179.03) --
	( 34.16,179.03);
\end{scope}
\begin{scope}
\path[clip] (  0.00,  0.00) rectangle (325.21,216.81);
\definecolor{drawColor}{RGB}{0,0,0}

\path[draw=drawColor,line width= 0.2pt,line join=round] ( 34.16, 30.69) --
	(319.71, 30.69);
\end{scope}
\begin{scope}
\path[clip] (  0.00,  0.00) rectangle (325.21,216.81);
\definecolor{drawColor}{gray}{0.20}

\path[draw=drawColor,line width= 0.6pt,line join=round] (107.04, 27.94) --
	(107.04, 30.69);

\path[draw=drawColor,line width= 0.6pt,line join=round] (146.98, 27.94) --
	(146.98, 30.69);

\path[draw=drawColor,line width= 0.6pt,line join=round] (186.92, 27.94) --
	(186.92, 30.69);

\path[draw=drawColor,line width= 0.6pt,line join=round] (226.86, 27.94) --
	(226.86, 30.69);

\path[draw=drawColor,line width= 0.6pt,line join=round] (266.80, 27.94) --
	(266.80, 30.69);

\path[draw=drawColor,line width= 0.6pt,line join=round] (306.74, 27.94) --
	(306.74, 30.69);
\end{scope}
\begin{scope}
\path[clip] (  0.00,  0.00) rectangle (325.21,216.81);
\definecolor{drawColor}{gray}{0.30}

\node[text=drawColor,anchor=base,inner sep=0pt, outer sep=0pt, scale=  0.88] at (107.04, 19.68) {0.5};

\node[text=drawColor,anchor=base,inner sep=0pt, outer sep=0pt, scale=  0.88] at (146.98, 19.68) {0.6};

\node[text=drawColor,anchor=base,inner sep=0pt, outer sep=0pt, scale=  0.88] at (186.92, 19.68) {0.7};

\node[text=drawColor,anchor=base,inner sep=0pt, outer sep=0pt, scale=  0.88] at (226.86, 19.68) {0.8};

\node[text=drawColor,anchor=base,inner sep=0pt, outer sep=0pt, scale=  0.88] at (266.80, 19.68) {0.9};

\node[text=drawColor,anchor=base,inner sep=0pt, outer sep=0pt, scale=  0.88] at (306.74, 19.68) {1};
\end{scope}
\begin{scope}
\path[clip] (  0.00,  0.00) rectangle (325.21,216.81);
\definecolor{drawColor}{RGB}{0,0,0}

\node[text=drawColor,anchor=base,inner sep=0pt, outer sep=0pt, scale=  1.10] at (176.94,  7.64) {$r$};
\end{scope}
\end{tikzpicture}
  \caption{Optimal ratio of the worst-case bias to standard deviation for fixed
    length CIs (FLCI), and maximum MSE (MSE) performance
    criteria.}\label{fig:bias-variance-flci}
\end{figure}

\begin{figure}[p]
  \centering%
  \input{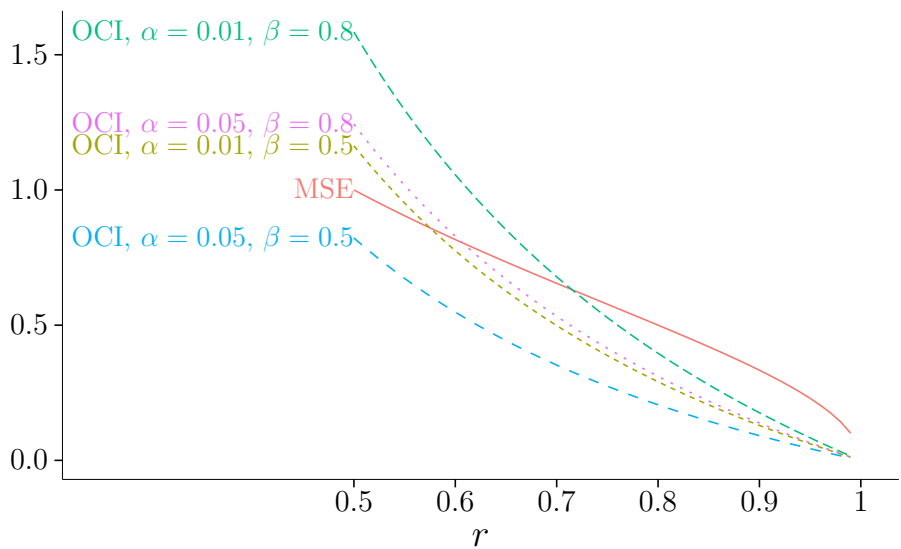}
  \caption{Optimal ratio of the worst-case bias to standard deviation for
    one-sided CIs (OCI), and maximum MSE (MSE) performance
    criteria.}\label{fig:bias-variance-oci}
\end{figure}

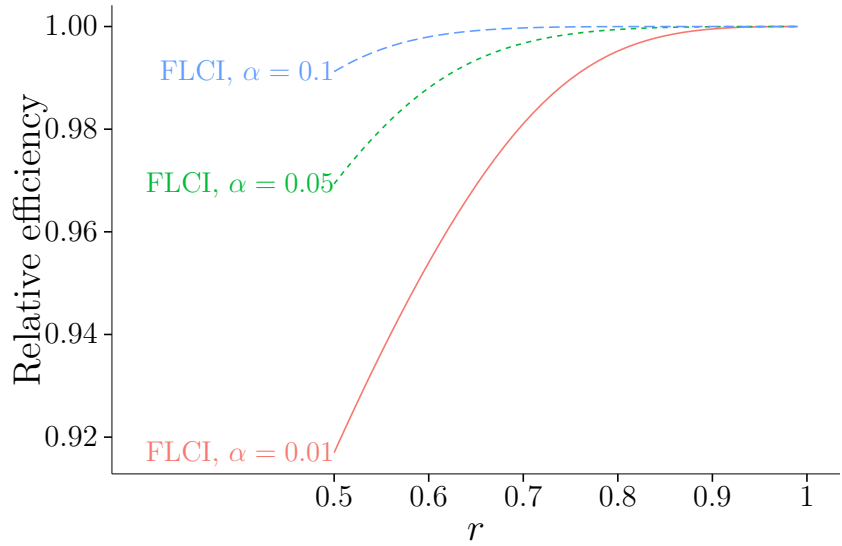
\begin{figure}[p]
  \centering%
\begin{tikzpicture}[x=1pt,y=1pt]
\definecolor{fillColor}{RGB}{255,255,255}
\path[use as bounding box,fill=fillColor,fill opacity=0.00] (0,0) rectangle (325.21,216.81);
\begin{scope}
\path[clip] (  0.00,  0.00) rectangle (325.21,216.81);
\definecolor{drawColor}{RGB}{255,255,255}
\definecolor{fillColor}{RGB}{255,255,255}

\path[draw=drawColor,line width= 0.6pt,line join=round,line cap=round,fill=fillColor] (  0.00,  0.00) rectangle (325.21,216.81);
\end{scope}
\begin{scope}
\path[clip] ( 38.56, 30.69) rectangle (319.71,211.31);
\definecolor{fillColor}{RGB}{255,255,255}

\path[fill=fillColor] ( 38.56, 30.69) rectangle (319.71,211.31);
\definecolor{drawColor}{RGB}{0,0,0}

\path[draw=drawColor,line width= 0.6pt,line join=round] (124.36, 38.90) --
	(128.02, 46.89) --
	(131.67, 54.75) --
	(135.32, 62.48) --
	(138.97, 70.06) --
	(142.62, 77.49) --
	(146.27, 84.77) --
	(149.92, 91.88) --
	(153.58, 98.82) --
	(157.23,105.57) --
	(160.88,112.14) --
	(164.53,118.51) --
	(168.18,124.68) --
	(171.83,130.63) --
	(175.48,136.36) --
	(179.14,141.86) --
	(182.79,147.13) --
	(186.44,152.14) --
	(190.09,156.91) --
	(193.74,161.41) --
	(197.39,165.65) --
	(201.04,169.62) --
	(204.70,173.33) --
	(208.35,176.76) --
	(212.00,179.92) --
	(215.65,182.82) --
	(219.30,185.46) --
	(222.95,187.86) --
	(226.60,190.01) --
	(230.26,191.94) --
	(233.91,193.66) --
	(237.56,195.18) --
	(241.21,196.51) --
	(244.86,197.67) --
	(248.51,198.68) --
	(252.16,199.54) --
	(255.82,200.27) --
	(259.47,200.89) --
	(263.12,201.40) --
	(266.77,201.82) --
	(270.42,202.16) --
	(274.07,202.43) --
	(277.72,202.64) --
	(281.38,202.80) --
	(285.03,202.92) --
	(288.68,203.00) --
	(292.33,203.05) --
	(295.98,203.08) --
	(299.63,203.09) --
	(303.28,203.10);

\path[draw=drawColor,line width= 0.6pt,dash pattern=on 2pt off 2pt ,line join=round] (124.36,142.47) --
	(128.02,147.09) --
	(131.67,151.52) --
	(135.32,155.75) --
	(138.97,159.78) --
	(142.62,163.59) --
	(146.27,167.20) --
	(149.92,170.59) --
	(153.58,173.76) --
	(157.23,176.72) --
	(160.88,179.47) --
	(164.53,182.00) --
	(168.18,184.33) --
	(171.83,186.47) --
	(175.48,188.41) --
	(179.14,190.18) --
	(182.79,191.77) --
	(186.44,193.20) --
	(190.09,194.49) --
	(193.74,195.63) --
	(197.39,196.65) --
	(201.04,197.55) --
	(204.70,198.35) --
	(208.35,199.05) --
	(212.00,199.67) --
	(215.65,200.20) --
	(219.30,200.67) --
	(222.95,201.07) --
	(226.60,201.42) --
	(230.26,201.71) --
	(233.91,201.97) --
	(237.56,202.18) --
	(241.21,202.36) --
	(244.86,202.51) --
	(248.51,202.64) --
	(252.16,202.74) --
	(255.82,202.82) --
	(259.47,202.89) --
	(263.12,202.95) --
	(266.77,202.99) --
	(270.42,203.02) --
	(274.07,203.05) --
	(277.72,203.06) --
	(281.38,203.08) --
	(285.03,203.09) --
	(288.68,203.09) --
	(292.33,203.10) --
	(295.98,203.10) --
	(299.63,203.10) --
	(303.28,203.10);

\path[draw=drawColor,line width= 0.6pt,dash pattern=on 4pt off 2pt ,line join=round] (124.36,185.73) --
	(128.02,187.81) --
	(131.67,189.70) --
	(135.32,191.41) --
	(138.97,192.95) --
	(142.62,194.32) --
	(146.27,195.54) --
	(149.92,196.62) --
	(153.58,197.57) --
	(157.23,198.40) --
	(160.88,199.13) --
	(164.53,199.76) --
	(168.18,200.30) --
	(171.83,200.77) --
	(175.48,201.17) --
	(179.14,201.51) --
	(182.79,201.80) --
	(186.44,202.05) --
	(190.09,202.25) --
	(193.74,202.42) --
	(197.39,202.56) --
	(201.04,202.68) --
	(204.70,202.77) --
	(208.35,202.85) --
	(212.00,202.91) --
	(215.65,202.96) --
	(219.30,203.00) --
	(222.95,203.02) --
	(226.60,203.05) --
	(230.26,203.06) --
	(233.91,203.08) --
	(237.56,203.08) --
	(241.21,203.09) --
	(244.86,203.09) --
	(248.51,203.10) --
	(252.16,203.10) --
	(255.82,203.10) --
	(259.47,203.10) --
	(263.12,203.10) --
	(266.77,203.10) --
	(270.42,203.10) --
	(274.07,203.10) --
	(277.72,203.10) --
	(281.38,203.10) --
	(285.03,203.10) --
	(288.68,203.10) --
	(292.33,203.10) --
	(295.98,203.10) --
	(299.63,203.10) --
	(303.28,203.10);
\end{scope}
\begin{scope}
\path[clip] ( 38.56, 30.69) rectangle (319.71,211.31);
\definecolor{drawColor}{RGB}{0,0,0}

\node[text=drawColor,anchor=base east,inner sep=0pt, outer sep=0pt, scale=  0.90] at (124.36, 35.80) {FLCI, $\alpha=0.01$};

\node[text=drawColor,anchor=base east,inner sep=0pt, outer sep=0pt, scale=  0.90] at (124.36,139.37) {FLCI, $\alpha=0.05$};

\node[text=drawColor,anchor=base east,inner sep=0pt, outer sep=0pt, scale=  0.90] at (124.36,182.63) {FLCI, $\alpha=0.1$};
\end{scope}
\begin{scope}
\path[clip] (  0.00,  0.00) rectangle (325.21,216.81);
\definecolor{drawColor}{RGB}{0,0,0}

\path[draw=drawColor,line width= 0.2pt,line join=round] ( 38.56, 30.69) --
	( 38.56,211.31);
\end{scope}
\begin{scope}
\path[clip] (  0.00,  0.00) rectangle (325.21,216.81);
\definecolor{drawColor}{gray}{0.30}

\node[text=drawColor,anchor=base east,inner sep=0pt, outer sep=0pt, scale=  0.88] at ( 33.61, 41.83) {0.92};

\node[text=drawColor,anchor=base east,inner sep=0pt, outer sep=0pt, scale=  0.88] at ( 33.61, 81.39) {0.94};

\node[text=drawColor,anchor=base east,inner sep=0pt, outer sep=0pt, scale=  0.88] at ( 33.61,120.95) {0.96};

\node[text=drawColor,anchor=base east,inner sep=0pt, outer sep=0pt, scale=  0.88] at ( 33.61,160.51) {0.98};

\node[text=drawColor,anchor=base east,inner sep=0pt, outer sep=0pt, scale=  0.88] at ( 33.61,200.07) {1.00};
\end{scope}
\begin{scope}
\path[clip] (  0.00,  0.00) rectangle (325.21,216.81);
\definecolor{drawColor}{gray}{0.20}

\path[draw=drawColor,line width= 0.6pt,line join=round] ( 35.81, 44.86) --
	( 38.56, 44.86);

\path[draw=drawColor,line width= 0.6pt,line join=round] ( 35.81, 84.42) --
	( 38.56, 84.42);

\path[draw=drawColor,line width= 0.6pt,line join=round] ( 35.81,123.98) --
	( 38.56,123.98);

\path[draw=drawColor,line width= 0.6pt,line join=round] ( 35.81,163.54) --
	( 38.56,163.54);

\path[draw=drawColor,line width= 0.6pt,line join=round] ( 35.81,203.10) --
	( 38.56,203.10);
\end{scope}
\begin{scope}
\path[clip] (  0.00,  0.00) rectangle (325.21,216.81);
\definecolor{drawColor}{RGB}{0,0,0}

\path[draw=drawColor,line width= 0.2pt,line join=round] ( 38.56, 30.69) --
	(319.71, 30.69);
\end{scope}
\begin{scope}
\path[clip] (  0.00,  0.00) rectangle (325.21,216.81);
\definecolor{drawColor}{gray}{0.20}

\path[draw=drawColor,line width= 0.6pt,line join=round] (124.36, 27.94) --
	(124.36, 30.69);

\path[draw=drawColor,line width= 0.6pt,line join=round] (160.88, 27.94) --
	(160.88, 30.69);

\path[draw=drawColor,line width= 0.6pt,line join=round] (197.39, 27.94) --
	(197.39, 30.69);

\path[draw=drawColor,line width= 0.6pt,line join=round] (233.91, 27.94) --
	(233.91, 30.69);

\path[draw=drawColor,line width= 0.6pt,line join=round] (270.42, 27.94) --
	(270.42, 30.69);

\path[draw=drawColor,line width= 0.6pt,line join=round] (306.94, 27.94) --
	(306.94, 30.69);
\end{scope}
\begin{scope}
\path[clip] (  0.00,  0.00) rectangle (325.21,216.81);
\definecolor{drawColor}{gray}{0.30}

\node[text=drawColor,anchor=base,inner sep=0pt, outer sep=0pt, scale=  0.88] at (124.36, 19.68) {0.5};

\node[text=drawColor,anchor=base,inner sep=0pt, outer sep=0pt, scale=  0.88] at (160.88, 19.68) {0.6};

\node[text=drawColor,anchor=base,inner sep=0pt, outer sep=0pt, scale=  0.88] at (197.39, 19.68) {0.7};

\node[text=drawColor,anchor=base,inner sep=0pt, outer sep=0pt, scale=  0.88] at (233.91, 19.68) {0.8};

\node[text=drawColor,anchor=base,inner sep=0pt, outer sep=0pt, scale=  0.88] at (270.42, 19.68) {0.9};

\node[text=drawColor,anchor=base,inner sep=0pt, outer sep=0pt, scale=  0.88] at (306.94, 19.68) {1};
\end{scope}
\begin{scope}
\path[clip] (  0.00,  0.00) rectangle (325.21,216.81);
\definecolor{drawColor}{RGB}{0,0,0}

\node[text=drawColor,anchor=base,inner sep=0pt, outer sep=0pt, scale=  1.10] at (179.14,  7.64) {$r$};
\end{scope}
\begin{scope}
\path[clip] (  0.00,  0.00) rectangle (325.21,216.81);
\definecolor{drawColor}{RGB}{0,0,0}

\node[text=drawColor,rotate= 90.00,anchor=base,inner sep=0pt, outer sep=0pt, scale=  1.10] at ( 13.08,121.00) {Relative efficiency};
\end{scope}
\end{tikzpicture}
  \caption{Efficiency of fixed-length CIs based on minimax MSE bandwidth
    relative to fixed-length CIs based on optimal
    bandwidth.}\label{fig:eff-mse-flci}
\end{figure}

\begin{figure}[p]
  \centering
  \input{./figure_4.tex}
  \caption{Monte Carlo simulation designs 1--3, and
    $\SC=2$.}\label{fig:mc-fkt}
\end{figure}

\begin{figure}[p]
  \centering \input{./figure_5.tex}
  \caption{Average county mortality rate per 100,000 for children aged 5--9 over
    1973--83 due to causes addressed as part of Head Start's health services
    (labeled ``Mortality rate'') plotted against poverty rate in 1960 relative
    to the 300th poorest county. Each point corresponds to an average for 25
    counties. Data are from \citet{LuMi07}.}\label{fig:lm-rawdata}
\end{figure}

\end{document}


\maketitle

\renewcommand{\theequation}{S\arabic{equation}}
\renewcommand{\thetable}{S\arabic{table}}
\renewcommand{\thefigure}{S\arabic{figure}}

\begin{appendices}
\crefalias{section}{sappsec}
\crefalias{subsection}{sappsubsec}
\crefalias{subsubsection}{sappsubsubsec}
\setcounter{section}{1}

These supplemental materials contain further appendices and additional tables
and figures. \Cref{verification_sec} verifies our regularity conditions
for some examples, and includes proofs of the results in
\begin{NoHyper}\Cref{sec:theoretical-results}\end{NoHyper}. \Cref{sec:addit-appl}
discusses two additional applications: estimation of density at a point, and
estimating a bidder valuation in first price auctions.
\Cref{sec:addit-deta-appl} contains additional details for the applications in
\begin{NoHyper}\Cref{sec:applications}\end{NoHyper}.
\Cref{data-driven_bw_sec_append} presents a formal analysis of the rule-of-thumb
choice of $\SC$ proposed in
\begin{NoHyper}\Cref{sec:pract-impl}\end{NoHyper}. Finally,
\Cref{sec:addit-monte-carlo} contains additional Monte Carlo results.

\section{Verification of regularity conditions}\label{verification_sec}

We verify the main
condition \begin{NoHyper}\eqref{performance_approx_eq}\end{NoHyper} in some
applications. \Cref{sufficient_conditions_sec} gives sufficient
conditions for \begin{NoHyper}\eqref{performance_approx_eq}\end{NoHyper} which
do not require convergence of moments. \Cref{white_noise_sec} shows that
\begin{NoHyper}\eqref{performance_approx_eq}\end{NoHyper} holds
in the Gaussian white noise model under a mild extension of conditions in
\citet{DoLo92}. Thus, the results apply to estimating, among other things, a
function or one of its derivatives evaluated at a given point, when the function
is observed in the white noise model. By equivalence results in \citet{BrLo96}
and \citet{nussbaum_asymptotic_1996}, our results also apply
when the function of interest is a density or conditional mean.
\Cref{sec:local-polyn-estim}
verifies \begin{NoHyper}\eqref{performance_approx_eq}\end{NoHyper} directly for
local polynomial estimators in the nonparametric regression setting, and
\Cref{sec:fuzzy-rd-suppl} verifies it for in the fuzzy RD application.

\subsection{Sufficient conditions for main regularity condition}\label{sufficient_conditions_sec}

This \namecref{sufficient_conditions_sec} gives sufficient conditions for the
main condition
\begin{NoHyper}\eqref{performance_approx_eq}\end{NoHyper}.  In particular, we
show that a version of \begin{NoHyper}\eqref{bias_var_scale_eq}\end{NoHyper}
stated in terms of convergence in distribution, rather than convergence of
moments, suffices for \begin{NoHyper}\eqref{performance_approx_eq}\end{NoHyper}
for the FLCI and OCI criteria, and for a truncated version of the RMSE
criterion. Such conditions are appropriate for functionals that involve smooth
nonlinear transformations, which preserve convergence in distribution but may
not preserve convergence of moments: we show in \Cref{delta_method_sec}
that a version of the delta method can be used to verify our conditions in such
cases.

As in the main text, we consider a general setup where, for each $n$ (which
typically denotes sample size), data are drawn from some distribution $P_f$,
which also implicitly depends on $n$, for some $f$.
Let $\mathcal{F}_n\subseteq\mathcal{F}$ be a sequence of function classes, and
let $T:\mathcal{F}\to \mathbb{R}$.  Let $\hat T=\hat T(h;k)$ be a sequence of
estimators indexed implicitly by $n$, and by a kernel $k$ and bandwidth
$h=h_n$, which also depends on $n$.  The function class $\mathcal{F}_n$ is indexed by a
sequence of constants $\SC_n$.

To make concise statements
about uniform-in-$f$ convergence, we introduce some additional notation.  For a
random variable $W_{n, f}$ indexed by the sample size $n$ and the distribution
$f$, we use $W_{n, f_n}\underset{f_n}{\overset{d}{\to}} \mathcal{L}$ to denote
that the distribution of $W_{n, f_n}$ converges in distribution to $\mathcal{L}$
under the sequence $f_n$.  When this holds for all sequences
$f_n\in\mathcal{F}_n$ for some sequence of sets $\mathcal{F}_n$, we write
$W_{n, f}\underset{\mathcal{F}_n}{\overset{d}{\to}} \mathcal{L}$, and we say that
$W_{n, f}$ converges in distribution to $\mathcal{L}$ uniformly over
$\mathcal{F}_n$.  When the limiting law $\mathcal{L}$ is a point mass at some
constant $a$, we write $W_{n, f_n}\underset{f_n}{\overset{p}{\to}} a$ and when
the convergence holds for all $f_n\in\mathcal{F}_n$, we write
$W_{n, f}\underset{\mathcal{F}_n}{\overset{p}{\to}} a$ and say that $W_{n, f}$
converges in probability to $a$ uniformly over $\mathcal{F}_n$.

We make the following assumption on the estimators $\hat T(h;k)$. This
assumption is similar to the
condition \begin{NoHyper}\eqref{bias_var_scale_eq}\end{NoHyper} in the main
text, but uses convergence in distribution rather than convergence of moments.

\begin{assumption}\label{general_bias_var_assump}
  For some sequences of random variables $Z_{n, h, f}$ and $b_{n, h, f}$, we have
  \begin{equation*}
    \hat T(h;k) = T(f) + h^{\gamma_b}\SC_n b_{n, h, f} + h^{\gamma_s}n^{-1/2} Z_{n, h, f}
  \end{equation*}
  where, for some sequence of constants $b_{n, h, f}^*$ and some $S(k)$ and
  $B(k)$,
  $\abs{b_{n, h, f} - b_{n, h, f}^*}\underset{\mathcal{F}_n}{\overset{p}{\to}}
  0$ and
  \begin{equation*}
    \lim_{n\to\infty} \sup_{f\in\mathcal{F}_n} b_{n, h, f}^* = B(k),
    \quad \lim_{n\to\infty} \inf_{f\in\mathcal{F}_n} b_{n, h, f}^* = -B(k), \quad
    Z_{n, h, f} \underset{\mathcal{F}_n}{\overset{d}{\to}} N(0,S(k)^{2}).
  \end{equation*}
\end{assumption}
We verify our main
condition \begin{NoHyper}\eqref{performance_approx_eq}\end{NoHyper} for a class
of performance criteria constructed as follows. Given a loss function
$\ell\colon\mathbb{R}\to\mathbb{R}^+$, let
$\tilde r_{\ell}(b_{0}, s) = E_{Z\sim N(0,1)}\ell(b_{0} + s Z)$ denote the risk
of an estimator that's normally distributed with standard deviation $s$ and bias
$b_{0}$. Let
\begin{equation*}
  \tilde{\rho}_{\ell}(b, s) = \sup_{b_{0}\in [-b, b]} \tilde r_{\ell}(b_{0}, s),
  \quad\text{and}\quad \tilde R_{\ell, \alpha}(b, s)=\inf \left\{\chi \colon
    \tilde \rho_{\ell}(b\chi^{-1}, s\chi^{-1})\le \alpha \right\}
\end{equation*}
denote its worst-case risk over the all biases bounded by $b$ in absolute
value, and the smallest scaling of the worst-case bias and the standard
deviation such that its worst-case risk is bounded by $\alpha$. Similarly, for
an estimator $\hat{T}$ of $T(f)$, let
\begin{equation*}
  \rho_{\ell, \chi}\left(\hat T; \mathcal{F}_n \right)
  =\sup_{f\in\mathcal{F}_n} E_f\ell\left(\chi^{-1}\left(\hat T-T(f) \right) \right),
  \quad\text{and}\quad
  R_{\ell, \alpha}(\hat T; \mathcal{F}_n) = \inf \left\{\chi: \rho_{\ell, \chi}\left(\hat T; \mathcal{F}_n \right)\le \alpha \right\}.
\end{equation*}
Note that if we set $\ell_{\FLCI}(x)=\1{\abs{x}>1}$, then
$R_{\ell_{\FLCI}, \alpha}$ and $\tilde{R}_{\ell_{\FLCI}, \alpha}$ yield the
performance criteria $R_{\FLCI, \alpha}$ and $\tilde{R}_{\FLCI, \alpha}$ as
defined in the main text. Similarly, $R_{\ell_{\RMSE},1}$ and
$\tilde{R}_{\ell_{\RMSE}, 1}$, where $\ell_{\RMSE}(x)=x^{2}$, give the
performance criteria $R_{\RMSE}$ and $\tilde{R}_{\RMSE}$ given in the main text.

To cover performance criteria such as OCI which are constructed from
requirements on multiple loss functions, we use the following construction. Let
$\ell_1,\ldots, \ell_{m}$ be loss functions and let $\alpha_1,\ldots, \alpha_m$ be
given. Let $\lambda\colon (0,\infty)^m\to (0,\infty)$ be continuous and
homogeneous of degree one (i.e.\ it satisfies $\lambda(ax)=a\lambda(x)$ for any
$a>0$). If $m=1$, one can take $\lambda$ to be the identity function. Let
\begin{align*}
  R(\hat T(h;k))&=\lambda(R_{\ell_1,\alpha_1}(\hat T(h;k)), \ldots,
                  R_{\ell_m, \alpha_m}(\hat T(h;k))), \\
  \tilde R(b, s)&=\lambda(R_{\ell_1,\alpha_1}(b, s), \ldots,
                 R_{\ell_m, \alpha_m}(b, s)).
\end{align*}
Note that since
$\tilde R_{\ell_{j}, \alpha_{j}}(t b, t s) =t\inf \{t^{-1}\chi \colon
\tilde{\rho}_{\ell_{j}}(t b\chi^{-1}, t s\chi^{-1})\le \alpha_{j} \} =t
\tilde{R}_{\ell_{j}, \alpha_{j}}(b, s)$, $\tilde{R}$
satisfies~\begin{NoHyper}\eqref{eq:homo-1}\end{NoHyper}. To show how this
generalization covers the OCI criterion $R_{\OCI, \alpha, \beta}$ defined in the
main text, define $\ell_{+}(x)=\1{x > 1}$ and $\ell_{-}(x)=\1{x< -1}$. Then
$R_{\ell_{+}, \alpha}(\hat T;\mathcal{F}_n)$ is the smallest value of $\chi_{+}$
such that $\hor{\hat T-\chi_+, \infty}$ is a one-sided CI with coverage
$1-\alpha$, since
$\rho_{\ell_{+}, \chi_+}(\hat T;\mathcal{F}_n)=\sup_{f\in\mathcal{F}_n}
P_f(\chi_{+}^{-1}(\hat T - T(f)) > 1)=\sup_{f\in\mathcal{F}_n} P_f(\hat T -
\chi_{+} > T(f))$ gives the probability of not covering $T(f)$. The worst-case
$\beta$ quantile of excess length of this CI is the smallest value of $\chi_-$
such that
$\inf_{f\in\mathcal{F}_n} P_f(T(f) - \hat T + \chi_{+} \le \chi_-)\ge \beta$, or
equivalently,
$\rho_{\ell_{-}, \chi_- - \chi_+}(\hat T;\mathcal{F}_n)
=\sup_{f\in\mathcal{F}_n} P_f(T(f) - \hat T > \chi_- - \chi_+)
=\sup_{f\in\mathcal{F}_n} P_f(T(f) - \hat{T} + \chi_{+} > \chi_-)\le 1-\beta$.
Thus, the worst case $\beta$-quantile of excess length of a one-sided CI based
on $\hat{T}$ is given by
$R_{\ell_{+}, \alpha}(\hat T;\mathcal{F}_n)+R_{\ell_{-},1-\beta}(\hat{T};
\mathcal{F}_n)=R_{\OCI, \alpha, \beta}(\hat{T})$. Similarly,
$\tilde R_{\ell_{+}, \alpha}(b, s)+\tilde R_{\ell_{-}, 1-\beta}(b, s)$ gives the criterion
$\tilde{R}_{\OCI, \alpha, \beta}(b, s)$ as defined in the main text.

We make the following assumption on each of the loss functions $\ell$.
\begin{assumption}\label{loss_function_assump}
  (i) $\ell:\mathbb{R}\to \hor{0,\infty}$ is bounded, weakly decreasing on
  $(-\infty,0)$ and weakly increasing on $(0,\infty)$, and continuous almost
  everywhere, and there does not exist a constant function that is almost
  everywhere equal to $\ell$. (ii) $\tilde b\mapsto \tilde r_{\ell}(\tilde b, s)$
  is quasiconvex.
\end{assumption}
For symmetric loss functions, part (ii) follows from part (i) by Anderson's
lemma.

It is immediate that the loss functions $\ell_{+}$, $\ell_{-}$, and
$\ell_{\FLCI}$ satisfy this assumption. The loss $\ell_{\RMSE}$, on the other
hand, does not satisfy this assumption because it is unbounded. However, note
that, for any $c>0$, \Cref{loss_function_assump} holds for the loss
function $\ell_{c}(x)=\min\{x^2, c^2\}$. Since
$\lim_{c\to\infty}R_{\ell_{c},1}(\hat{T}, \mathcal{F}_{n})=R_{\ell_{\RMSE},1}(\hat{T}, \mathcal{F}_{n})$,
and
$\lim_{c\to\infty}\tilde{R}_{\ell_{c},1}(b, s)=\tilde R_{\ell_{\RMSE},1}(b, s)$,
we may interpret this criterion as a truncated version of RMSE\@.

\begin{theorem}\label{general_sufficient_condition_thm}
  Let $h_n$ be a sequence with
  \begin{equation}\label{eq:hn-condition}
  0<\liminf_n h_n (n
  M^2)^{1/[2(\gamma_b-\gamma_s)]}\le \limsup_n h_n (n
  M^2)^{1/[2(\gamma_b-\gamma_s)]}<\infty.
  \end{equation}
  Suppose that $\hat T(h;k)$ satisfies \Cref{general_bias_var_assump}
  for the sequence $h=h_n$. Let $R(\hat T(h;k))$ and $\tilde R(b, s)$ be given
  above, where $\ell_1,\ldots, \ell_{m}$ are loss functions satisfying
  \Cref{loss_function_assump}, and suppose that
  $\tilde R_{\ell_j, \alpha_j}(b, s)>0$ for all $b\ge 0$ and $s>0$ for
  $j=1,\ldots, m$. Then
  \begin{NoHyper}\eqref{performance_approx_eq}\end{NoHyper}
  holds for $R$ and $\tilde R$. Furthermore, if $b_{n, h, f}=b_{n, h, f}^*$,
  $E_{f}Z_{n, h, f}=0$ and $E_{f}Z_{n, h, f}^2\to S(k)^{2}$ uniformly over
  $f\in\mathcal{F}_n$, then
  $\sup_{f\in\mathcal{F}} E_f(\hat T(h;k)-T(f))=-\inf_{f\in\mathcal{F}}
  E_f(\hat{T}(h;k)-T(f))(1+o(1))=h^{\gamma_b}B(k)(1+o(1))$, and
  $\sd_f(\hat T(h;k))=h^{\gamma_s}n^{-1/2}S(k)(1+o(1))$ uniformly over
  $f\in\mathcal{F}_n$,
  and \begin{NoHyper}\eqref{performance_approx_eq}\end{NoHyper} holds with $R$
  and $\tilde R$ given by $R_{\RMSE}$ and $\tilde R_{\RMSE}$.
\end{theorem}

The theorem implies that if \Cref{general_bias_var_assump} holds for bandwidth
sequences $h_{n}$ satisfying~\Cref{eq:hn-condition}, minimizing the criterion
$\lim_{c\to\infty}\lim_{n\to\infty}n^{r/2}M^{r-1}R_{\ell_{c}}(\hat{T}(h;k))$
discussed in \begin{NoHyper}\cref{fn:truncation}\end{NoHyper} in the main text,
where $\ell_{c}$ is the truncated squared error loss defined above, is
equivalent to minimizing the asymptotic RMSE:
\begin{multline*}
  \lim_{c\to\infty}\lim_{n\to\infty}n^{r/2}M^{r-1}R_{\ell_{c}}(\hat{T}(h;k))
  =S(k)^{r}B(k)^{1-r}\lim_{c\to\infty} t^{r-1} \tilde{R}_{\ell_{c}}(t,1)\\
  =S(k)^{r}B(k)^{1-r} t^{r-1} \tilde{R}_{\ell_{\RMSE},1}(t,1).
\end{multline*}
Thus, under this criterion, the optimal bandwidth is given by $\hrmse$.

To prove \Cref{general_sufficient_condition_thm}, we first note some
properties of loss and risk functions in our setup. Note that, under
\Cref{loss_function_assump}, $E\ell(W_n)\to EW$ for any sequence of
random variables $W_n\stackrel{d}{\to}W$ such that $W$ is continuously
distributed (this follows from the continuous mapping theorem and the fact that
$\ell$ is bounded). This also implies that $\tilde r_{\ell}(\tilde b, s)$ is
continuous in $\tilde b$ and $s$ (since $s_{n}Z+\tilde b_n\stackrel{d}{\to} sZ+b$
for $Z\sim N(0,1)$ and $\tilde b_n\to \tilde b$, $s_n\to s$). Also, by part
(ii),
$\tilde \rho_{\ell}(\chi^{-1}b, \chi^{-1}s)=\max_{\tilde b\in\{-b, b\}}E_{Z\sim
  N(0,1)}\ell\left(\chi^{-1}\left(Zs+b \right) \right)$, which is continuous
in $(b, s, \chi)$, and is strictly decreasing in $\chi$ (since $\ell(\chi^{-1}t)$
is weakly decreasing in $\chi$ for each $t$, and, for any $0<\chi<\tilde \chi$,
there is a positive measure set of values of $t$ such that
$\ell(\chi^{-1}t)>\ell(\tilde\chi^{-1}t)$ for $t$ on this set). This implies
that $\tilde R_{\ell, \alpha}(b, s)$, taken as a function of $\alpha$, is the
inverse of the strictly increasing function
$\chi\mapsto \tilde \rho_{\ell}(b\chi^{-1}, s\chi^{-1})$. Since convergence of a
sequence of strictly increasing functions to a continuous, strictly increasing
function implies convergence of their inverse, this implies that
$\tilde R_{\ell, \alpha}(b, s)$ is continuous in $(b, s)$.

We will use the following lemma.

\begin{lemma}\label{general_asymptotic_performance_lemma}
  Let $b, s$ be given. Suppose that $\ell$ satisfies
  \Cref{loss_function_assump}. Suppose that, for any sequence $f_n$,
  there exists $\tilde b\in[-b, b]$ and a subsequence along which
  $a_n (\hat T-T(f_n)) \underset{f_n}{\overset{d}{\to}} N(\tilde b,s^2)$.
  Furthermore, suppose that there exists a sequence $f_n$ such that
  $a_n (\hat T-T(f_n)) \underset{f_n}{\overset{d}{\to}} N(b,s^2)$, and a
  sequence $f_n$ such that
  $a_n (\hat T-T(f_n)) \underset{f_n}{\overset{d}{\to}} N(-b,s^2)$. Then
  $\lim_{n\to\infty}\rho_{\ell, \chi/a_n}(\hat{T};\mathcal{F}_n)=\tilde\rho_{\ell}(\chi^{-1}b, \chi^{-1}s)$
  and
  $\lim_{n\to\infty} a_{n}R_{\ell, \alpha}(\hat T;\mathcal{F}_n)=\tilde R_{\ell, \alpha}(b, s)$.
\end{lemma}
\begin{proof}
  To show
  $\limsup_n\rho_{\ell, \chi/a_n}(\hat T;\mathcal{F}_n)\le
  \tilde\rho_{\ell}(\chi^{-1}b, \chi^{-1}s)$ it suffices to show that, for every
  sequence $f_n$, there is a subsequence along which
  $E_{f_n}\ell\left(a_n\chi^{-1}\left(\hat T-T(f_n) \right) \right)$ converges
  to a constant that is no greater than
  $\tilde\rho_{\ell}(\chi^{-1}b, \chi^{-1}s)$. By assumption, there exists a
  $\tilde b\in[-b, b]$ and a subsequence along which
  $a_n (\hat{T}-T(f_n)) \underset{f_n}{\overset{d}{\to}} N(\tilde b,s^2)$,
  which, under the assumptions on the loss function, implies
  $E_{f_n}\ell\left(a_n\chi^{-1}\left(\hat T-T(f_n) \right) \right)\to
  \tilde{r}_{\ell}(\chi^{-1}\tilde b, \chi^{-1}s)\le
  \rho_{\ell}(\chi^{-1}b, \chi^{-1}s)$ along this subsequence. To show that this
  $\limsup$ is a limit and the inequality is an equality, note that, letting
  $f_n$ be a sequence such that
  $a_n (\hat T-T(f_n)) \underset{f_n}{\overset{d}{\to}} N(b,s^2)$, we have
  $\rho_{\ell, \chi/a_n}(\hat T;\mathcal{F}_n)\ge E_{f_n}\ell\left(
    \chi^{-1}\left(\hat T-T(f_n) \right) \right)\to \tilde r_{\ell}(\chi^{-1}
  b, \chi^{-1}s)$. Similarly, taking a sequence for which the limiting
  distribution is $N(-b,s^2)$, we have
  $\liminf_n \rho_{\ell, \chi/a_n}(\hat T;\mathcal{F}_n)\ge
  \tilde{r}_{\ell}(-\chi^{-1} b, \chi^{-1}s)$. Noting that, under
  \Cref{loss_function_assump}, $\rho_{\ell}(\chi^{-1}b, \chi^{-1}s)$ is
  equal to either $\tilde r_{\ell}(\chi^{-1}\tilde b, \chi^{-1}s)$ or
  $\tilde r_{\ell}(-b\chi^{-1}\tilde b, \chi^{-1}s)$ (or both), it now follows
  that
  $\liminf_n\rho_{\ell, \chi/a_n}(\hat T:\mathcal{F}_n)\ge
  \tilde\rho_{\ell}(\chi^{-1}b, \chi^{-1}s)$. Thus,
  $\lim_{n\to\infty}\rho_{\ell, \chi/a_n}(\hat T:\mathcal{F}_n)=
  \tilde\rho_{\ell}(\chi^{-1}b, \chi^{-1}s)$.

  To derive the limit of $R_{\ell, \alpha}(\hat T;\mathcal{F}_n)$, first note
  that $\rho_{\ell, \chi}(\hat T;\mathcal{F}_n)$ is weakly decreasing in $\chi$
  for any $\chi>0$ for each $n$, since $\ell(\chi^{-1}t)$ is weakly decreasing
  in $\chi$ for all $t$ under \Cref{loss_function_assump}. Also,
  $\tilde{\rho}_{\ell}(\chi^{-1}b, \chi^{-1}s)$ is strictly decreasing in $\chi$.
  Thus, for $\chi> \tilde R_{\ell, \alpha}(b, s)$, we have
  $\tilde \rho_{\ell}(\chi^{-1}b, \chi^{-1}s)<\alpha$ so that, for large enough
  $n$, we have $\rho_{\ell, \chi/a_n}(\hat T;\mathcal{F}_n)<\alpha$ for all
  $\tilde\chi\ge \chi$, which implies
  $R_{\ell, \alpha}(\hat T;\mathcal{F}_n)\le \chi/a_n$. Similarly, for
  $\chi< \tilde R_{\ell, \alpha}(b, s)$, we have
  $\tilde \rho_{\ell}(\chi^{-1}b, \chi^{-1}s)>\alpha$ so that, for large enough
  $n$, we have $\rho_{\ell, \chi/a_n}(\hat T;\mathcal{F}_n)>\alpha$ for all
  $\tilde\chi\le \chi$, which implies
  $R_{\ell, \alpha}(\hat T;\mathcal{F}_n)\ge \chi/a_n$. Thus, for any $\eta>0$,
  we have, for large enough $n$,
  $\tilde R_{\ell, \alpha}(b, s)-\eta\le a_n R_{\ell, \alpha}(\hat{T};
  \mathcal{F}_n)\le \tilde R_{\ell, \alpha}(b, s)+\eta$. It follows that
  $a_n R_{\ell, \alpha}(\hat T;\mathcal{F}_n)\to \tilde R_{\ell, \alpha}(b, s)$.
\end{proof}

We are now ready to prove \Cref{general_sufficient_condition_thm}.

\begin{proof}[Proof of \Cref{general_sufficient_condition_thm}]
  The last statement (regarding convergence of standard deviation and worst-case
  bias and RMSE) follows immediately from the assumptions. To
  show \begin{NoHyper}\eqref{performance_approx_eq}\end{NoHyper} for $R$ and
  $\tilde R$ constructed from loss functions $\ell_1,\ldots, \ell_m$ satisfying
  \Cref{loss_function_assump}, it suffices to show that, for every
  subsequence, there exists a further subsequence along which
  $R(\hat T(h;k))= \tilde{R}(h^{\gamma_b}\SC B(k),
  h^{\gamma_s}n^{-1/2}S(k))(1+o(1))$. By the conditions on $h_n$, we can choose
  this subsequence so that
  $h_n (n\SC_n^2)^{1/[2(\gamma_b-\gamma_s)]}\to h_\infty$ for some $h_\infty>0$.

  Along this subsequence, we have
\begin{equation*}
  h_n^{\gamma_b}\SC_n = h_{\infty}^{\gamma_b}(n \SC_n^2)^{-\gamma_b/[2(\gamma_b-\gamma_s)]} \SC_n (1+o(1))
  = h_{\infty}^{\gamma_b} \SC_n^{1-r} n^{-r/2}(1+o(1))
\end{equation*}
and
\begin{equation*}
  h_n^{\gamma_s}n^{-1/2}
  = h_{\infty}^{\gamma_s} (n \SC_n^2)^{-\gamma_s/[2(\gamma_b-\gamma_s)]}n^{-1/2}(1+o(1))
  = h_{\infty}^{\gamma_s} n^{-r/2}
  \SC_n^{1-r} (1+o(1)).
\end{equation*}
  Thus, on this subsequence, the conditions of \Cref{general_asymptotic_performance_lemma} hold
  with $a_n=M_n^{r-1}n^{r/2}$, $b=h_{\infty}^{\gamma_b}B(k)$ and
  $s=h_{\infty}^{\gamma_s}S(k)$, so that, for each $j=1,\ldots, m$,
  \begin{equation*}
    M_n^{r-1}n^{r/2}R_{\ell_j, \alpha_j}(\hat T(h_n;k);\mathcal{F}_n)\to
    \tilde{R}_{\ell_j, \alpha_j}(h_\infty^{\gamma_b}B(k), h_{\infty}^{\gamma_s}S(k)).
  \end{equation*}
  Also, on this subsequence, using homogeneity and continuity of $\tilde R_{\ell, \alpha}$,
  \begin{multline*}
    M_n^{r-1}n^{r/2}\tilde R_{\ell_j, \alpha_j}(h_n^{\gamma_b}\SC_n B(k), h_n^{\gamma_s}n^{-1/2}S(k)) \\
    = \tilde R_{\ell_j, \alpha_j}(M_n^{r-1}n^{r/2}h_n^{\gamma_b}\SC_n
    B(k), M_n^{r-1}n^{r/2}h_n^{\gamma_s}n^{-1/2}S(k)) \to
    \tilde{R}_{\ell_j, \alpha_j}(h_\infty^{\gamma_b}
    B(k), h_\infty^{\gamma_s}S(k)).
  \end{multline*}
  Combining this with the previous display and using homogeneity of the function
  $\lambda$, it follows that
  \begin{NoHyper}\eqref{performance_approx_eq}\end{NoHyper} holds along this
  subsequence, which gives the result.
\end{proof}

\subsubsection{Delta method}\label{delta_method_sec}

Let $\mathcal{F}_n\subseteq\mathcal{F}$ be a sequence of function classes, and
let $L\colon\mathcal{F}\to \mathbb{R}^m$. We are interested in a parameter
$T(f)=\phi(L(f))$, where $\phi\colon\mathbb{R}^m\to\mathbb{R}$. To cover cases
where $\phi$ may be nonlinear, we assume that $\mathcal{F}_n$ is localized
around a particular value $L^*$ in the range of $L$:
\begin{equation*}
  L(f_n)\to L^*\;\text{for all sequences}\; f_n\in\mathcal{F}_n.
\end{equation*}
This localization of the parameter space plays a similar role to local
asymptotic efficiency results in parametric and regular semiparametric settings
\citep[see, for example, Theorem 8.11 in][]{van_der_vaart_asymptotic_1998}.

We now show that, if $\hat{L}(h;k)$ satisfies a multivariate version of
\Cref{general_bias_var_assump} and $\phi$ is smooth, then
\Cref{general_bias_var_assump} holds for
$\hat{T}(h;k)=\phi(\hat{L}(h;k))$, with $B(k)$ and $S(k)$ defined below.
This is essentially a version of the delta method applied to our setup.

\begin{assumption}\label{general_bias_var_multivariate_assump}
  The function $\phi$ is continuously differentiable at $L^*$, with Jacobian matrix
  $\phi'(L)$ and,
  for some sequences of random vectors $Z_{n, h, f}$ and $b_{n, h, f}$, we have
  \begin{equation*}
    \hat{L}(h;k) = L(f) + h^{\gamma_b}\SC_n b_{n, h, f} + h^{\gamma_s}n^{-1/2}
    Z_{n, h, f},
  \end{equation*}
  where, for a uniformly bounded sequence of constant vectors
  $b_{n, h, f}^*\in\mathbb{R}^m$ and some $\Sigma(k)$ and $B(k)$,
  $|b_{n, h, f} - b_{n, h, f}^*| \underset{\mathcal{F}_n}{\overset{p}{\to}} 0$
  and
  \begin{equation*}
    \lim_{n\to\infty} \sup_{f\in\mathcal{F}_n} \phi'(L^*)b_{n, h, f}^* = B(k),
    \quad \lim_{n\to\infty} \inf_{f\in\mathcal{F}_n} \phi'(L^*)b_{n, h, f}^* =
    -B(k), \quad
    Z_{n, h, f} \underset{\mathcal{F}_n}{\overset{d}{\to}} N(0,\Sigma(k)).
  \end{equation*}
\end{assumption}

\begin{theorem}\label{delta_method_thm}
  Suppose that \Cref{general_bias_var_multivariate_assump} holds, and put
  $S(k)^{2}=\phi'(L^*)\Sigma(k)\phi'(L^*)'$. Then, if $h^{\gamma_b}\SC_n\to 0$
  and $h^{\gamma_s}n^{-1/2}\to 0$, \Cref{general_bias_var_assump} holds for
  $\hat{T}(h;k)=\phi(\hat{L}(h;k))$.
\end{theorem}
\begin{proof}
  First, note that the conditions on the bandwidth imply
  $\hat L \underset{\mathcal{F}_n}{\overset{p}{\to}} L^*$. Then, by a Taylor
  expansion, for some $\tilde{L}=\tilde{L}(\hat L, L(f))$ on the line segment
  between $\hat L$ and $L(f)$, we have
\begin{multline*}
  \phi(\hat L) - \phi(L(f)) = \phi'(\tilde L) [\hat L - L(f)] \\
  = \phi'(\tilde L) [h^{\gamma_b}\SC_n b_{n, h, f} + h^{\gamma_s}n^{-1/2} Z_{n,
    h, f}] = h^{\gamma_b}\SC_n \tilde b_{n, h, f} + h^{\gamma_s}n^{-1/2}
  \tilde{Z}_{n, h, f},
\end{multline*}
where
$\tilde Z_{n, h, f} = \phi'(\tilde L)Z_{n, h, f}
\underset{\mathcal{F}_n}{\overset{d}{\to}} N(0, S(k)^{2})$
by the continuous mapping theorem and
$\tilde b_{n, h, f} = \phi'(\tilde{L}) b_{n, h, f}$ satisfies
$|\tilde b_{n, h, f} - \tilde b_{n, h, f}^*| = |\phi'(\tilde L)b_{n, h, f} -
\phi'(L^*)b_{n, h, f}^*| \underset{\mathcal{F}_n}{\overset{p}{\to}} 0$ where
$\tilde b_{n, h, f}^*=\phi'(L^*)b_{n, h, f}^*$. Thus,
\Cref{general_bias_var_assump} holds with $\tilde b_{n, h, f}$ playing
the role of $b_{n, h, f}$, and $\tilde b_{n, h, f}^*$ playing the role of
$b_{n, h, f}^*$.
\end{proof}

If the function class $\mathcal{F}_n$ places separate restrictions on each
mapping $x\mapsto f_j(x)$ for $j=1, \dotsc, m$, then the set of limits of the
biases $b_{n, h, f}^*$ will take the form
$[-\bar B_{1}(k), \bar B_{1}(k)]\times \dotsb\times [-\bar B_m(k),
\bar{B}_m(k)]$. In this case, the limiting worst-case bias takes the form
\begin{equation}\label{eq:Bk-separate-restrictions}
  B(k)=\sum_{j=1}^m\abs{\phi_j'(L^*)\bar B_j(k)}.
\end{equation}
Note that, while \Cref{delta_method_thm} shows that
\Cref{general_bias_var_assump} is preserved under smooth nonlinear
transformations, such a statement does not hold for a version of this assumption
stated in terms of moments, rather than weak convergence. For such a result, one
needs to either use truncation or place stronger conditions on the class of
estimators. This is analogous to parametric and regular semiparametric settings
such as instrumental variables, in which the asymptotic variance may only be
finite if defined in terms of convergence in distribution.

\subsection{Gaussian white noise model}\label{white_noise_sec}

The approximation \begin{NoHyper}\eqref{performance_approx_eq}\end{NoHyper}
holds as an exact equality (i.e.\ with the $o(1)$ term equal to zero) for the
RMSE, OCI, and FLCI criteria in the Gaussian white noise model whenever the
problem renormalizes in the sense of \citet{DoLo92}. We show this below, using
notation taken mostly from that paper. Consider a Gaussian white noise model
\begin{equation*}
  Y(dt)=(Kf)(t)\, dt +(\sigma/\sqrt{n}) W(dt), \quad t\in\mathbb{R}^d.
\end{equation*}
We are interested in estimating the linear functional $T(f)$ where $f$ is known
to be in the class $\mathcal{F}=\{f\colon J_2(f)\le C\}$ where
$J_2(f):\mathcal{F}\to\mathbb{R}$ and $C\in\mathbb{R}$ are given. Let
$\mathcal{U}_{a, b}$ denote the renormalization operator
$\mathcal{U}_{a, b}f(t)=af(bt)$. Suppose that $T$, $J_2$, and the inner product
are homogeneous: $T(\mathcal{U}_{a, b}f)=ab^{s_0}T(f)$,
$J_2(\mathcal{U}_{a, b}f)=ab^{s_2}J_2(f)$ and $\langle
K\mathcal{U}_{a_1,b}f, K\mathcal{U}_{a_2,b}g\rangle=a_1a_2b^{2s_1}\langle
Kf, Kg\rangle$. These are the same conditions as in \citet{DoLo92} except for the
last one, which is slightly stronger since it must hold for the inner product
rather than just the norm.

Consider the class of linear estimators based on a given kernel $k$:
\begin{equation*}
  \hat{T}(h;k) =h^{s_h}\int (Kk(\cdot/h))(t)\, dY(t)
     =h^{s_h}\int [K\mathcal{U}_{1,h^{-1}}k](t)\, dY(t)
\end{equation*}
for some exponent $s_h$ to be determined below. The worst-case bias of this
estimator is
\begin{equation*}
  \maxbias(\hat T(h;k)) =\sup_{J_2(f)\le C} \left|T(f)-h^{s_h}\langle
    Kk(\cdot/h), Kf\rangle\right|.
\end{equation*}
Note that $J_2(f)\le C$ iff.\ $f=\mathcal{U}_{h^{s_2}, h^{-1}}\tilde f$ for some
$\tilde f$ with $J_2(\tilde f)=J_2(\mathcal{U}_{h^{-s_2}, h}f)=J_2(f)\le C$. This
gives
\begin{align*}
  \maxbias(\hat{T}(h;k))
  &=\sup_{J_2(f)\le C} \left|T(\mathcal{U}_{h^{s_2}, h^{-1}}f)-h^{s_h}\langle Kk(\cdot/h), K\mathcal{U}_{h^{s_2}, h^{-1}}f\rangle\right| \\
  &=\sup_{J_2(f)\le C} \left|h^{s_2-s_0}T(f)-h^{s_h+s_2-2s_1}\langle
    Kk(\cdot), Kf\rangle\right|.
\end{align*}
If we set $s_h=-s_0+2s_1$ so that $s_2-s_0=s_h+s_2-2s_1$, the problem will
renormalize, giving
\begin{equation*}
  \maxbias(\hat T(h;k)) =h^{s_2-s_0} \maxbias(\hat T(1;k)).
\end{equation*}
The variance does not depend on $f$ and is given by
\begin{align*}
  \var_f(\hat T(h;k))&=h^{2s_h}(\sigma^2/n)\langle
  K\mathcal{U}_{1,h^{-1}}k, K\mathcal{U}_{1,h^{-1}}k\rangle
  =h^{2s_h-2s_1}(\sigma^2/n)\langle Kk, Kk\rangle \\
  &=h^{-2s_0+2s_1}(\sigma^2/n)\langle Kk, Kk\rangle.
\end{align*}
Thus, \begin{NoHyper}\Cref{bias_var_scale_eq}\end{NoHyper} holds with
$\gamma_b=s_2-s_0$, $\gamma_s=s_1-s_0$,
\begin{equation*}
  B(k)=\maxbias(\hat T(1;k)) =\sup_{J_2(f)\le C}
  \abs{T(f)-\langle Kk, Kf\rangle},
\end{equation*}
and $S(k)=\sigma \|Kk\|$ and with both $o(1)$ terms equal to zero. This implies
that~\begin{NoHyper}\eqref{performance_approx_eq}\end{NoHyper} holds with the
$o(1)$ term equal to zero, since the estimator is normally distributed.

\subsection{Local polynomial estimators in fixed design
  regression}\label{sec:local-polyn-estim}

This \namecref{sec:local-polyn-estim}
proves \begin{NoHyper}\Cref{theorem:maximum-bias-lp} and
  \Cref{eq:sd-rescaling-lp} in \Cref{sec:inference-point-theory}\end{NoHyper}.

We begin by deriving the worst-case bias of a general linear estimator
\begin{equation*}
  \hat T=\sum_{i=1}^n w(x_i)y_i
\end{equation*}
under Hölder and Taylor classes. For both $\FSY{p}(M)$ and $\FHol{p}(M)$ the
worst-case bias is infinite unless $\sum_{i=1}^n w(x_i)=1$ and $\sum_{i=1}^n
w(x_i)x^j=0$ for $j=1,\ldots,p-1$, so let us assume that $w(\cdot)$ satisfies
these conditions. For $f\in\FSY{p}(M)$, we can write $f(x)=\sum_{j=0}^{p-1}
x^{j}f^{(j)}(0)/j!+r(x)$ with $\abs{r(x)}\le M \abs{x}^p/{p!}$. As noted by
\citet{SaYl78}, this gives the bias under $f$ as $\sum_{i=1}^n w(x_i)r(x_i)$,
which is maximized at $r(x)=M\sign(w(x)) \abs{x}^p/{p!}$, giving
$\maxbias_{\FSY{p}}(\hat T)=\SC\sum_{i=1}^{n} \abs{w(x_i)x}^p/{p!}$.

For $f\in\FHol{p}(\SC)$, the $(p-1)$th derivative is Lipschitz and hence
absolutely continuous. Furthermore, since $\sum_{i=1}^n w(x_i)=1$ and
$\sum_{i=1}^n w(x_i)x^j=0$, the bias at $f$ is the same as the bias at $x\mapsto
f(x)-\sum_{j=0}^{p-1}x^{j}f^{(j)}(0)/{j!}$, so we can assume without loss of
generality that $f(0)=f'(0)=\cdots=f^{(p-1)}(0)$. This allows us to apply the
following lemma.

\begin{lemma}\label{integration_by_parts_lemma}
  Let $\nu$ be a finite measure on $\mathbb{R}$ (with the Lebesgue
  $\sigma$-algebra) with finite support and let
  $w\colon\mathbb{R}\to\mathbb{R}$ be a bounded measurable function with finite
  support. Let $f$ be $p-1$ times differentiable with bounded $p$th derivative
  on a set of Lebesgue measure $1$ and with
  $f(0)=f'(0)=f''(0)=\dotsb=f^{(p-1)}(0)=0$. Then
  \begin{equation*}
    \int_{0}^\infty w(x)f(x)\, d\nu(x) =\int_{s=0}^\infty
    \bar{w}_{p, \nu}(s)f^{(p)}(s)\, ds
  \end{equation*}
  and
  \begin{equation*}
    \int_{-\infty}^0 w(x)f(x)\, d\nu(x) =\int_{s=-\infty}^0
    \bar{w}_{p, \nu}(s)f^{(p)}(s)\, ds
  \end{equation*}
  where
  \begin{equation*}
    \bar w_{p, \nu}(s) =\begin{cases}
      \int_{x=s}^\infty \frac{w(x)(x-s)^{p-1}}{(p-1)!}\, d\nu(x) & s\ge 0 \\
      \int_{x=-\infty}^s \frac{w(x)(s-x)^{p-1}(-1)^p}{(p-1)!}\, d\nu(x) & s<0.
    \end{cases}
  \end{equation*}
\end{lemma}
\begin{proof}
  By the Fundamental Theorem of Calculus and the fact that the first $p-1$
  derivatives at $0$ are $0$, we have
  \begin{equation*}
    f(x)=\int_{t_1=0}^x\int_{t_2=0}^{t_1}\cdots\int_{t_{p}=0}^{t_{p-1}}
    f^{(p)}(t_{p})\, dt_{p}\cdots dt_2dt_1 =\int_{s=0}^x
    \frac{f^{(p)}(s)(x-s)^{p-1}}{(p-1)!} \, ds.
  \end{equation*}
  Thus, by Fubini's Theorem,
  \begin{align*}
    \int_{x=0}^\infty w(x)f(x)\, d\nu(x) &=\int_{x=0}^\infty w(x) \int_{s=0}^x
    \frac{f^{(p)}(s)(x-s)^{p-1}}{(p-1)!}
    \, ds d\nu(x) \\
    &=\int_{s=0}^\infty f^{(p)}(s) \int_{x=s}^\infty
    \frac{w(x)(x-s)^{p-1}}{(p-1)!} \, d\nu(x) ds
  \end{align*}
  which gives the first display in the lemma. The second display in the lemma
  follows from applying the first display with $f(-x)$, $w(-x)$ and $\nu(-x)$
  playing the roles of $f(x)$, $w(x)$ and $\nu(x)$.
\end{proof}

Applying \Cref{integration_by_parts_lemma} with $\nu$ given by the counting
measure that places mass $1$ on each of the $x_i$'s ($\nu(A)=\#\{i\colon x_i\in
A\})$, it follows that the bias under $f$ is given by $\int w(x)f(x)\, d\nu=\int
\bar w_{p, \nu}(s)f^{(p)}(s)\, ds$. This is maximized over $f\in\FHol{p}(\SC)$ by
taking $f^{(p)}(s)=\SC\sign(\bar{w}_{p, \nu}(s))$, which gives
$\maxbias_{\FHol{p}(\SC)}(\hat{T})=\SC\int \abs{\bar{w}_{p, \nu}(s)}\, ds$.

We collect these results in the following theorem.
\begin{theorem}\label{finite_sample_bias_thm}
  For a linear estimator $\hat T=\sum_{i=1}^n w(x_i)y_i$ such that $\sum_{i=1}^n
  w(x_i)=1$ and $\sum_{i=1}^n w(x_i)x^j=0$ for $j=1,\ldots, p-1$,
  \begin{equation*}
    \maxbias_{\FSY{p}(\SC)}(\hat T)=\SC \sum_{i=1}^n \abs{w(x_i)x}^p/p!
    \quad\text{and}\quad
    \maxbias_{\FHol{p}(\SC)}(\hat{T})=\SC\int\abs{\bar{w}_{p, \nu}(s)}\, ds
  \end{equation*}
  where $\bar w_{p, \nu}(s)$ is as defined in
  \Cref{integration_by_parts_lemma} with $\nu$ given by the counting
  measure that places mass $1$ on each of the $x_i$'s.
\end{theorem}

Note that, for $t>0$ and any $q$,
\begin{align}\label{wq_integral_eq}
  \int_{s=t}^\infty \overline w_{q, \nu}(s)\, ds& =\int_{s=t}^\infty
  \int_{x=s}^\infty\frac{w(x)(x-s)^{q-1}}{(q-1)!}\, d\nu(x) ds =\int_{x=t}^\infty
  \int_{s=t}^x
  \frac{w(x)(x-s)^{q-1}}{(q-1)!}\, ds d\nu(x) \nonumber \\
  &=\int_{x=t}^\infty w(x) \left[\frac{-(x-s)^{q}}{q!}\right]_{s=t}^x\, d\nu(x)
  =\int_{x=t}^\infty \frac{w(x)(x-t)^{q}}{q!}\, d\nu(x) =\bar w_{q+1,\nu}(t).
\end{align}
Let us define $\bar w_{0,\nu}(x)=w(x)$, so that this holds for $q=0$ as well.

For the boundary case with $p=2$, the bias is given by (using the fact that the
support of $\nu$ is contained in $\hor{0,\infty}$)
\begin{equation*}
  \int_0^{\infty} w(x)f(x)\,
  d\nu(x)=\int_{0}^{\infty}\bar{w}_{2,\nu}(x)f^{(2)}(x)\, dx
  \quad\text{where}\quad \bar w_{2,\nu}(s)=\int_{x=s}^{\infty} w(x)(x-s)\,
  d\nu(x).
\end{equation*}
For a local linear estimator based on a kernel with nonnegative weights and
support $[-A, A]$, the equivalent kernel $w(x)$ is positive at $x=0$ and negative
at $x=A$ and changes signs once. From~\eqref{wq_integral_eq}, it follows that,
for some $0\le b\le A$, $\bar w_{1,\nu}(x)$ is negative for $x>b$ and
nonnegative for $x<b$. Applying~\eqref{wq_integral_eq} again, this also holds
for $\bar w_{2,\nu}(x)$. Thus, if $\bar w_{2,\nu}(\tilde s)$ were strictly
positive for any $\tilde s> 0$, we would have to have $\bar w_{2,\nu}(s)$
nonnegative for $s\in [0,\tilde s]$. Since $\bar w_{2,\nu}(0)=\sum_{i=1}^n
w(x_i)x_i=0$, we have
\begin{equation*}
  0< \bar w_{2,\nu}(0)-\bar w_{2,\nu}(\tilde s) =-\int_{x=0}^{\tilde{s}}
  w(x)(x-\tilde s)\, d\nu(x)
\end{equation*}
which implies that $\int_{x=\underline s}^{\overline s}\, w(x)d\nu(x)<0$ for
some $0\le \underline s<\overline s<\tilde s$. Since $w(x)$ is positive for
small enough $x$ and changes signs only once, this means that, for some $s^*\le
\tilde s$, we have $w(x)\ge 0$ for $0\le x\le s^*$ and $\int_{x=0}^{s^*}\,
w(x)d\nu(x)>0$. But this is a contradiction, since it means that
$\bar{w}_{2,\nu}(s^*)=-\int_{0}^{s^*}w(x)(x-s^*)\, d\nu(x)<0$. Thus,
$\bar{w}_{2,\nu}(s)$ is weakly negative for all $s$, which implies that the bias
is maximized at $f(x)=-(\SC/2)x^2$.

We now provide a proof for
\begin{NoHyper}\Cref{theorem:maximum-bias-lp}\end{NoHyper}
by proving the result for a more general sequence of estimators of the form
\begin{equation*}
  \hat T=\frac{1}{nh}\sum_{i=1}^n \tilde{k}_n(x_i/h)y_{i},
\end{equation*}
where $\tilde{k}_{n}$ satisfies $\frac{1}{nh}\sum_{i=1}^n \tilde{k}_n(x_i/h)=1$
and $\frac{1}{nh}\sum_{i=1}^n \tilde k_n(x_i/h)x_i^j=0$ for $j=1,\ldots,p-1$. We
further assume
\begin{assumption}\label{kstar_assump}
  The support and magnitude of $\tilde k_n$ are bounded uniformly over $n$, and,
  for some $\tilde k$, $\sup_{u\in\mathbb{R}} |\tilde k_n(u)-\tilde k(u)|\to 0$.
\end{assumption}
\begin{theorem}\label{taylor_holder_asymptotic_thm}
  Suppose \begin{NoHyper}\Cref{x_assump}\end{NoHyper} and \Cref{kstar_assump}
  hold. Then for any bandwidth sequence $h_{n}$ such that $n h_{n}\to \infty$,
  $\liminf_{n} h_n (n\SC^2)^{1/(2p+1)}>0$, and
  $\limsup_{n} h_n (n\SC^2)^{1/(2p+1)}<\infty$,
  \begin{equation*}
    \maxbias_{\FSY{p}(\SC)}(\hat{T}) =
    \frac{\SC h_{n}^p}{p!}\tilde{\mathcal{B}}_{p}^{\text{T}}(\tilde k)(1+o(1)), \qquad
    \tilde{\mathcal{B}}_{p}^{\text{T}}(\tilde k)=d \int_{\mathcal{X}} |u^p
    \tilde{k}(u)|\, du
  \end{equation*}
  and
  \begin{multline*}
    \maxbias_{\FHol{p}(\SC)}(\hat{T}) =
    \frac{\SC h_{n}^p}{p!}\tilde{\mathcal{B}}_{p}^{\textnormal{Höl}}(\tilde{k})
    (1+o(1)), \\ \tilde{\mathcal{B}}_{p}^{\textnormal{Höl}}(\tilde{k}) =d p
    \int_{t=0}^\infty \left|\int_{u\in\mathcal{X}, |u|\ge
        t}\tilde{k}(u)(|u|-t)^{p-1}\, du\right|\, dt.
  \end{multline*}
  If \begin{NoHyper}\Cref{sigma_assump}\end{NoHyper} holds as well,
  then
  \begin{equation*}
    \sd(\hat T) =h_{n}^{-1/2}n^{-1/2}S(\tilde k)(1+o(1)),
  \end{equation*}
  where
  $S(\tilde k)=d^{1/2}\sigma(0)\sqrt{\int_{\mathcal{X}} \tilde k(u)^2\, du}$,
  and~\begin{NoHyper}\eqref{performance_approx_eq}\end{NoHyper} holds for the
  RMSE, FLCI and OCI performance criteria with $\gamma_b=p$ and $\gamma_s=-1/2$.
\end{theorem}
\begin{proof}
  Let $K_{s}$ denote the bound on the support of $\tilde{k}_{n}$, and $K_{m}$
  denote the bound on the magnitude of $\tilde{k}_{n}$.

  The first result for Taylor classes follows immediately since
  \begin{equation*}
    \maxbias_{\FSY{p}(\SC)}(\hat T)
    =\frac{\SC}{p!}h^p\frac{1}{nh}\sum_{i=1}^n |\tilde k_n(x_i/h)||x_i/h|^p
    =\left(\frac{\SC}{p!}h^p d\int_{\mathcal{X}} |\tilde k(u)||u|^p\, du\right)
    (1+o(1)),
  \end{equation*}
  where the first equality follows from \Cref{finite_sample_bias_thm} and
  the second equality follows from the fact that for any function $g(u)$ that is
  bounded over $u$ in compact sets,
  \begin{multline}\label{x_kng_limit}
    \Abs{\frac{1}{nh}\sum_{i=1}^n \tilde{k}_n(x_i/h)g(x_i/h)-d
      \int_{\mathcal{X}} k(u)g(u)\, du} \\
    \leq \Abs{\frac{1}{nh}\sum_{i=1}^n \tilde{k}(x_i/h)g(x_i/h)-d
      \int_{\mathcal{X}} k(u)g(u)\, du}+ \frac{1}{nh}\sum_{i=1}^{n}
    \Abs{\tilde{k}_n(x_i/h)g(x_i/h)-\tilde{k}(x_{i}/h)g(x_{i}/h)}\\
    \leq o(1)+ \frac{1}{nh}\sum_{i=1}^{n}\1{\abs{x_{i}/h} \leq
      K_{s}}\sup_{u\in[-K_{s}, K_{s}]}\abs{g(u)}\cdot \sup_{u\in[-K_{s},
      K_{s}]}\abs{\tilde{k}_{n}(u)-\tilde{k}(u)}=o(1),
  \end{multline}
  where the second line follows by triangle inequality, the third line by
  \begin{NoHyper}\Cref{x_assump}\end{NoHyper} applied to the first
  summand (with $x\mapsto\tilde k(x)g(x)$ playing the role of $g(\cdot)$ in
  \begin{NoHyper}\Cref{x_assump}\end{NoHyper}), and the last equality follows by
  \begin{NoHyper}\Cref{x_assump}\end{NoHyper} applied to the first term, and
  \Cref{kstar_assump} applied to the last term.

  For Hölder classes,
  \begin{equation*}
    \maxbias_{\FHol{p}(\SC)}(\hat T(h;\tilde k_n)) =\SC\int
    \abs{\bar w_{p, \nu}(s)}\, ds
  \end{equation*}
  by \Cref{finite_sample_bias_thm} where $\bar w_{p, \nu}$ is as defined
  in that theorem with $w(x)=\frac{1}{nh}\tilde k_n(x/h)$. We have, for $s>0$,
  \begin{align*}
    \bar w_{p, \nu}(s)&=\int_{x\ge s}
    \frac{\frac{1}{nh}\tilde{k}_n(x/h)(x-s)^{p-1}}{(p-1)!}\, d\nu(x)
    =\frac{1}{nh}\sum_{i=1}^n \frac{\tilde k_n(x_i/h)(x_i-s)^{p-1}}{(p-1)!}\1{x_i\ge s} \\
    &=h^{p-1}\frac{1}{nh}\sum_{i=1}^n
    \frac{\tilde{k}_n(x_i/h)(x_i/h-s/h)^{p-1}}{(p-1)!}\1{x_i/h\ge s/h}.
  \end{align*}
  Thus, by \Cref{x_kng_limit}, for $t\geq 0$,
  $h^{-(p-1)}\bar{w}_{p, \nu}(t\cdot h)\to d\cdot \bar w_p(t)$, where
  \begin{equation*}
    \bar{w}_p(t)=\int_{u\ge t}\frac{\tilde k(u)(u-t)^{p-1}}{(p-1)!}\, du
  \end{equation*}
  (i.e. $\bar w_p(t)$ denotes $\bar w_{p, \nu}(t)$ when $w=\tilde k$ and $\nu$ is
  the Lebesgue measure). Furthermore,
  \begin{equation*}
    |h^{-(p-1)}\bar w_{p, \nu}(t\cdot h)|
    \le \left[\frac{K_{m}}{nh}\sum_{i=1}^n \frac{\1{0\le x_i/h\le
          K_{s}}(x_i/h)^{p-1}}{(p-1)!}\right]\cdot \1{t\le K_{s}} \leq K_1\cdot
    \1{t\le K_{s}},
  \end{equation*}
  where the last inequality holds for some $K_1$ by
  \begin{NoHyper}\Cref{x_assump}\end{NoHyper}. Thus,
  \begin{equation*}
    M\int_{s\ge 0} |\bar w_{p, \nu}(s)|\, ds
    =h^{p} M \int_{t\ge 0} |h^{-(p-1)}\bar w_{p, \nu}(t\cdot h)|\, dt =h^p M
    \left[d\int_{t\ge 0} |\bar w_{p}(t)|\, dt\right](1+o(1))
  \end{equation*}
  by the Dominated Convergence Theorem. Combining this with a symmetric argument
  for $t\le 0$ gives the result.

  For the second part of the theorem, the variance of $\hat{T}$ doesn't depend
  on $f$, and equals
  \begin{equation*}
    \var(\hat{T})=
    \frac{1}{n^2h^2}\sum_{i=1}^n \tilde k_n(x_i/h)^2\sigma^2(x_i) =\frac{1}{nh}
    \tilde{S}_n^2,\qquad{\text{where}}\quad
    \tilde{S}_n^2=\frac{1}{nh}\sum_{i=1}^n \tilde k_n(x_i/h)^2\sigma^2(x_i).
  \end{equation*}
  By the triangle inequality,
  \begin{multline*}
    \Abs{\tilde{S}^{2}_{n}-d\sigma^{2}(0)\int_{\mathcal{X}}\tilde{k}(u)^{2}\, d
      u} \\\leq \sup_{\abs{x}\leq hK_{s}} \Abs{\tilde{k}_{n}(x/h)^{2}\sigma^{2}(x)
      -\tilde{k}(x/h)^{2}\sigma^{2}(0)}\cdot
    \frac{1}{nh}\sum_{i=1}^{n}\1{\abs{x_{i}/h}\leq K_{s}}
    \\
    +\sigma^{2}(0)\Abs{\frac{1}{nh}\sum_{i=1}^{n}\tilde{k}(x_{i}/h)^{2}
      -d\int_{\mathcal{X}}\tilde{k}(u)^{2}\, d u}=o(1),
  \end{multline*}
  where the equality follows by
  \begin{NoHyper}\Cref{x_assump}\end{NoHyper} applied to the second
  summand and the second term of the first summand, and
  \begin{NoHyper}\Cref{sigma_assump}\end{NoHyper}
  and \Cref{kstar_assump} applied to the first term of the first summand. This
  gives the second display in the theorem.

  To show the last statement (verification of
  \begin{NoHyper}\Cref{performance_approx_eq}\end{NoHyper}), we note
  that the above arguments show that \Cref{general_bias_var_assump} holds with
  $b_{n, h, f}=b_{n, h, f}^*$ equal to the bias of the estimator and
  $E_{f}Z_{n, h, f}^2\to S(k)$ uniformly over $\mathcal{F}$, so long as we can
  verify the uniform central limit theorem for
  $Z_{n, h, f}=(nh)^{1/2}[\hat T-E_f\hat{T}]=(nh)^{-1/2}\sum_{i=1}^n
  \tilde{k}_n(x_i/h)u_i$. By the conditions on the errors $u_i$, this follows
  from the Lindeberg central limit theorem so long as
  $\max_{i} [(nh)^{-2}k_n(x_i/u)]^2/(nh)^{-1}=\max_{i} n h k_n(x_i/u)/(nh)\to
  0$. By uniform boundedness of the kernel $k_n$, this holds so long as
  $nh\to\infty$.
\end{proof}

The local polynomial estimator takes the form given above with
\begin{equation*}
  \tilde k_n(u)=e_1'\left(\frac{1}{nh}\sum_{i=1}^n
    k(x_i/h)m_q(x_i/h)m_q(x_i/h)'\right)^{-1}m_q(u)k(u).
\end{equation*}
If $k$ is bounded with bounded support, then, under
\begin{NoHyper}\Cref{x_assump}\end{NoHyper}
this sequence satisfies \Cref{kstar_assump} with
\begin{equation*}
  \tilde k(u)=e_1'\left(d\int_{\mathcal{X}} k(u)m_q(u)m_q(u)'\, d
    u\right)^{-1}m_q(u)k(u) =d^{-1} k^*_q(u),
\end{equation*}
where $k^*_q$ is the equivalent kernel defined in
\begin{NoHyper}\Cref{equivalent_kernel_eq}\end{NoHyper}. \begin{NoHyper}\Cref{theorem:maximum-bias-lp}
  and \Cref{eq:sd-rescaling-lp}\end{NoHyper} then follow immediately by applying
\Cref{taylor_holder_asymptotic_thm} with this choice of $\tilde k_n$ and
$\tilde k$.

\subsection{Fuzzy RD}\label{sec:fuzzy-rd-suppl}

We consider the sequence of parameter spaces
$\mathcal{F}_{n}\subseteq \mathcal{F}(\SC_{1}, \SC_{2})$, such that
$L(f_{n})\to L^{*}$ for all sequences $f_{n}\in\mathcal{F}_{n}$. Here
$L^{*}\in\mathbb{R}^{2}$ is a fixed vector such that $L^{*}_{2}\neq 0$. Let
$\SC=\SC_{1}$, and suppose \begin{NoHyper}\Cref{x_assump}\end{NoHyper} holds
(since the ratio $\SC_{1}/\SC_{2}$ is fixed, it suffices to verify the
assumption for $\SC=\SC_{1}$). Assume also that the random variables
$\{u_i\}_{i=1}^n$ are independent with $E u_{i}=0$, $\var(u_{i})=\Omega(x_i)$
and $E (u^{2}_{1i}+u_{2i}^{2})^{1+\eta}\le 1/\eta$ for some $\eta>0$, and that
the covariance function $\Omega(x)$ is left- and right- continuous at $x=0$ with
$\Omega_{+}(0)=\lim_{x\downarrow 0}\Omega(x)>0$ and
$\Omega_{-}(0)=\lim_{x\uparrow 0}\Omega(x)>0$. It then follows by adapting
arguments in the proof
of \begin{NoHyper}\Cref{theorem:maximum-bias-lp}\end{NoHyper} that for any
bandwidth sequence $h_{n}$ with $n h_n\to\infty$ and
$0<\liminf_n h_n (n\SC^2)^{1/(2p+1)}<\limsup_n h_n (n\SC^2)^{1/(2p+1)}<\infty$,
\begin{equation*}
  \hat{L}(h;k)=L(f)+h^{2}
  \begin{pmatrix}
    M_{1}b^{*}_{n, h, f,1}\\
    M_{2}b^{*}_{n, h, f,2}
  \end{pmatrix}+\frac{1}{\sqrt{nh}}Z_{n, h, f},
\end{equation*}
where $Z_{n, h, f}$ converges in distribution to $N(0,\Sigma(k))$ uniformly over
$\mathcal{F}_{n}$ with
\begin{equation*}
  \Sigma(k)=\int_{0}^{\infty}k_{1}^{*}(u)^{2}\, \dd u \cdot
  (\Omega_{+}(0)+\Omega_{-}(0))/d,
\end{equation*}
and
$b^{*}_{n, h, f, j}=\sum_{i=1}^{n}(w_{+}(x_{i})+w_{-}(x_{i}))f_{j}(x_{i})/M_{j}$
for $j=1, 2$, and the limits of these biases lie in the set
$[\tilde{B}(k), -\tilde{B}(k)]^{2}$, where
$\tilde{B}(k)=\int_{0}^{\infty}u^{2}k_{1}^{*}(u)\, \dd u$.
From~\eqref{eq:Bk-separate-restrictions}, we obtain that
\Cref{general_bias_var_multivariate_assump} holds with $\gamma_{b}=2$,
$\gamma_{s}=-1/2$, and
\begin{equation*}
  B(k)=-(\abs{\phi_{1}'(L^{*})}+\SC_{2}/\SC_{1}\abs{\phi_{2}'(L^{*})})
  \int_{0}^{\infty}u^{2}k_{1}^{*}(u)\, \dd u
  =-\frac{1+\SC_{2}/\SC_{1}\abs{L_{1}^{*}/L_{2}^{*}}}{\abs{L_{2}^{*}}}
  \int_{0}^{\infty}u^{2}k_{1}^{*}(u)\, \dd u.
\end{equation*}
Thus, by \Cref{delta_method_thm}, condition
\begin{NoHyper}\eqref{performance_approx_eq}\end{NoHyper} holds for FLCI, OCI,
and truncated RMSE with
\begin{equation*}
  S(k)^{2}=\frac{\int_{0}^{\infty}k_{1}^{*}(u)^{2}\, \dd u}{d}
  \frac{\varsigma_{+}^{2}(0;L_{1}^{*}/L_{2}^{*})+
    \varsigma_{+}^{2}(0;L_{1}^{*}/L_{2}^{*})}{(L_{2}^{*})^{2}},
\end{equation*}
where $\varsigma^{2}(x;T)=(1, -T)\Omega(x)(1, -T)'$,
$\varsigma_{+}^{2}(0;T)=\lim_{x\downarrow 0}\varsigma^{2}(x;T)$, and
$\varsigma_{-}^{2}(0;T)=\lim_{x\uparrow 0}\varsigma^{2}(x;T)$.

The expressions for $\operatorname{avar}(\hat{T}(h;k))$ and
$\operatorname{\overline{abias}}(\hat{T}(h; k))$ in the main text then follow by
observing that
$\sum_{i=1}^{n}\tilde{w}^{n}(x_{i};h, k)^{2}\phi'(L(f))\Omega(x_{i})
\phi'(L(f))'=S(k)/nh(1+o(1))$, and
$(\abs{\phi_{1}'(L(f))}\SC_{1}+\abs{\phi_{2}'(L(f))}\SC_{2})
\sum_{i=1}^{n}\tilde{w}^{n}(x_{i};h, k)/2=\SC_{1}h^{2}B(k)(1+o(1))$.

\section{Additional applications}\label{sec:addit-appl}

This \namecref{sec:addit-appl} considers additional applications not considered
in the main text, using the sufficient conditions from
\Cref{sufficient_conditions_sec}. \Cref{density_sec} verifies our conditions in
the density setting, and \Cref{auctions_sec} applies these results to a problem
in the auctions literature.

\subsection{Density estimation}\label{density_sec}

Consider estimating a density at a point, which we normalize to $0$. We observe
$\{X_i\}_{i=1}^n$ iid with density $f$ on the intersection of $\mathcal{X}$ and
some neighborhood of $0$, where either $\mathcal{X}=\mathbb{R}$ or
$\mathcal{X}=\hor{0,\infty}$. We are interested in $T(f)=f(0)$. Let
$\hat T=\hat{T}(h;k)=\frac{1}{nh}\sum_{i=1}^n k(X_i/h)$ be a kernel estimate
where $k$ is a kernel with $\int_{\mathcal{X}} k(u)\, du=1$ and finite support.
Let $\mathcal{F}=\mathcal{F}(\SC)$ denote the Hölder class $\FHol{p}(\SC)$ or
Taylor class $\FSY{p}(\SC)$ of order $p$, as defined in the paper. Assume that
the kernel $k$ satisfies $\int_{\mathcal{X}}u^{j}k(u)\, du=0$ for
$j=1,\ldots,p-1$. Let $f^*>0$ be given, and let $a_n$ be a sequence converging
to zero more slowly than any polynomial. Let $\mathcal{F}(\SC, [-a, a])$ denote
the class for which the Hölder or Taylor condition is imposed only for
$x\in [-a, a]\cap\mathcal{X}$, and let
$\mathcal{F}_n=\mathcal{F}(\SC_n;[-a_n, a_n])\cap \{f: |f(x)-f^*|\le
a_n\;\text{all}\; x\in [-a_n, a_n]\cap \mathcal{X}, \, f(x)\ge 0\;\text{all $x$}, \,
\int f(x)\, dx=1\}$.

We show that \begin{NoHyper}\eqref{performance_approx_eq}\end{NoHyper} holds for
the performance criteria considered in the main text by verifying
\Cref{general_bias_var_assump}. This gives a generalization of the
results in \citet{SaYl81}, who consider RMSE optimal kernels in Taylor classes,
to performance criteria other than RMSE, and to cover Hölder classes in addition
to Taylor classes. Note that $\mathcal{F}_n$ localizes the parameter space
around a density with $T(f)=f^*$, similar to \Cref{delta_method_sec}.
This differs slightly from \citet{SaYl81}, who consider a fixed parameter space
$\mathcal{F}$ which only places an upper bound $f^*$ on $f(0)$. However, the
result given below is essentially the same, since the worst-case risk over this
class is taken in a shrinking neighborhood of $f^*$ (i.e.\ the worst-case risk
is the same as in our setup). Also, note that we only impose the Hölder or
Taylor condition in the set $[-a_n, a_n]$, although we would obtain the same
result if we did not impose this condition so long as $\SC_n$ increases slowly
enough so that the function can be extended to satisfy the smoothness condition
outside of $[-a_n, a_n]$.

\begin{theorem}\label{density_thm}
  For any bandwidth sequence with
  $h_n\to 0$,
  $h_n^p\SC_n \to 0$, $n h_n\to\infty$ and
  \begin{equation*}
  0<\liminf_n h_n (n
  \SC^2)^{1/(2p-1)}\le \limsup_n h_n (n
  \SC^2)^{1/(2p-1)}<\infty,
  \end{equation*}
  the kernel density estimator satisfies \Cref{general_bias_var_assump} with
  $S(k)=\sqrt{f^*\int_{\mathcal{X}}k(u)^2\, du}$, $B(k)$ given in
  \begin{NoHyper}\Cref{theorem:maximum-bias-lp}\end{NoHyper} and with
  $\gamma_b=p$ and $\gamma_s=-1/2$. In
  particular, \begin{NoHyper}\eqref{performance_approx_eq}\end{NoHyper} holds
  for the FLCI and OCI criteria. Furthermore, we can take
  $b_{n, h, f}=b_{n, h, f}^*$ to be nonrandom, and $E_{f}Z_{n, h, f}=0$ and
  $E_{f}Z_{n, h, f}^2\to S(k)$ uniformly over $\mathcal{F}_n$, so that
  \begin{NoHyper}\eqref{performance_approx_eq}\end{NoHyper}
  holds for the RMSE criterion.
\end{theorem}
\begin{proof}
  We have
\begin{equation}\label{density_estimate_expansion_eq}
  \hat T(h;k) = T(f) + h^p \SC b_{n, h, f} + (nh)^{-1/2}Z_{n, h, f}
\end{equation}
where
\begin{equation*}
  b_{n, h, f} = h^{-p}\SC^{-1}[E_f\hat T(h;k)-T(f)] = h^{-p}\SC^{-1}\frac{1}{h}\int_{\mathcal{X}}k(x/h)[f(x)-f(0)]\, dx
\end{equation*}
is nonrandom and can be taken to be equal to $b_{n, h, f}^*$, and
\begin{equation*}
Z_{n, h, f}=\frac{1}{\sqrt{nh}}\sum_{i=1}^n [k(X_i/h) - E_{f}k(X_i/h)].
\end{equation*}
Once $h_n$ is small enough relative to $a_n$ and $f^*$,
the set of possible biases for the class $\mathcal{F}_n$ will be the same as
for the Taylor or Hölder class $\mathcal{F}(\SC)$, without the additional
local restriction of $f(x)$ for $x$ near zero, or the restriction that $f$ be a
density (note, in particular, that, letting $C$ be a bound on the support of the
kernel $k$, the bias depends only
on $f(x)$ for $x$ in $[-C h_n, C h_n]$, and that the first $p-1$ derivatives of
$f$ at zero can be
taken to be equal to zero without loss of generality, so that, for any function
$f$ satisfying the Hölder or Taylor condition, $f(x)$ is bounded from below
by $f^*-a_n- \tilde C \SC_{n}h_n^{p}$ on this set for some constant $\tilde C$;
this function can then be extrapolated so that it is positive on $[-a_n, a_n]$
while maintaining the Hölder or Taylor condition, and then defined outside
of $[-a_n, a_n]$ so that it integrates to one), so that
\begin{equation*}
  \left\{b_{n, h, f} : f\in\mathcal{F}_n \right\} = \left\{h^{-p}\SC^{-1}
    \frac{1}{h}\int_{\mathcal{X}}k(x/h)[f(x)-f(0)]\, dx : f\in\mathcal{F}(M) \right\}.
\end{equation*}
By the renormalization property of $\mathcal{F}$ ($f\in\mathcal{F}(1)$ iff.
$x\mapsto h^{p}\SC f(x/h)$ is in $\mathcal{F}(\SC)$), the set in the above
display remains the same if $h$ and $\SC$ are each replaced by $1$. Thus, the
expressions for asymptotic bias derived in
\begin{NoHyper}\Cref{theorem:maximum-bias-lp}\end{NoHyper} holds exactly
with $\gamma_b=p$ and $B(k)$ given in
\begin{NoHyper}\Cref{theorem:maximum-bias-lp}\end{NoHyper}
(with $k$ playing the role of the
equivalent kernel, $k^*_q$).
For the variance, we have
\begin{equation*}
  \var_f(Z_{n, h, f}) = \frac{1}{h}\int_{\mathcal{X}} k(x/h)^2f(x)\, dx - \frac{1}{h} \left[\int_{\mathcal{X}} k(x/h)f(x)\, dx \right]^2.
\end{equation*}
The second term converges to $0$ uniformly over $\mathcal{F}_n$, and the first
term converges to $f^*\int_{\mathcal{X}} k(u)^2\, du$
uniformly over $\mathcal{F}_n$.
To verify the Lindeberg condition for asymptotic normality, note that
$\frac{1}{nh}\sum_{i=1}^n E_{f}K(X_i/h)^2\1{K(X_i/h)^2\ge \varepsilon n h}\to 0$
uniformly over $f\in\mathcal{F}_n$ since $nh\to\infty$.
\end{proof}

\subsection{First price auctions}\label{auctions_sec}

Our results for density estimation and nonparametric regression can be combined
with the delta method (\Cref{delta_method_thm}) to verify our conditions
for nonlinear functions of densities and nonparametric regression functions
evaluated at finitely many points.  To illustrate, we consider a setting from
the auctions literature involving a nonlinear function of a density.

\citet{guerre_optimal_2000} consider the problem of recovering valuations from
bids in a first price auction setting.  Here, we consider a simple version of
their setting with no covariates, and the same number of bidders in each
auction.  We observe $n$ total bids from symmetric independent private value
sealed bid auctions with $I>1$ bidders each, with independent valuations.  The
bids $\{X_i\}_{i=1}^n$ are then iid and, letting $f$
denote their density, the valuation for a bidder with bid $X_i=x$ is given by
\begin{equation*}
\xi(x;f, I) = x + \frac{1}{I-1} \frac{\int_{-\infty}^x f(t)\, dt}{f(x)}
\end{equation*}
\citep[Equation (3) in][]{guerre_optimal_2000}. Consider the problem of
estimating $T(f)=\xi(x_0;f, I)$ at a particular point $x_0$. Let
$\mathcal{F}_{GPV, n}$ be defined in the same way as the class $\mathcal{F}_n$
defined in \Cref{density_sec} with $\mathcal{X}=\mathbb{R}$, but with an
additional local restriction on the cumulative distribution function (CDF)
$\int_{-\infty}^{x}f(t)\, dt$:
$\mathcal{F}_{GPV, n}=\mathcal{F}_n\cap \{f: |\int_{-\infty}^x f(t) \, dt -
F^*|\le a_n\}$ where $F^*\in (0,1)$ is given.

Let
$\hat L(h;k)=(\hat L_{1}(h;k), \hat L_{2}(h, k))=\left(\frac{1}{n}\sum_{i=1}^n
  \1{X_i\le x_0}, \frac{1}{nh}\sum_{i=1}^{n}k((X_i-x_0)/h)\right)$, where $k$ is
a kernel satisfying the conditions in \Cref{density_sec} and $h$ satisfies the
conditions of \Cref{density_thm} for some $p$. Let
$\phi(L)=x_0+\frac{1}{I-1}\frac{L_1}{L_2}$. Then a plug-in estimator of $T(f)$
is given by $\hat{T}(h;k)=\phi(\hat{L}(h;k))$. To
verify \begin{NoHyper}\eqref{performance_approx_eq}\end{NoHyper}, we verify
\Cref{general_bias_var_multivariate_assump}. First, note that, by a slight
generalization of \Cref{density_thm}, $\hat L_2(h;k)$
satisfies~\eqref{density_estimate_expansion_eq}, where $b_{n, h, f}$ is
nonrandom and, for large enough $n$, ranges over the set
$[-B_{2}(k), B_{2}(k)]$, with $B_{2}(k)$ given by $B(k)$ in
\begin{NoHyper}\Cref{theorem:maximum-bias-lp}\end{NoHyper}, and with
$Z_{n, h, f}$ converging to a $N(0,S_{2}(k))$ distribution uniformly over
$\mathcal{F}_{GPV, n}$, where $S_{2}(k)=\sqrt{f^*\int k(u)\, du}$. (This follows
from the arguments in \Cref{density_thm} along with the observation that
the local restriction on $\int_{-\infty}^x f(t)\, dt$ does not restrict the set
of possible biases $b_{n, h, f}$ for large enough $n$.) Also, $\hat L_{1}(h;k)$
satisfies
$\hat L_{1}(h;k)=L_{1}(f)+h^{\gamma_b}M_{n} b_{n, h, f, 1} +
h^{\gamma_s}n^{-1/2}Z_{n, h, f, 1}$ with $\gamma_s=-1/2$, where
$b_{n, h, f, 1}=0$ and
$Z_{n, h, f, 1}=n^{1/2}h^{-\gamma_s}\left(\hat L_{1}(h;k)-L_{1}(h;k) \right)$
converges in probability to zero uniformly over $\mathcal{F}_{GPV, n}$. Thus,
\Cref{general_bias_var_multivariate_assump} holds with $b_{n, h, f}$
ranging over the set $\{0\}\times [-B_{2}(k), B_{2}(k)]$ and with
$\Sigma(k)= \bigl(\begin{smallmatrix}
  0 & 0 \\
  0 & S_{2}(k)
  \end{smallmatrix}
  \bigr)$ and $\phi'(L^*)=\frac{1}{I-1}[\frac{1}{f^*}, -\frac{F^*}{f^*}]$. It
  follows that \begin{NoHyper}\eqref{performance_approx_eq}\end{NoHyper} holds
  for the FLCI and OCI criteria, with $\gamma_s=-1/2$ and $\gamma_b=p$,
  $B(k)=B_{2}(k)\frac{F^*}{(I-1)f^*}$, and
  $S(k)=S_{2}(k)\frac{{F^*}^2}{(I-1)^2{f^*}^2}$. Note, however, that, since a
  density estimator appears in the denominator of the estimator of $T(f)$, the
  RMSE may not even be finite, and so truncation will be needed to apply our
  results to the RMSE criterion.

  We note that the class $\mathcal{F}_{GPV, n}$ places assumptions conditions
  directly on the bid distribution, and does not incorporate additional
  restrictions that may arise from the assumption that $f$ arises from an
  equilibrium in a first price auction model. We leave for future research
  whether such restrictions place sharper bounds on the bias, as well as the
  question of deriving primitive conditions on the value distribution for our
  smoothness assumptions on the bid distribution. Such questions are addressed
  by \citet{guerre_optimal_2000}, although they focus on a slightly different
  setting, since they consider rate optimality in the supremum norm for
  estimation of the value distribution (rather than asymptotic constants for
  estimation of the function $\xi(x;f, I)$ at a given point $x_0$).

\section{Additional details for applications}\label{sec:addit-deta-appl}
This \namecref{sec:addit-deta-appl} gives additional details for applications in
\begin{NoHyper}\Cref{sec:applications}\end{NoHyper}.
\Cref{sec:rd-with-different} calculates the efficiency gain from using different
bandwidths on either side of the cutoff in sharp RD\@.
\Cref{sec:optim-kern-details} gives details of optimal kernel calculations
discussed in \begin{NoHyper}\Cref{sec:inference-point-theory}\end{NoHyper}.
\Cref{sec:kernel-constants} gives the kernels constants
$\int_{\mathcal{X}}k^{*}_{\po}(u)^{2}\,\dd u$, and $\mathcal{B}_{p, \po}(k)$ for
selected kernels.

\subsection{Regression discontinuity with different band\-widths on either side of
  the cutoff}\label{sec:rd-with-different}

We consider a slightly more general setup than that considered in
\begin{NoHyper}\Cref{sec:sharp-regr-disc}\end{NoHyper}. Consider
estimating a parameter $T(f)$, $f\in\mathcal{F}$, using a class of estimators
$\hat{T}(h_{+}, h_{-};k)$ indexed by two bandwidths $h_{-}$ and $h_{+}$. Suppose
that the worst-case (over $\mathcal{F}$) performance of
$\hat{T}(h_{+}, h_{-};k)$ according to a given criterion satisfies
\begin{equation}\label{eq:two-bandwidths}
  R(\hat{T}(h_{+}, h_{-};k))=\tilde{R}(\SC B(k)(h_{-}^{\gamma_{b}}+h_{+}^{\gamma_{b}}),
  n^{-1/2}(S_{+}(k)^{2}h_{+}^{2\gamma_{s}}+S_{-}(k)^{2}h_{-}^{2\gamma_{s}})^{1/2})(1+o(1)),
\end{equation}
where $\tilde{R}(b, s)$ denotes the value of the criterion when
$\hat{T}(h_{+}, h_{-};k)-T(f)\sim N(b, s^{2})$, and $S(k)>0$ and $B(k)>0$. Assume
that $\tilde{R}$ satisfies~\begin{NoHyper}\eqref{eq:homo-1}\end{NoHyper}.

In the RD application in \begin{NoHyper}\Cref{sec:sharp-regr-disc}\end{NoHyper},
if \begin{NoHyper}\Cref{x_assump,sigma_assump}\end{NoHyper} hold (with the
requirement that $\sigma^{2}(x)$ is continuous $0$ replaced by right- and
left-continuity of $\sigma^{2}_{+}(x)$ and $\sigma^{2}_{-}(x)$), then
Condition~\eqref{eq:two-bandwidths} holds with $\gamma_{s}=-1/2$,
$\gamma_{b}=2$,
$S_{+}(k)=\sigma^{2}_{+}(0)\int_{0}^{\infty}k_{1}^{*}(u)^{2}\, \dd u/d$,
$S_{-}(k)=\sigma^{2}_{-}(0)\int_{0}^{\infty}k_{1}^{*}(u)^{2}\, \dd u/d$, and
$B(k)=-\int_{0}^{\infty}u^{2}k_{1}^{*}(u)\dd u/2$.

Let $\rho=h_{+}/h_{-}$ denote the ratio of the bandwidths, and let $t$ denote
the ratio of the leading worst-case bias and standard deviation terms,
\begin{equation*}
  t=\frac{\SC B(k)(h_{-}^{\gamma_{b}}+h_{+}^{\gamma_{b}})}{
    n^{-1/2}(S_{+}(k)^{2}h_{+}^{2\gamma_{s}}+S_{-}(k)^{2}h_{-}^{2\gamma_{s}})^{1/2}
  }= h_{-}^{\gamma_{b}-\gamma_{s}}\frac{\SC B(k)(1+\rho^{\gamma_{b}})}{
    n^{-1/2}
    (S_{+}(k)^{2}\rho^{2\gamma_{s}}+S_{-}(k)^{2})^{1/2} }.
\end{equation*}
Substituting $h_{+}=\rho h_{-}$ and $h_{-}=(tn^{-1/2}
(S_{+}(k)^{2}\rho^{2\gamma_{s}}+S_{-}(k)^{2})^{1/2}\SC^{-1}B(k)^{-1}(1+\rho^{\gamma_{b}})^{-1}
)^{1/(\gamma_{b}-\gamma_{s})}$ into~\eqref{eq:two-bandwidths} and using
linearity of $\tilde{R}$ gives
\begin{equation*}
  \begin{split}
    R(\hat{T}(h_{+}, h_{-};k))&=
    \tilde{R}(\SC B(k)h_{-}^{\gamma_{b}}(1+\rho^{\gamma_{b}}),
    h_{-}^{\gamma_{s}}n^{-1/2}(S_{+}(k)^{2}\rho^{2\gamma_{s}}+S_{-}(k)^{2})^{1/2})(1+o(1))\\
    &=\SC^{1-r}n^{-r/2} (1+\varsigma(k)^{2} \rho^{2\gamma_{s}})^{r/2}
    \left(1+\rho^{\gamma_{b}}\right)^{1-r}S_{-}(k)^{r}B(k)^{1-r}
    \tilde{R}(t,1)(1+o(1)),
  \end{split}
\end{equation*}
where $r=\gamma_{b}/(\gamma_{b}-\gamma_{s})$ is the rate exponent, and
$\varsigma(k)=S_{+}(k)/S_{-}(k)$ is the ratio of the variance constants.
Therefore, the optimal bias-sd ratio is given by
$t^{*}_{R}=\argmin_{t>0}\tilde{R}(t,1)$, and depends only on the performance
criterion. The optimal bandwidth ratio $\rho$ is given by
\begin{equation*}
  \rho_{*}=\argmin_{\rho}(1+\varsigma(k)^{2} \rho^{2\gamma_{s}})^{r/2}
  \left(1+\rho^{\gamma_{b}}\right)^{1-r}=\varsigma(k)^{\frac{2}{\gamma_{b}-2\gamma_{s}}},
\end{equation*}
and doesn't depend on the performance criterion.

Consequently, inference that restricts the two bandwidths to be the same
(i.e.~restricting $\rho=1$) has asymptotic efficiency given by
\begin{equation*}
  \begin{split}
    \lim_{n\to\infty} \frac{\min_{h_{+}, h_{-}}R(\hat{T}(h_{+}, h_{-};k))}{
      \min_{h}R(\hat{T}(h;k))}&= \left(\frac{(1+\varsigma(k)^{2}
        \rho_{*}^{2\gamma_{s}})^{\gamma_{b}/2}
        \left(1+\rho_{*}^{\gamma_{b}}\right)^{-\gamma_{s}}}{
        (1+\varsigma(k)^{2})^{\gamma_{b}/2} 2^{-\gamma_{s}}}
    \right)^{\frac{1}{\gamma_{b}-\gamma_{s}}} \\
    & = 2^{r-1} \frac{\left(1+ \varsigma(k)^{\frac{2r}{2-r}}\right)^{1-r/2}}{
      (1+\varsigma(k)^{2})^{r/2}}.
  \end{split}
\end{equation*}
In the RD application in
\begin{NoHyper}\Cref{sec:sharp-regr-disc}\end{NoHyper},
$\varsigma(k)=\sigma_{+}(0)/\sigma_{-}(0)$, and $r=4/5$. The display above
implies that the efficiency of restricting the bandwidths to be the same on
either side of the cutoff is at least 99.0\% if
$2/3\leq \sigma_{+}/\sigma_{-}\leq 3/2$, and the efficiency is still 94.5\% when
the ratio of standard deviations equals 3. There is therefore little gain from
allowing the bandwidths to be different.

\subsection{Optimal kernels for inference at a
  point}\label{sec:optim-kern-details}

The optimal equivalent kernel under the Taylor class $\FSY{p}(\SC)$ solves
\begin{NoHyper}\Cref{eq:minimax-problem}\end{NoHyper} in the main
text. The solution is given by
\begin{equation*}
  k_{SY, p}(u)=\left(b+\textstyle\sum_{j=1}^{p-1}\alpha_{j}u^{j}-\abs{u}^{p}\right)_{+}
  -\left(b+\textstyle\sum_{j=1}^{p-1}\alpha_{j}u^{j}+\abs{u}^{p}\right)_{-},
\end{equation*}
the coefficients $b$ and $\alpha$ solving
\begin{align*}
  \int_{\mathcal{X}}u^{j}k_{SY, p}(u)\, \dd u&=0,\qquad j=1,\dotsc,p-1,&\text{and}&&
  \int_{\mathcal{X}}k_{SY, p}(u)\, \dd u&=1.
\end{align*}
For $p=1$, the triangular kernel $k_{\text{Tri}}(u)=(1-\abs{u})_{+}$ is optimal
both in the interior and on the boundary. In the interior for $p=2$,
$\alpha_{1}=0$ solves the problem, yielding the Epanechnikov kernel
$k_{\text{Epa}}(u)=\frac{3}{4}(1-u^{2})_{+}$ after rescaling. For other cases,
the solution can be easily found numerically. \Cref{fig:sy} plots the
optimal equivalent kernels for $p=2$, $3$, and $4$, rescaled to be supported on
$[0,1]$ and $[-1,1]$ in the boundary and interior case, respectively.

The optimal equivalent kernel under the Hölder class $\FHol{2}(\SC)$ has the
form of a quadratic spline with infinite number of knots on a compact interval.
In particular, in the interior, the optimal kernel is given by
$f_{\text{Höl},2}^{\text{Int}}(u)/\int_{-\infty}^{\infty}f^{\text{Int}}_{\text{Höl},2}(u)\, \dd
u$, where
\begin{equation*}
  f^{\text{Int}}_{\text{Höl},2}(u)=1-\frac{1}{2}x^{2}+\sum_{j=0}^{\infty}(-1)^{j}(\abs{x}-k_{j})_{+}^{2},
\end{equation*}
and the knots $k_{j}$ are given by
$k_{j}=\frac{(1+q)^{1/2}}{1-q^{1/2}}(2-q^{j/2}-q^{(j+1)/2})$, where $q$ is a
constant $q=(3+{\sqrt{33}-\sqrt{26+6\sqrt{33}}})^{2}/16$.

At the boundary, the
optimal kernel is given by
$f_{\text{Höl},2}^{\text{Bd}}(u)/\int_{-\infty}^{\infty}f^{\text{Bd}}_{\text{Höl},2}(u)\, \dd
u$, where
\begin{equation*}
  f_{\text{Höl},2}^{\text{Bd}}(u)=(1-x_{0}x+x^{2}/2)\1{0\leq x\leq x_{0}}
  +(1-x_{0}^{2})f^{\text{Int}}_{\text{Höl},2}((x-x_{0})/(x_{0}^{2}-1))\1{x>x_{0}},
\end{equation*}
with $x_{0}\approx 1.49969$, so that for $x>x_{0}$, the optimal boundary kernel
is given by a rescaled version of the optimal interior kernel. The optimal
kernels are plotted in \Cref{fig:hol2}.

\subsection{Kernel constants}\label{sec:kernel-constants}
For the uniform, triangular, and Epanechnikov kernels, the kernel constants
$\int_{\mathcal{X}}k^{*}_{\po}(u)^{2}\, \dd u$,
$\mathcal{B}_{p, \po}^{\text{T}}(k)$, and
$\mathcal{B}^{\textnormal{Höl}}_{p, \po}(k)$ discussed in
\begin{NoHyper}\Cref{sec:inference-point-theory}\end{NoHyper} involve integrals that can be
computed in closed form. \Cref{tab:lp-boundary-constants} gives these constants
for the case in which the point of interest is an interior point, and
\Cref{tab:lp-interior-constants} gives them for the boundary case.

\section{Data-driven Bandwidths}\label{data-driven_bw_sec_append}

This \namecref{data-driven_bw_sec_append} considers CIs with the bandwidth
chosen based on the data, with the smoothness constant $\SC$ treated as unknown.
In particular, we formalize the statements
in \begin{NoHyper}\Cref{sec:pract-impl}\end{NoHyper} regarding honesty and
near-optimality of CIs based on the rule-of-thumb bandwidth suggested in that
\namecref{sec:pract-impl}, over a regularity class that imposes further
restrictions.

Consider the regression setting in
\begin{NoHyper}\Cref{sec:setup-estimators}\end{NoHyper}. Let $\mathcal{F}(\SC)$
denote the Taylor or Hölder class defined in
\begin{NoHyper}\Cref{sec:inference-point-theory}\end{NoHyper}, which
places the bound $\SC$ on the $p$th derivative of the regression function. Let
$\mathcal{F}(\SC;\eta)$ denote the class that imposes this bound only over
$x\in [-\eta, \eta]$.
We note that all of our asymptotic results for $\mathcal{F}(\SC)$ hold for
$\mathcal{F}(\SC;\eta)$ as well.
Let $\hat T_{\po}(h;k)$ denote the $\po$th order local polynomial estimator,
with $\po\ge p-1$. Let
$h_n=h(\SC)=(n^{-1/2}S(k)t/(\SC B(k)))^{1/(\gamma_b-\gamma_s)}$ denote a
sequence of bandwidths corresponding to bias-sd ratio $t$. Here, $B(k)$ and
$S(k)$ are given in
\begin{NoHyper}\Cref{theorem:maximum-bias-lp}\end{NoHyper}
and $\gamma_b=p$ and $\gamma_{s}=-1/2$.
Let $r=2p/(2p-1)$ denote the rate exponent.
It follows from the results in the main text that the CI $\{\hat T_{\po}(h_n;k)\pm \widehat{\se}(h_n;k)\cdot \cv_{1-\alpha}(t)\}$ has
correct asymptotic coverage, and it is near-optimal if highly efficient choices
for $t$ and $k$ are used.

We consider the CI
$\{\hat T_{\po}(\hat h;k)\pm \widehat{\se}(\hat h;k)\cdot \cv_{1-\alpha}(t)\}$,
which uses a data-driven bandwidth $\hat h$ to estimate the optimal bandwidth
$h_n=h(\SC)$, thereby avoiding the requirement of prior knowledge of $\SC$. As
discussed in the main text, results from \citet{low97}, \citet{CaLo04} and
\citet{ArKo18optimal} imply that it is impossible for such a CI to achieve
coverage and near-optimality over $\mathcal{F}(\SC;\eta)$ when $\SC$ is unknown.
We therefore consider a class $\mathcal{G}(\SC)\subsetneq \mathcal{F}(\SC;\eta)$
that imposes additional conditions that allow $M$ to be estimated consistently.
We allow $\mathcal{G}(\SC)$ to depend directly on the sample size as well, but
we leave this implicit in the notation.
\Cref{est_h_general_results} presents results under high level
consistency conditions on $\hat{h}$ over the class $\mathcal{G}(\SC)$.
\Cref{polynomial_approx_conditions_sec} defines a particular class
$\mathcal{G}(\SC)$ that formalizes the notion that local smoothness of $f$ is no
smaller than its smoothness at large scales, and verifies that the rule-of-thumb
bandwidth suggested in \begin{NoHyper}\Cref{sec:pract-impl}\end{NoHyper}
leads to honest CIs over this class. \Cref{est_h_lower_bounds_sec}
derives asymptotic efficiency bounds that show formally that the CI with
rule-of-thumb bandwidth considered in
\Cref{polynomial_approx_conditions_sec} is highly efficient over the
class $\mathcal{G}(\SC)$. In particular, it is impossible to substantively
improve upon this CI using the additional restrictions in the class
$\mathcal{G}(\SC)$. \Cref{est_h_limiting_model_sec} presents auxiliary
results and intuition for the efficiency bounds presented in
\Cref{est_h_lower_bounds_sec}.

\subsection{General results for estimated
  \texorpdfstring{$h$}{h}}\label{est_h_general_results}

We maintain \begin{NoHyper}\Cref{x_assump,sigma_assump}\end{NoHyper}. We make
the following additional assumptions on the kernel.

\begin{assumption}\label{kernel_assump_est_h}
  The kernel $k$ is bounded and Lipschitz continuous with finite support.
\end{assumption}

\begin{theorem}\label{general_est_h_theorem}
  Let $h(\SC)=(n^{-1/2}S(k)t/(\SC B(k)))^{2/(2p+1)}$ where $t>0$. Let $\hat h$
  be a bandwidth sequence, which may depend on the data, such that
  $\hat{h}/h(\SC)\stackrel{p}{\to} 1$ and $n h(\SC)\to \infty$ uniformly over
  $\cup_{\SC\in[\underline \SC_n, \overline \SC_n]} \mathcal{G}(\SC)$, where
  $\mathcal{G}(\SC)\subset \mathcal{F}(\SC;\eta)$. Let $\widehat\se(h;k)$ be a
  standard error such that $\widehat\se(\hat h;k)/\sd_f(\hat h;k)$ converges in
  probability to one uniformly over
  $\cup_{\SC\in[\underline \SC_n, \overline \SC_n]} \mathcal{G}(\SC)$. Let
  \begin{NoHyper}\Cref{sigma_assump}\end{NoHyper} and
  \Cref{kernel_assump_est_h} hold, and let
  \begin{NoHyper}\Cref{x_assump}\end{NoHyper} hold for any sequence
  $\SC_n\in[\underline\SC_n, \overline\SC_n]$. Then
  \begin{equation*}
    \liminf_{n\to\infty}\inf_{f\in\cup_{\SC\in[\underline
        \SC_n, \overline \SC_n]}
      \mathcal{G}(\SC)} P_f\left(T(f)\in \left(\hat T_{\po}(\hat h;k)\pm \widehat\se(\hat h;k)\cv_{1-\alpha}(t) \right) \right)\ge 1-\alpha.
  \end{equation*}
The length of the CI satisfies
\begin{equation*}
\lim_{n\to\infty}\sup_{M\in[\underline M_n, \overline M_n]}\sup_{f\in \mathcal{G}(M)}
  P_f\left(\left| \frac{2 \widehat\se(\hat h;k) \cv_{1-\alpha}(t)}
    {2 n^{-r/2} \SC^{1-r} S(k)^{r}B(k)^{1-r}  t^{r-1} \cv_{1-\alpha}(t)}
  - 1 \right| > \delta \right)
  \to 0
\end{equation*}
for any $\delta>0$.
\end{theorem}

To prove this theorem, let $\SC_n\in[\underline \SC_n, \overline \SC_n]$ be
given, and let $f_n$ be a sequence of functions in $\mathcal{G}(\SC_n)$.  Let $h_n=h(\SC_n)$.
For any sequence $c_n\to 0$, the coverage probability under $f_n$ is bounded from below by
\begin{equation*}
  P_{f_n}\left(\left| \frac{\hat T_{\po}(h_n;k)-T(f_n)}{\widehat \se(\hat h;k)} \right|\le \cv_{1-\alpha}(t)(1-c_n) \right)
  - P_{f_n}\left(\left| \frac{\hat T_{\po}(\hat h;k)-\hat T_{\po}(h_n;k)}{\widehat \se(\hat h;k)} \right|> \cv_{1-\alpha}(t)c_n \right).
\end{equation*}
For the first term, we first note that
\begin{NoHyper}\Cref{two_R_thm}\end{NoHyper} continues to hold with
$\sqrt{1/r-1}$ replaced by $t$ and $\hrmse$ replaced by $h_n$, with obvious
modifications to the proof. The first term is asymptotically bounded from below
by $1-\alpha$ by
\begin{NoHyper}\Cref{theorem:maximum-bias-lp}\end{NoHyper} and this
generalization of \begin{NoHyper}\Cref{two_R_thm}\end{NoHyper}, applied
with $\widehat\se(\hat h;k)(1-c_n)$ playing the role of the standard error in
\begin{NoHyper}\Cref{two_R_thm}\end{NoHyper}
(note that, by
\begin{NoHyper}\Cref{theorem:maximum-bias-lp}\end{NoHyper}
and
the assumptions on $\hat h$, $\widehat\se(\hat h;k)/[n^{-1/2}h_n^{-1/2}S(k)]$
converges in probability to one under $f_n$). The second term will converge to
zero for $c_n$ decreasing slowly enough so long as
$\sqrt{nh_n}\left(\hat T_{\po}(\hat h;k)-\hat T_{\po}(h_n;k) \right)$ converges
in probability to zero (again using the fact that
$\widehat\se(\hat h;k)/[n^{-1/2}h_n^{-1/2}S(k)]$ converges in probability to
one).

Let
\begin{equation*}
  a_n(h)=\left(\frac{1}{n h}\sum_{i=1}^{n}k(x_i/h)m_{\po}(x_i/h)m_{\po}(x_i/h)' \right)^{-1}e_1,
  \quad
  b_n(x_i;h)=\frac{1}{nh}m_{\po}(x_i/h)k(x_i/h)
\end{equation*}
and let $w_{\po}^n(x;h, k)=a_n(h)'b_n(x_i;h)$.
We have
\begin{multline}\label{T_diff_eq_est_h}
  \sqrt{n h_n}\left[\hat T_{\po}(h_n;k) - \hat T_{\po}(\hat h;k) \right]
    = \sqrt{n h_n}\sum_{i=1}^n [w_{\po}^n(x_i;h_n, k)-w_{\po}^n(x_i;\hat h, k)] y_i \\
  = \sqrt{n h_n}\sum_{i=1}^n [w_{\po}^n(x_i;h_n, k)-w_{\po}^n(x_i;\hat h, k)] f(x_i)\\
   + \sqrt{n h_n}\sum_{i=1}^n [w_{\po}^n(x_i;h_n, k)-w_{\po}^n(x_i;\hat h, k)] u_i.
\end{multline}
Using a Taylor approximation to $f(x_i)$ around $x=0$ and the fact that
$\sum_{i=1}^{n}w_{\po}^n(x_i;h, k)x_i^j=0$ for $j<p$, it follows that the first term is
bounded by
\begin{equation*}
  \sqrt{n h_n}\SC_n \sum_{i=1}^n |w_{\po}^n(x_i;h_n, k)-w_{\po}^n(x_i;\hat h, k)| \frac{\abs{x_{i}}^{p}}{p!}
  =\frac{t S(k)}{B(k)p!}\sum_{i=1}^n |w_{\po}^n(x_i;h_n, k)-w_{\po}^n(x_i;\hat h, k)|\abs*{\frac{x_{i}}{h_{n}}}^{p},
\end{equation*}
where we substitute
$M_n=t n^{-1/2}S(k)/(B(k)h_n^{p+1/2})$.
Letting $C$ be a bound on the support of the kernel $k$, we have $|x_i|\le
C\max\{\hat h, h_n\}$ for any $x_i$ such that the summand is nonzero.  Thus, on
the event $\hat h\le 2 h_n$, the above display is bounded by $\frac{(2C)^p t
  S(k)}{B(k)p!}$ times
\begin{equation*}
  \sum_{i=1}^n |w_{\po}^n(x_i;h_n, k)-w_{\po}^n(x_i;\hat h, k)|.
\end{equation*}
Using the fact that
$w_{\po}^n(x_i;h_n, k)-w_{\po}^n(x_i;\hat{h}, k)=a_n(h_n)'[b_n(x_i;h_n)-b_n(x_i;\hat{h})]
+ [a_n(h)-a_n(\hat h)]' b_n(x_i;\hat{h})$, it follows that the above display is
bounded by
\begin{equation*}
  \|a_n(h_n)\| \sum_{i=1}^n \| b_n(x_i;h_n)-b_n(x_i;\hat h) \|
  + \|a_n(h_n)-a_n(\hat h)\| \sum_{i=1}^n \|b_n(x_i;\hat h) \|.
\end{equation*}
Similarly, the last term in~\eqref{T_diff_eq_est_h} is bounded by
\begin{equation*}
  \|a_n(h_n)\| \left\|\sqrt{n h_n}\sum_{i=1}^n [b_n(x_i;h_n)-b_n(x_i;\hat h)] u_i\right\|
  + \|a_n(h_n)-a_n(\hat h)\| \left\|\sqrt{n h_n}\sum_{i=1}^n b_n(x_i;\hat h) u_i\right\|.
\end{equation*}

Both of these quantities converge in probability to zero by the following lemma.

\begin{lemma}
  Suppose that \begin{NoHyper}\Cref{x_assump}\end{NoHyper} and
  \Cref{kernel_assump_est_h} hold. Let $\tilde g(x)=k(x)x^j$ or
  $\tilde g(x)=|k(x)x^j|$ for some $j\ge 0$. Then
  \begin{equation*}
    \lim_{\delta\to 0}\limsup_{n\to\infty} \sup_{s\in[1-\delta,1+\delta]} \frac{1}{n h_n}\sum_{i=1}^n
    \left| \tilde g(x_i/(s h_n))-\tilde g(x_i/ h_n) \right|= 0.
  \end{equation*}
  and
  \begin{equation*}
    \lim_{\delta\to 0}\limsup_{n\to\infty} \sup_{s\in[1-\delta,1+\delta]} \left| \frac{1}{n s h_n}\sum_{i=1}^n
      \tilde g(x_i/(s h_n))- d\int_{\mathcal{X}}\tilde g(u)\, du \right| = 0.
  \end{equation*}
  If, in addition, \begin{NoHyper}\Cref{sigma_assump}\end{NoHyper} holds, then,
  for all $\varepsilon>0$,
  \begin{equation*}
    \lim_{\delta\to 0}\limsup_{n\to\infty} \sup_{s\in[1-\delta,1+\delta]} P\left(\sup_{s\in[1-\delta,1+\delta]}
    \left| \frac{1}{\sqrt{nh_n}}\sum_{i=1}^n [\tilde g(x_i/(s h_n))-\tilde g(x_i/h_n)]u_i\right|>\varepsilon \right) = 0.
  \end{equation*}
\end{lemma}
\begin{proof}
  By \begin{NoHyper}\Cref{x_assump}\end{NoHyper}, the second display in the
  lemma follows from the first. By \Cref{kernel_assump_est_h}, for large enough
  $C$,
  $|\tilde g(u)-\tilde g(u')|\le C\abs{u-u'}\1{\max\{\abs{u}, \abs{u'}\}\le C}$.
  Thus, the first display in the lemma is bounded by
  \begin{multline*}
    \lim_{\delta\to 0}\limsup_{n\to\infty} \sup_{s\in[1-\delta,1+\delta]}\frac{1}{n h_n}\sum_{i=1}^n C\cdot |s^{-1}-1| \1{\abs{x_i/h_n}\le 2C} \\
    = \lim_{\delta\to 0} \left[\sup_{s\in[1-\delta,1+\delta]}|s^{-1}-1| \right]
    \limsup_{n\to\infty} \frac{1}{n h_n}\sum_{i=1}^n
    C\cdot\1{|x_i/h_n|\le 2C} \\
    = \lim_{\delta\to 0}\left[\sup_{s\in[1-\delta,1+\delta]}|s^{-1}-1|
    \right]\int_{\mathcal{X}}\1{u\le 2C}\, du\cdot C = 0.
  \end{multline*}
  For the second part of the lemma, we have, for $s, \tilde s$ in a small enough
  neighborhood of $1$, letting $\overline\sigma^2$ denote a bound on
  $\sigma^2(x)$ in a neighborhood of zero,
  \begin{multline*}
    E\left(\sum_{i=1}^n\left[\frac{1}{\sqrt{nh_n}}\tilde g(x_i/(sh_n))-\tilde g(x_i/(\tilde sh_n)) \right]u_i \right)^2
      \le \overline\sigma^2\frac{1}{nh_n}\sum_{i=1}^n \left[\tilde g(x_i/(sh_n))-\tilde g(x_i/(\tilde sh_n)) \right]^2 \\
    \le \overline\sigma^2\frac{1}{nh_n}\sum_{i=1}^n C^2 \abs{x_i/h_n}^2 \abs{s^{-1}-\tilde s^{-1}}^2\1{|x_i/h_n|\le 2C}.
  \end{multline*}
  For large enough $n$, this is bounded by $|s^{-1}-\tilde s^{-1}|^2$
  times a constant that does not depend on $n$.
  The result now follows from Example 2.2.12 in \citet{van_der_vaart_weak_1996}.
\end{proof}

Finally, for the last statement of the theorem, note that the length of the CI
is given by $2 \widehat\se(\hat h;k) \cv_{1-\alpha}(t)$ which, under the
sequence $f_n$, is equal to a $1+o_P(1)$ term times
\begin{equation*}
  2 n^{-1/2}h_n^{-1/2}S(k) \cv_{1-\alpha}(t)
  = 2 n^{-r/2} \SC_n^{1-r} S(k)^{r}B(k)^{1-r} t^{r-1} \cv_{1-\alpha}(t).
\end{equation*}

\subsection{Bounds based on global polynomial approximations}\label{polynomial_approx_conditions_sec}

We now verify the conditions of \Cref{general_est_h_theorem} in a
particular setting. In particular, we consider classes $\mathcal{G}$ that relate
$\SC$ to a global polynomial approximation to the regression function, along
with a plug-in bandwidth $\hat h$ based on this assumption.

Let $\mathcal{F}(\SC)$ be the Taylor or Hölder class of order $p$, and let
$\mathcal{F}(\SC;\eta)$ denote the class that imposes this bound only over
$x\in [-\eta, \eta]$. Let $\tilde p\ge p$ be given. Let $Q_{\tilde p}f$ denote
the minimum mean squared error $\tilde p$th order polynomial predictor for the
regression function $f$:
\begin{equation*}
  Q_{\tilde p}f=\arg\min_{h} \int (f(x)-h(x))^2 d(x)\sigma^2(x)\, dx
\end{equation*}
where the minimum is taken over polynomials of order $\tilde p$. Here, $d(x)$ is
such that the $x_i$'s behave as if drawn from a distribution with density $x_i$,
as formalized in the \Cref{density_polynomial_predictor_assump} below.

Let $x_{\min}, x_{\max}$ be given with $-\infty<x_{\min}<x_{\max}<\infty$.  Let
\begin{equation*}
  J(f)=J(f;\tilde p, x_{\min}, x_{\max})=\sup_{x\in [x_{\min}, x_{\max}]}|[Q_{\tilde p}f]^{(p)}(x)|
\end{equation*}
denote the maximum $p$th derivative of the minimum mean squared error
$\tilde{p}$th order approximation of $f$.

Let $\varepsilon>0$ be given. Let
\begin{gather*}
  \mathcal{Q}(\SC, \tilde p, x_{\min}, x_{\max}, \varepsilon)=\{f: J(f) = \varepsilon\SC\},\\
  \mathcal{G}(\SC)=\mathcal{G}(\SC;\tilde p, \varepsilon, \eta, x_{\min}, x_{\max})
                                                             = \mathcal{F}(\SC, \eta)\cap \mathcal{Q}(\SC, \tilde p, x_{\min}, x_{\max}, \varepsilon)\cap \{f: \sup_{x}|f(x)|\le K\},
\end{gather*}
where $K$ is some large constant, and
\begin{equation*}
  \mathcal{H}(\underline \SC, \overline \SC)=\cup_{M\in[\underline M, \overline M]}
  \mathcal{G}(\SC;\tilde p, \varepsilon, \eta, x_{\min}, x_{\max}).
\end{equation*}
This class formalizes the notion that the $p$th derivative in a neighborhood of
zero is bounded by $\varepsilon^{-1}$ times the maximum $p$th derivative of a
global $\tilde p$th order global polynomial approximation. Setting
$\varepsilon=1$ corresponds to the suggestion in
the main text.

Let
\begin{equation*}
  \hat Q_{\tilde p}=\arg\min_{h} \sum_{i=1}^n (y_i-h(x_i))^2,\quad
  \hat J=\sup_{x\in [x_{\min}, x_{\max}]}|\hat Q_{\tilde p}^{(p)}(x)|
\end{equation*}

We make the following additional assumption on the $x_i$'s.

\begin{assumption}\label{density_polynomial_predictor_assump}
  For some bounded function $d(x)$ and a sequence $c_n$ with $c_n\to\infty$ and $c_n/\sqrt{n}\to 0$, we have, for each $j=0,\ldots, \tilde p$,
  \begin{equation*}
    c_n\left| \frac{1}{n}\sum_{i=1}^n x_{i}^{j}f_n(x_i)- \int u^{j}f_n(u)d(u)\, du \right|\to 0
  \end{equation*}
  for any uniformly bounded sequence of functions $f_n$.
  Furthermore, the $\tilde p+1$ by $\tilde p+1$ matrix with $(j, \ell)$th element
  given by $\int u^{j+\ell-2}d(u)\, du$ is invertible.
\end{assumption}

Given a sequence $c_n$ satisfying the conditions of
\Cref{density_polynomial_predictor_assump}, if the $x_i$'s are drawn iid from a
distribution with density $d(x)$ for which all moments are finite, then
\Cref{density_polynomial_predictor_assump} will hold with probability
approaching one.

We note the following consistency result for $\hat J$.

\begin{lemma}\label{Jf_consistency_lemma}
  Suppose \begin{NoHyper}\Cref{sigma_assump}\end{NoHyper} holds with
  $\sigma^2(x)$ bounded and that \Cref{density_polynomial_predictor_assump}
  holds. Then $c_n|\hat J-J(f)|\stackrel{p}{\to} 0$ uniformly over
  $\{f: \sup_x|f(x)|\le K\}$.
\end{lemma}
\begin{proof}
  Let $A$ denote the $\tilde p+1$ by $\tilde p+1$ matrix with $(j, \ell)$th
  element given by $\int u^{j+\ell-2}d(u)\, du$, and let $\hat A$ denote the
  sample analogue with $(j, \ell)$th element given by
  $\frac{1}{n}\sum_{i=1}^{n}x_i^{j+\ell-2}$. Let $b_f$ be the
  $(\tilde p+1)\times 1$ vector with $j$th element $\int u^j f(u)d(u)\, du$ and
  $\hat b$ be the sample analogue with $j$th element
  $\frac{1}{n}\sum_{i=1}^n x_i^{j-1} y_i$. Then $A^{-1}b_f$ gives the
  coefficients of the polynomial $Q_{\tilde p}f$, and $\hat A^{-1}\hat b$ gives
  the coefficients of the polynomial $\hat Q$. Let $s(A, b)$ denote the function
  that takes the maximum of the $p$th derivative of this polynomial over
  $[x_{\min}, x_{\max}]$, so that $J(f)=s(A, b_f)$ and $\hat{J}=s(\hat A, \hat b)$.
  Note that $|s(\hat A, \hat b) - s(A, b_f)|$ is bounded by
  $\max\{\|\hat A-A\|, \|\hat b-b_f\|\}$ times a constant that does not depend on
  $f$, so it suffices to show that $c_n\max\{\|\hat A-A\|, \|\hat b-b_f\|\}$
  converges in probability to zero uniformly over bounded $f$.

  We have $c_n\|\hat A-A\|\to 0$ by \Cref{density_polynomial_predictor_assump}.
  The $j$th element of $c_n(\hat b-b_f)$ is given by
  \begin{equation*}
    \frac{c_n}{n}\sum_{i=1}^n u_{i}x_i^{j-1}
    + c_n\left(\frac{1}{n}\sum_{i=1}^{n}f(x_i)x_i^{j-1} - \int f(u)u^{j-1}d(u)\, du \right).
  \end{equation*}
  The expectation of the square of the first term converges to zero, since it is
  bounded by $c_n^2/n^2$ times a sequence that converges to a constant by
  \Cref{density_polynomial_predictor_assump}. The last term converges to zero
  uniformly over bounded $f$ by \Cref{density_polynomial_predictor_assump}.
  Thus, $c_n\|\hat b-b_f\|\stackrel{p}{\to} 0$ uniformly over bounded $f$.
\end{proof}

Let $\SC_n$ and $\varepsilon_n$ be given, and consider honesty over the sequence
of classes
$\mathcal{G}(\SC_n;\tilde p, \varepsilon_n, \allowbreak \eta, x_{\min},
x_{\max})$.
Let $t$ be given, and let
$\hat h=(n^{-1/2}\hat S(k)t/(\hat\SC \hat{B}(k)))^{2/(2p+1)}$ where
$\hat S(k)/S(k)$ and $\hat B(k)/B(k)$ converge in probability to one uniformly
over $\mathcal{G}(\SC_n)$ (as discussed in
\begin{NoHyper}\Cref{sec:pract-impl}\end{NoHyper}, we can also
directly minimize the sample analogue of the criterion such that $t$ is the
asymptotically optimal bias-sd ratio). Then $\hat h$ will satisfy the conditions
of \Cref{general_est_h_theorem} so long as $\hat h/h(\SC_n)$ converges in
probability to one uniformly over $\mathcal{G}(\SC_n)$, where
\begin{equation*}
  h(\SC)=(n^{-1/2}S(k)t/(\SC B(k)))^{2/(2p+1)}.
\end{equation*}
For this, it suffices that $\hat{\SC}/\SC_n$ converges in probability to one
uniformly over $\mathcal{G}(\SC_n)$.

According to \Cref{Jf_consistency_lemma}, we can use the estimate
$\hat \SC =\varepsilon^{-1}\hat J$, which gives
\begin{equation*}
  \frac{\hat \SC}{\SC_n}-1=\frac{\varepsilon_n^{-1}[\hat J-J(f)]}{\SC_n}
  =o_P(1/(\varepsilon_n\SC_{n}c_n))
\end{equation*}
uniformly over
$\mathcal{G}(M;\tilde p, \varepsilon_n, \eta, x_{\min}, x_{\max})$. If
\Cref{density_polynomial_predictor_assump} holds for any $c_n$ with
$c_n/\sqrt{n}\to 0$, then this can be made to go to zero so long as
$\varepsilon_n M_n\sqrt{n}\to\infty$. Thus, the resulting CI is honest over the
class $\mathcal{H}(\underline M_n, \overline M_n)$ so long as
$\varepsilon_n \underline M_n\sqrt{n}\to\infty$, and such that
\begin{NoHyper}\Cref{x_assump}\end{NoHyper} holds for the sequences
$\underline M_n$ and $\overline M_n$. Note also that, if one uses
$\hat M=\tilde\varepsilon^{-1}\hat J$ where $\tilde\varepsilon<\varepsilon$
(thereby choosing $\varepsilon$ to be ``too small''), then the resulting CI will
be wider, but will still have correct coverage.

While \begin{NoHyper}\Cref{x_assump}\end{NoHyper} is stated as a high level
condition, note that, in order for this condition to hold with probability
approaching one when the $x_i$'s are drawn iid from a distribution satisfying
appropriate regularity conditions, we will need $nh_n\to\infty$ and $h_n\to 0$
for the given sequence $h_n$. This will be ensured for any sequence
$M_n\in [\underline M_n, \overline M_n]$ iff. $\underline M_n$ satisfies
$n\underline M_n^2\to\infty$ and $\overline M_n$ satisfies
$\overline M_n/n^p\to 0$ so that
$n (n\overline M_n^2)^{-1/(2p+1)}=n^{2p/(2p+1)}\overline M_n^{-2/(2p+1)}\to
\infty$. Also, note that we have assumed a uniform bound on the magnitude of the
regression function, which means that $\varepsilon_n\overline M_n$ must be
bounded uniformly over $n$ (although this condition could likely be relaxed).

\subsection{Lower bounds}\label{est_h_lower_bounds_sec}

The CI in \Cref{general_est_h_theorem} has the property that the ratio of its length
to the length of an ``oracle'' FLCI that uses the unknown true $M$ converges to
one. If the optimal kernel is used and the bias-sd ratio is chosen to be optimal
for FLCI length, then this CI is efficient among FLCIs over the class
$\mathcal{F}(\SC;\eta)$. Furthermore, it is highly efficient among all CIs that
are honest over the class $\mathcal{F}(\SC;\eta)$, since one can apply bounds
such as Corollary 3.3 in \citet{ArKo18optimal}. However, these results do not
apply to the class $\mathcal{G}(\SC)$ over which the feasible CI with estimated
optimal bandwidth has coverage, since
$\mathcal{G}(\SC)\subsetneq \mathcal{F}(\SC;\eta)$: they do not rule out the
possibility that this restricted class might allow for a more informative CI\@.
To address this, we now derive efficiency bounds for the class
$\mathcal{G}(\SC)=\mathcal{G}(\SC;\tilde p, \varepsilon, \eta, x_{\min}, x_{\max})$
used in \Cref{polynomial_approx_conditions_sec}.

\begin{theorem}\label{est_h_lower_bound_thm}
  Let $\SC$, $\varepsilon$, $\eta$ and $[x_{\min}, x_{\max}]$ be given. Suppose
  that \begin{NoHyper}\Cref{x_assump,sigma_assump}\end{NoHyper} hold with
  $\sigma(x)$ bounded from above and below away from zero and $u_i$ following a
  normal distribution, and that \Cref{density_polynomial_predictor_assump} holds
  with $d(x)$ strictly positive on some open set in
  $\mathbb{R}\backslash [-\eta, \eta]$. Then, if the constant $K$ used to define
  $\mathcal{G}(\SC)$ is large enough, the following holds. For any sequence of
  CIs $\{\hat T\pm \hat \chi\}$ with asymptotic coverage at least $1-\alpha$
  under $\mathcal{G}(M)$,
  \begin{equation*}
    \lim_{C\to\infty}\liminf_n \inf_{f\in\mathcal{G}(M)}
    E_{f_n} \min\{2n^{r/2}\hat \chi, C\}
    \ge
    \frac{2 M^{1-r} S(k^*)^{r}B(k^*)^{1-r}}{r^r(1-r)^{r-1}}
    \int_{z=-\infty}^{z_{1-\alpha}}(z_{1-\alpha}-z)^{r}\, d\Phi(z)
  \end{equation*}
  where $k^*$ minimizes $S(k^*)^{r}B(k^*)^{1-r}$.
\end{theorem}

If $\hat h$ and $\widehat{\se}(h;k)$ satisfy the conditions of
\Cref{general_est_h_theorem}, then, by
\Cref{est_h_lower_bound_thm}, the relative efficiency of any CI
$\{\hat T\pm \hat{\chi}\}$ to
$\{\hat T_q(\hat h;k)\pm \widehat{\se}(\hat h;k)\cv_{1-\alpha}(t)\}$ satisfies
the lower bound
\begin{multline*}
  \lim_{C\to\infty}\liminf_n \sup_{f\in\mathcal{G}(M)}
    \frac{E_{f} \min\{2n^{r/2}\hat \chi, C\}}
    {E_{f} \min\{2n^{r/2}\widehat{\se}(\hat k;k)\cv_{1-\alpha}(t), C\}} \\
  \ge
    \frac{\int_{z=-\infty}^{z_{1-\alpha}}(z_{1-\alpha}-z)^{r}\, d\Phi(z)}
    {r^r(1-r)^r \inf_{\tilde t} \tilde t^{r-1}\cv_{1-\alpha}(\tilde t)}
    \cdot \frac{S(k^*)^{r} B(k^*)^{1-r}}{S(k)^{r} B(k)^{1-r}}
    \cdot \frac{\inf_{\tilde t} \cv_{1-\alpha}(\tilde t)}{t^{r-1}\cv_{1-\alpha}(t)}.
\end{multline*}
The first term is the lower bound in Theorem E.1 of \citet{ArKo18optimal}, which
corresponds to the lower bound in Corollary 3.3 of that paper applied to the
case where the modulus $\omega(\delta)$ is proportional to $\delta^r$ (as is the
case in the relevant limiting experiment in the present setting; see
\Cref{est_h_limiting_model_sec}). The second term is the relative
efficiency of the kernel $k$, and the final term is the efficiency of the
bias-sd ratio used in the bandwidth $\hat h$ relative to the optimal bias-sd
ratio for FLCI construction.

We now prove \Cref{est_h_lower_bound_thm}. We begin by noting some
properties of the optimal kernel $k^*$.

\begin{lemma}\label{opt_kern_lemma_est_h}
  Let $\kappa^*$ solve
\begin{equation*}
  \max_{\kappa}\kappa(0)\quad\text{s.t.}\quad \int_{\mathcal{X}}\kappa(u)^2\, du\le 1,\, \kappa\in\mathcal{F}(1)
\end{equation*}
and let $k^*(x)=\kappa^*(x)/\int_{\mathcal{X}}\kappa(u)\, du$. Then $k^*$ has
finite support, and it minimizes $S(k)^{r}B(k)^{1-r}$ over kernels $k$.
Furthermore, $S(k^*)=[\sigma^2(0)/d]^{1/2}r \kappa^*(0)$ and
$B(k^*)=(1-r)\kappa^*(0)$, so that
$S(k^*)^{r}B(k^*)^{1-r}=[\sigma^2(0)/d]^{r/2}r^r(1-r)^{1-r}\kappa^*(0)$.
\end{lemma}
\begin{proof}
  The result follows from \citet{low95} and \citet{DoLo92}. See
  \Cref{est_h_optimal_kernel_sec}.
\end{proof}

The next lemma uses functions constructed from $\kappa^*$ to derive testing bounds.

\begin{lemma}\label{asymptotic_power_lemma}
  Suppose that the conditions of \Cref{est_h_lower_bound_thm} hold. Given
  $c\in\mathbb{R}$, let
  $\mathcal{K}_{c, n}=\{f\colon f(0)=c n^{-p/(2p+1)}\}\cap \mathcal{G}(\SC)$. Then, if
  the constant $K$ used to define $\mathcal{G}(\SC)$ is larger than a constant
  that depends only on $\varepsilon$ and $M$, there exists a sequence of
  functions $\tilde \kappa_{0,n}\in \mathcal{K}_{0,n}$ such that the following
  holds. For any $c\in\mathbb{R}$ and any sequence of tests with asymptotic size
  $\alpha$ under $\mathcal{K}_{c, n}$, the asymptotic power under
  $\tilde\kappa_{0,c}$ is no greater than
  \begin{equation*}
    \Phi\left(|c/\kappa^*(0)|^{(2p+1)/(2p)}\SC^{-1/(2p)}[d/\sigma^2(0)]^{1/2} -
      z_{1-\alpha} \right).
  \end{equation*}
\end{lemma}
\begin{proof}
  It suffices to prove the result for $c>0$. Let $A$ and $b_f$ be defined as in
  the proof of \Cref{Jf_consistency_lemma}, so that the coefficients of the
  minimum mean squared error $\tilde p$th order polynomial predictor are given
  by $A^{-1}b_f$. We first note that, under the conditions of the lemma, there
  exist bounded functions $f_1,\ldots, f_{\tilde p+1}$ supported on
  $\mathbb{R}\backslash [-\eta, \eta]$ such that the vectors
  $b_{f_1}, \ldots, b_{f_{\tilde p+1}}$ are linearly independent. Thus, these
  vectors span $\mathbb{R}^{\tilde p+1}$, which means that there exist functions
  $g_1,\ldots, g_{f_{\tilde p+1}}$, which are linear combinations of the $f_j$'s
  (and therefore also bounded and supported on
  $\mathbb{R}\backslash [-\eta, \eta]$) such that $b_{g_j}=e_j$ for each $j$,
  where $e_j$ denotes the $j$th standard basis vector.

We construct functions in the sets $\mathcal{K}_{c, n}$ as follows.
Let $\tilde g$ be a bounded function supported on $\mathbb{R}\backslash [-\eta, \eta]$
such that $J(\tilde g)=\varepsilon\SC$.  This function can be constructed by
finding a polynomial such that the supremum of the $p$th derivative over
$[x_{\min}, x_{\max}]$ is equal to $\varepsilon\SC$, and constructing a function
with the given polynomial predictor coefficients as a linear combination of the $g_j$s defined above.
Given a function $f$ supported on $[-\eta, \eta]$, the function
$b_{f,1}g_1+b_{f,2}g_2+\cdots+b_{f, \tilde p+1}g_{\tilde p+1}$
is supported on $\mathbb{R}\backslash [-\eta, \eta]$ and has the same polynomial
predictor coefficients as $f$.  Thus, the function
$f-(b_{f,1}g_1+b_{f,2}g_2+\cdots+b_{f, \tilde p+1}g_{\tilde p+1})+\tilde g$
has the same polynomial predictor coefficients as $\tilde g$.  It therefore
follows that, if $f\in\mathcal{F}(\SC;\eta)$ and $K$ is larger than some
constant that depends only on an upper bound for the elements of $b_{f}$ and the
functions $g_1,\ldots, g_{\tilde p+1}$ and $\tilde g$, this function will be in
$\mathcal{G}(\SC)$.

Let $\tilde \kappa_{c, \SC, n}$ be defined in this way with
the function $\kappa_{c, \SC, n}$ playing the role of $f$, where
$\kappa_{c, \SC, n}(x)=\SC h_{c, n}^p\kappa^*(x/h_{c, n})$ with $h_{c, n}=\tilde c
n^{-1/(2p+1)}$ where $\tilde c=|c/[M\kappa^*(0)]|^{1/p}$.
Note that $\kappa_{c, \SC, n}\in\mathcal{F}(\SC)$ by the renormalization property of
Taylor and Hölder classes.  Thus, once $n$ is large enough that the support
of $\kappa_{c, \SC, n}$ is contained in $[-\eta, \eta]$, we will have
$\tilde\kappa_{c, \SC, n}\in \mathcal{K}_{c, n}$.

It follows that, for large enough $n$, the power under $\tilde\kappa_{0,\SC, n}$ of a level $\alpha_n$ test of $\mathcal{K}_{c, n}$ is
bounded by the power under $\tilde\kappa_{0,\SC, n}$ of a test with rejection
probability no greater than $\alpha_n$ under $\tilde\kappa_{c, \SC, n}$.
By the Neyman-Pearson lemma and standard calculations, this is no greater than
$\Phi(s_n-z_{1-\alpha_n})$ where
\begin{multline*}
  s_n^2=\sum_{i=1}^n \left[\tilde\kappa_{c, \SC, n}(x_i)-\tilde\kappa_{0,\SC, n}(x_i) \right]^2\sigma^{-2}(x_i) =
  \SC^2 h_{c, n}^{2p}\sum_{i=1}^n \kappa^*(x_i/h_{c, n})^2\sigma^{-2}(x_i)\\
  + \sum_{i=1}^n \left[\sum_{j=1}^{\tilde p+1} g_{j}(x_i)\sigma^{-2}(x_i) \int \SC h_{c, n}^{p}\kappa^*(u/h_{c, n})u^{j-1}d(u)\, du \right]^2.
\end{multline*}
Note that
$h_{c, n}^{2p}= \tilde c^{2p} n^{-2p/(2p+1)}
=n^{-1}\tilde c^{2p+1}n^{1/(2p+1)}\tilde c^{-1}
= (n \tilde c n^{-1/(2p+1)})^{-1}\tilde c^{2p+1}$.
Thus, the first term equals
$\tilde c^{2p+1}\SC^2\frac{1}{n h_{c, n}}\sum_{i=1}^n \kappa^*(x_i/h_{c, n})^2\sigma^2(x_i)
  \to \sigma^{-2}(0)\tilde c^{2p+1}\SC^2d \int_{\mathcal{X}} \kappa^*(u)^2\, du$.
The last term is bounded from above by a constant times
\begin{equation*}
  n \left[h_{c, n}^p \int\kappa^*(u/h_{c, n})\, du \right]^2
  = n \left[h_{c, n}^{p+1} \int\kappa^*(v)\, dv \right]
  = n^{1-(2p+2)/(2p+1)} \tilde c^{(2p+2)/p} \left[\int\kappa^*(u)\, du \right]^2
  \to 0.
\end{equation*}
The result then follows by plugging in $\tilde c$ and noting $\int\kappa^*(u)^2\, du=1$.
\end{proof}

To derive the lower bound on expected length, we argue as in the proof of
Theorem C.2 in \citet{ArKo18sensitivity}. Consider the set
$\mathcal{I}(m)=\{\tilde c_n j/m: j\in \mathbb{Z}, \, |j|\le m^2\}$ where
$\tilde c_n=\kappa^*(0)\SC^{1/(2p+1)}[\sigma^2(0)/d]^{p/(2p+1)}n^{-p/(2p+1)}$.
Let $\hat T\pm \hat \chi$ be a CI with asymptotic coverage at least $1-\alpha$
over $\mathcal{G}(\SC)$, and let $\mathcal{N}(n, m)$ denote the number of
elements in $\mathcal{I}(m)$ that are in this confidence interval. Note that
$\min\{2\hat \chi, 2 \tilde c_n m\}\ge \tilde c_n [\mathcal{N}(m, n)-1]/m$. Let
$\kappa_{0,n}$ and $\mathcal{K}_{c, n}$ be as defined in
\Cref{asymptotic_power_lemma}. Let $\psi_{n, j}$ denote the test that
rejects when the point $\tilde c_n j/m\in\mathcal{N}(n, m)$ is not in the CI
$\hat T\pm \hat{\chi}$. Then $\psi_{n, j}$ is an asymptotically level $\alpha$
test of $\mathcal{K}_{c, n}$, so, by \Cref{asymptotic_power_lemma},
\begin{equation*}
  E_{\kappa_{0,n}} \mathcal{N}(m, n)
  =\sum_{j=-m^2}^{m^2} (1-E_{\kappa_{0,n}}\psi_{n, j})
  \ge \sum_{j=-m^2}^{m^2} (1-\Phi(|j/m|^{(2p+1)/2p}-z_{1-\alpha})) + o(1).
\end{equation*}
Thus, for all $m\in\mathbb{N}$,
$\lim_{C\to\infty}\liminf_{n}E_{\kappa_{0,n}}\min\{2\tilde c_n^{-1}\hat \chi, C\}$
is bounded from below by
\begin{multline*}
  \frac{1}{m}\sum_{j=-m^2}^{m^2} \Phi(z_{1-\alpha} - |j/m|^{(2p+1)/(2p)})
  = \frac{1}{m}\sum_{j=-m^2}^{m^2} \int \1{|j/m|^{(2p+1)/(2p)}\le z_{1-\alpha}-z}\, d\Phi(z) \\
  = \frac{1}{m}\sum_{j=-m^2}^{m^2} \int \1{|j|\le (z_{1-\alpha}-z)^{2p/(2p+1)} m}\, d\Phi(z) \\
  \ge
  \int_{z=-\infty}^{z_{1-\alpha}}\frac{1}{m}\min\left\{2\left[(z_{1-\alpha}-z)^{2p/(2p+1)}
      m - 1 \right], m\right\}\, d\Phi(z).
\end{multline*}
This converges to
$2\int_{z=-\infty}^{z_{1-\alpha}}(z_{1-\alpha}-z)^{2p/(2p+1)}\, d\Phi(z)$
by the Dominated Convergence Theorem.
Thus,
\begin{multline*}
  \lim_{C\to\infty}\liminf_n E_{\kappa_{0,n}} \min\{2n^{p/(2p+1)}\hat \chi, C\} \\
  \ge 2 \kappa^*(0) M^{1/(2p+1)}[\sigma^2(0)/d]^{p/(2p+1)}
  \int_{z=-\infty}^{z_{1-\alpha}}(z_{1-\alpha}-z)^{2p/(2p+1)}\, d\Phi(z) \\
  = 2 \kappa^*(0) M^{1-r}[\sigma^2(0)/d]^{r/2}
  \int_{z=-\infty}^{z_{1-\alpha}}(z_{1-\alpha}-z)^{r}\, d\Phi(z).
\end{multline*}
Plugging in
$S(k^*)^{r}B(k^*)^{1-r}=[\sigma^2(0)/d]^{r/2}r^r(1-r)^{1-r}\kappa^*(0)$
gives the result.

\subsection{Limiting model and optimal kernel}\label{est_h_limiting_model_sec}

In this \namecref{est_h_limiting_model_sec} we derive the properties of the optimal kernel given in
\Cref{opt_kern_lemma_est_h}. To do so, we apply results from \citet{low95}
and \citet{DoLo92} to the limiting model
\begin{equation}\label{est_h_limiting_model_eq}
  Y(dt) = f(t)\, dt + \lambda W(dt), \quad t\in\mathcal{X}
\end{equation}
where $\mathcal{X}=\mathbb{R}$ in the case where the point of interest is on the
interior of the support of $x_i$ and $\mathcal{X}=\hor{0, \infty}$ when it is on
the boundary. We also use this limiting model to give some intuitive motivation
for the efficiency bound in \Cref{est_h_lower_bound_thm}.

The white noise model~\eqref{est_h_limiting_model_eq} is the same model as in
\Cref{white_noise_sec}, with $\lambda$ playing the role of
$\sigma/\sqrt{n}$ in that \namecref{white_noise_sec}. \citet{BrLo96} establish a formal sense in
which this white noise model, with $\lambda$ replaced by the function
$\lambda_n(t)=[\sigma^2(t)/(n d(t))]^{1/2}$, is asymptotically equivalent to the
fixed design regression model. Since the asymptotic behavior of our estimators
and bounds depends only on $x_i$ in a shrinking neighborhood of zero, we then
expect that $\lambda_n(t)$ can be replaced by the constant function
$\lambda_n(0)$. For technical reasons, however, the proof of
\Cref{est_h_lower_bound_thm} uses direct arguments, rather than appealing
to the equivalence results of \citet{BrLo96} (in particular, these results do
not apply immediately for Taylor classes, or when smoothness is only assumed in
the neighborhood $[-\eta, \eta]$).

\subsubsection{Kernel estimators}

Let $k$ be a kernel with $\int_{\mathcal{X}} k(u)\, du=1$ and
$\int_{\mathcal{X}} k(u)u^j\, du=0$ for $j=1,\ldots,p-1$. The kernel $k$ will
play the role of the equivalent kernel $k^*_q$ in
\begin{NoHyper}\Cref{sec:inference-point-theory}\end{NoHyper}. A linear estimator in
the white noise model takes the form
\begin{equation*}
\hat T(h;k)=h^{-1}\int k(t)\, dY(t).
\end{equation*}
Since this falls into the \citet{DoLo92} framework given in
\Cref{white_noise_sec}, it follows that
\begin{NoHyper}\Cref{bias_var_scale_eq}\end{NoHyper}
holds with the $o(1)$ terms equal to zero. Indeed, under $f\in\mathcal{F}(\SC)$,
$\hat T(h;k)$ follows a normal distribution with bias
\begin{equation*}
h^{-1}\int_{\mathcal{X}} k(t/h)(f(t)-f(0))\, dt
  = \int_{\mathcal{X}} k(u)(f(h u)-f(0))\, du
  = M h^{p} \int_{\mathcal{X}} k(u)(\tilde f(u)-\tilde f(0))\, du
\end{equation*}
where $\tilde f(u)=\SC^{-1}h^{-p}f(h u)$ is in $\mathcal{F}(1)$ iff. $f\in
\mathcal{F}(\SC)$, by the renormalization property of the Hölder and Taylor class.
The variance is given by
\begin{equation*}
  \lambda^2 h^{-2} \int_{\mathcal{X}} k(t/h)^2\, dt
  = \lambda^2h^{-1} \int_{\mathcal{X}} k(u)^2\, du.
\end{equation*}
Thus, if we take $\lambda=[\sigma^2(0)/(n d)]^{1/2}$,
\begin{NoHyper}\Cref{bias_var_scale_eq}\end{NoHyper}
holds with
$S(k)=\sigma(0)d^{-1/2}\sqrt{\int_{\mathcal{X}}k(u)\, du}$,
$B(k)=\sup_{\tilde f\in\mathcal{F}(1)}\int_{\mathcal{X}}
k(u)(\tilde{f}(u)-\tilde{f}(0))\, du$, $\gamma_b=p$ and $\gamma_s=-1/2$. Note
that $S(k)$ matches Equation (5) with $k$ playing the role of the equivalent
kernel $k^*_q$ in Equation (5). In addition, $B(k)$ matches the expression given
in \begin{NoHyper}\Cref{theorem:maximum-bias-lp}\end{NoHyper}
(this can be shown by deriving $B(k)$
using the arguments in the proof of this theorem).

\subsubsection{Modulus of continuity}

The modulus of continuity for the limiting model, as defined in
\citet{donoho94}, is given by
\begin{equation*}
  \omega(\delta) = 2\sup_{f} f(0)\quad\text{s.t.}\quad \int_{\mathcal{X}}f(x)^2\, dx\le \delta^2/4,\quad f\in\mathcal{F}(\SC).
\end{equation*}
Let $f^*_{\delta, \SC}$ denote the solution to this problem. Note that the
function $\kappa^*$ defined in \Cref{opt_kern_lemma_est_h} is given by
$f^*_{2,1}$. By \citet{DoLo92}, we have
$f^*_{\delta, \SC}(x)=\SC \tilde h_{\delta, M}^p\kappa^*(x/\tilde h_{\delta, M})$
where $\tilde h_{\delta, \SC}=(\delta/(2\SC))^{2/(2p+1)}$, which gives
\begin{equation*}
\omega(\delta)=2\SC (\delta/(2\SC))^{2p/(2p+1)}\kappa^*(0)
=(2\SC)^{1-r} \delta^{r}\kappa^*(0)
\end{equation*}
where $r=2p/(2p+1)$ is the rate exponent.
Note that
\begin{equation*}
\omega'(\delta)
  =r(2\SC)^{1-r} \delta^{r-1}\kappa^*(0)
  =r\delta^{-1}\omega(\delta).
\end{equation*}

\subsubsection{Optimal kernel}\label{est_h_optimal_kernel_sec}

By \citet{low95}, the bias-sd optimizing kernel takes the form
$t\mapsto f^*_{\delta, \SC}(t)/\int_{\mathcal{X}}f^*_{\delta, \SC}(u)\, du$ for
some $\delta$, so this implies that
$k^*(t)=\kappa^*(t)/\int_{\mathcal{X}}\kappa^*(u)\, du$ is the optimal kernel.
For Taylor classes, the support can be seen to be compact by examining the
formula given in
\begin{NoHyper}\Cref{sec:inference-point-theory}\end{NoHyper}. For Hölder
classes, this can be shown indirectly \citep[see][]{lt00}. The worst-case bias
of the estimate with bandwidth $h_{\delta, \SC}$ is given by
\begin{equation*}
(1/2)(\omega(\delta)-\delta \omega'(\delta))
  =(1/2)\omega(\delta)(1-r)
  =(1/2)(1-r)(2\SC)^{1-r} \delta^{r}\kappa^*(0)
  =\SC (1-r) \kappa^*(0)h_{\delta, \SC}^p
\end{equation*}
where we substitute $\delta=2\SC h_{\delta, \SC}^{(2p+1)/2}$ in the last step.
This gives the formula $B(k^*)=(1-r)\kappa^*(0)$.
The standard deviation is given by
\begin{equation*}
 \lambda\omega'(\delta)=\lambda r(2\SC)^{1-r} \delta^{r-1}\kappa^*(0)
  = \lambda r \kappa^*(0) h_{\delta, \SC}^{-1/2}
  = [\sigma^2(0)/d]^{1/2}r \kappa^*(0) n^{-1/2}h_{\delta, \SC}^{-1/2},
\end{equation*}
which gives $S(k^*)=[\sigma^2(0)/d]^{1/2}r \kappa^*(0)$.
Thus, the leading term in the minimax performance is
$S(k^*)^{r}B(k^*)^{1-r}=[\sigma^2(0)/d]^{r/2}r^r(1-r)^{1-r}\kappa^*(0)$.

\subsubsection{Optimal FLCI and efficiency bound}

We now show that the efficiency bound in \Cref{est_h_lower_bound_thm}
corresponds to the bound given in Corollary 3.3 in \citet{ArKo18optimal},
applied to the class $\mathcal{F}$ in the limiting
model~\eqref{est_h_limiting_model_eq}. Thus, \Cref{est_h_lower_bound_thm}
can be interpreted as showing that this efficiency bound holds in a formal
asymptotic sense, with $\mathcal{F}(M;\eta)$ replaced by the smaller class
$\mathcal{G}(M)$. We note that, for Taylor classes, such a bound is given for
the class $\mathcal{F}(M)$ in Theorem E.1 in \citet{ArKo18optimal}.
\Cref{est_h_lower_bound_thm} shows that this efficiency bound holds for
$\mathcal{G}(M)$.

First, we derive the length of the optimal FLCI, which is the denominator of the
expression in Corollary 3.3 in \citet{ArKo18optimal}.
The bias-sd ratio is
\begin{equation*}
  t_\delta
   = \frac{(1/2)(1-r)(2\SC)^{1-r} \delta^{r}\kappa^*(0)}{\lambda r(2\SC)^{1-r} \delta^{r-1}\kappa^*(0)}
  =(1/2)(1/r-1)\delta/\lambda.
\end{equation*}
Since optimizing over the bandwidth is equivalent to optimizing over $\delta$,
it follows that the optimal FLCI has length
\begin{multline*}
  \inf_{\delta} 2\cv_{1-\alpha}(t_{\delta})\cdot \lambda\omega'(\delta)
  = \inf_{\delta} 2\cv_{1-\alpha}(t_{\delta})\cdot \lambda r(2\SC)^{1-r} \delta^{r-1}\kappa^*(0) \\
  = \inf_{\delta} 2\cv_{1-\alpha}(t_{\delta})\cdot \lambda r(2\SC)^{1-r} t_{\delta}^{r-1}\lambda^{r-1}(1/r-1)^{1-r}2^{r-1}\kappa^*(0) \\
  = \lambda^r \SC^{1-r}r(1/r-1)^{1-r}\kappa^*(0) \inf_{\delta} 2\cv_{1-\alpha}(t_{\delta})\cdot t_{\delta}^{r-1}.
\end{multline*}
Plugging in $\lambda=[\sigma^2(0)/(n d)]^{1/2}$ and
$S(k^*)^{n}B(k^*)^{1-r}=[\sigma^2(0)/d]^{r/2}r^r(1-r)^{1-r}\kappa^*(0)$ gives
$2 n^{-r/2}M^{1-r}S(k^*)^{r}B(k^*)^{1-r}\inf_{\delta}
\cv_{1-\alpha}(t_{\delta})\cdot t_\delta^{r-1}$, which is the asymptotic length
of the CI given in \Cref{general_est_h_theorem} with $k$ and $h$ chosen
optimally.

The lower bound given the numerator of the expression in Corollary 3.3 in \citet{ArKo18optimal} is
\begin{equation*}
  \int_{z=-\infty}^{z_{1-\alpha}} \omega(2\lambda (z_{1-\alpha}-z))\, dz
  = (2\SC)^{1-r} \kappa^*(0) 2^r\lambda^r
  \int_{z=-\infty}^{z_{1-\alpha}} (z_{1-\alpha}-z)^r\, dz.
\end{equation*}
Plugging in $\lambda=[\sigma^2(0)/(n d)]^{1/2}$ and
$S(k^*)^{r}B(k^*)^{1-r}=[\sigma^2(0)/d]^{r/2}r^r(1-r)^{1-r}\kappa^*(0)$ gives
$2 n^{-r/2}M^{1-r}\frac{S(k^*)^{r}
  B(k^*)^{1-r}}{r^r(1-r)^{1-r}}\int_{z=-\infty}^{z_{1-\alpha}}
(z_{1-\alpha}-z)^r\, dz$, which is the asymptotic lower bound given in
\Cref{est_h_lower_bound_thm}.

\section{Additional Monte Carlo results}\label{sec:addit-monte-carlo}

In this \namecref{sec:addit-monte-carlo}, we revisit the simulation study from
\begin{NoHyper}\Cref{monte_carlo_sec}\end{NoHyper} in the paper, and
consider an additional method for constructing CIs, as well as a number of
variations on the DGP\@.

In particular, we also consider a conventional CI based on the coverage-error
optimal bandwidth $\hcedpi$, which can be considered a form of undersmoothing,
but without any bias correction. \Cref{tab:mc2} reports the results for Designs
1--3 with this additional methods added. Using the bandwidth $\hcedpi$ leads to
better coverage of conventional CIs relative to $\hfgrot$ when $\SC=2$, but
worse coverage when $\SC=6$.

Next, we investigate the robustness of the results to a number of variations on
the baseline design. \Cref{tab:mc3} reports the results when $x_{i}$ is
drawn from a $\operatorname{Beta}(2, 5)$ distribution. In \Cref{tab:mc4},
to consider the effects of heteroskedasticity, we draw the errors form the
distribution $\mathcal{N}(0, 1/4(1+\sqrt{\abs{x_{i}}})^{2})$, while $x_{i}$ is
drawn from a uniform distribution, as in the baseline. In \Cref{tab:mc5},
$x_{i}\sim\operatorname{Beta}(2, 5)$ distribution, and
$u_{i}\sim\mathcal{N}(0, 1/4(1+\sqrt{\abs{x_{i}}})^{2})$. In
\Cref{tab:mc6}, we draw $u_{i}$ from a log-normal distribution, scaled to
have mean zero and variance $1/4$, while $x_{i}$ is drawn from a uniform
distribution. \Cref{tab:mc7} reports the results for $u_{i}$ drawn from a
log-normal distribution, scaled to have mean zero and variance $1/4$, and
$x_{i}\sim\operatorname{Beta}(2, 5)$. \Cref{tab:mc8} returns to the
baseline specification, but with $u_{i}\sim \mathcal{N}(0,1/16)$. Finally, in
\Cref{tab:mc9} we consider a smooth approximation to the functions
$f_{1}, f_{2}$, and $f_{3}$. In particular, we replace the function
$\textsf{s}(\cdot)$ in the definition of these functions by the function
$s_{\lambda}(x)=-\Li_{2}(-e^{\lambda x})/\lambda^{2}$, where
$\Li_{2}(x)=-\int_{0}^{x}\frac{\log(1-s)}{s}ds$ is the dilogarithm function.
The function $s_{\lambda}$ is analytic for any $\lambda$, and it converges to
$\textsf{s}$ as $\lambda\to\infty$. We set $\lambda=40$.

The results in \Cref{tab:mc9} are nearly identical to those in \Cref{tab:mc2},
indicating that the lack of differentiability is not driving the results. The
FLCIs perform well for all designs in terms of coverage when the correct or
conservative $\SC$ is used, or when one uses $\Mrot$. The coverage is at least
92.5\% in all designs except \Cref{tab:mc6}, where the coverage, where the FLCIs
undercover slightly for Design 3, with coverage around 90\%.
The RBC CIs with bandwidth chosen based on uniform-in-$f$ asymptotics (either
$\hhrmse{2}$, $\hhrmse{6}$, or $\hhrmse{\Mrot}$) also perform well in terms of
coverage, with coverage at least 93\% for all designs, although they are longer
than FLCI CIs.
The remaining CIs, based on pointwise-in-$f$ asymptotics, suffer from poor
coverage in these alternative specifications, just like in the baseline
specification in the main text.

\end{appendices}

\bibliography{../np-testing-library}

\begin{table}[p]
  \centering
  \caption{Kernel constants for standard deviation and maximum bias of
    local polynomial regression estimators of order $\po$ for selected kernels.
    Inference at a boundary point}\label{tab:lp-boundary-constants}
  \vspace{1ex}
  \begin{tabular}{@{}lclllllll@{}}
    & & & \multicolumn{3}{c}{$\mathcal{B}_{p, \po}^{\textnormal{T}}(k)=\int_{0}^{1}
      \abs{u^{p}k^{*}_{\po}(u)}\, \dd u$} &\multicolumn{3}{c}{
      $\mathcal{B}_{p, \po}^{\textnormal{Höl}}(k)$}\\
    \cmidrule(rl){4-6}\cmidrule(rl){7-9}
    Kernel ($k(u)$) & $\po$ & \multicolumn{1}{c}{$\int_{0}^{1}
      k^{*}_{\po}(u)^{2}\, \dd u$} & $p=1$ & $p=2$ & $p=3$& $p=1$ & $p=2$ & $p=3$\\
    \midrule
    \multirow{3}{*}{\begin{tabular}{l}Uniform\\ $\1{\abs{u}\leq 1}$\end{tabular}}
    &0 & 1 & $\frac{1}{2}$ &  &                 & $\frac{1}{2}$  \\
    &1 & 4 & $\frac{16}{27}$ & $\frac{59}{162}$ & &
    $\frac{8}{27}$ & $\frac{1}{6}$  \\
    &2 & 9 & 0.7055 & 0.4374 & 0.3294 & 0.2352& $\frac{216}{3125}$  &$\frac{1}{20}$\\[1ex]
    \multirow{3}{*}{\begin{tabular}{l}Triangular\\ $(1-\abs{u})_{+}$\end{tabular}}
    &0 & $\frac{4}{3}$ & $\frac{1}{3}$ &  &
    & $\frac{1}{3}$ & & \\
    &1 & $\frac{24}{5}$ &$\frac{3}{8}$ & $\frac{3}{16}$ &
    & $\frac{27}{128}$& $\frac{1}{10}$\\
    &2 & $\frac{72}{7}$ & 0.4293 & 0.2147 & 0.1400 &
    0.1699 & $\frac{32}{729}$ &$\frac{1}{35}$ \\[1ex]
    \multirow{3}{*}{\begin{tabular}{l}Epanechnikov\\ $\frac{3}{4}(1-u^{2})_{+}$\end{tabular}}
    &0 & $\frac{6}{5}$ & $\frac{3}{8}$ &  & &
    $\frac{3}{8}$\\
    &1 & 4.498 & 0.4382 & 0.2290 & &
    0.2369& $\frac{11}{95}$  &\\
    &2 & 9.816 & 0.5079 & 0.2662 & 0.1777 &
    0.1913& 0.0508  &$\frac{15}{448}$\\[1ex]
  \end{tabular}
\end{table}

\begin{table}[p]
  \centering
  \caption{Kernel constants for standard deviation and maximum bias of
    local polynomial regression estimators of order $\po$ for selected kernels.
    Inference at an interior point.}\label{tab:lp-interior-constants}
  \vspace{1ex}
 \begin{tabular}{@{}lclllllll@{}}
    & & & \multicolumn{3}{c}{$\mathcal{B}_{p, \po}^{\textnormal{T}}(k)=\int_{-1}^{1}
      \abs{u^{p}k^{*}_{\po}(u)}\, \dd u$} &\multicolumn{3}{c}{
      $\mathcal{B}_{p, \po}^{\textnormal{Höl}}(k)$}\\
    \cmidrule(rl){4-6}\cmidrule(rl){7-9}
    Kernel & $\po$ & \multicolumn{1}{c}{$\int_{-1}^{1} k^{*}_{\po}(u)^{2}\, \dd u$} & $p=1$ & $p=2$ & $p=3$& $p=1$ & $p=2$ & $p=3$\\
    \midrule
    \multirow{3}{*}{\begin{tabular}{l}Uniform\\ $\1{\abs{u}\leq 1}$\end{tabular}}
    &0 & $\frac{1}{2}$ & $\frac{1}{2}$   &      &     &
    $\frac{1}{2}$\\
    &1 & $\frac{1}{2}$ & $\frac{1}{2}$ & $\frac{1}{3}$     &     &
    $\frac{1}{2}$ & $\frac{1}{3}$\\
    &2 & $\frac{9}{8}$ & 0.4875& 0.2789  & 0.1975 &
    0.2898& 0.0859  &  $\frac{1}{16}$\\[1ex]
    \multirow{3}{*}{\begin{tabular}{l}Triangular\\ $(1-\abs{u})_{+}$\end{tabular}}
    &0 & $\frac{2}{3}$ & $\frac{1}{3}$   &      &    &$\frac{1}{3}$
    \\
    &1 & $\frac{2}{3}$ & $\frac{1}{3}$   & $\frac{1}{6}$     & &
    $\frac{1}{3}$  & $\frac{1}{6}$     &    \\
    &2 &$\frac{456}{343}$&0.3116&0.1399  & 0.0844 &
    0.2103&0.0517  &$\frac{8}{245}$\\[1ex]
    \multirow{3}{*}{\begin{tabular}{l}Epanechnikov\\ $\frac{3}{4}(1-u^{2})_{+}$\end{tabular}}
    &0 & $\frac{3}{5}$ &$\frac{3}{8}$    &      &     &
    $\frac{3}{8}$\\
    &1 & $\frac{3}{5}$ & $\frac{3}{8}$    & $\frac{1}{5}$     &     &
    $\frac{3}{8}$    & $\frac{1}{5}$     &    \\
    &2 & $\frac{5}{4}$ &0.3603 & 0.1718  & 0.1067 &
    0.2347& 0.0604  &$\frac{5}{128}$ \\[1ex]
  \end{tabular}
\end{table}

\begin{landscape}
  \renewcommand{\arraystretch}{1.1} %
  \footnotesize
  \begin{ThreePartTable}
    \begin{TableNotes}
    \item \emph{Legend:} SE---average standard error; $E[h]$---average (over
      Monte Carlo draws) bandwidth; Cov---coverage of CIs (in \%); RL---relative
      (to optimal FLCI) length.
    \item \emph{Bandwidth descriptions:} $\hmsedpi$---plugin estimate of
      pointwise MSE optimal bandwidth (bw); $\bmsedpi$---analog for estimate of
      the bias; $\hcedpi$---plugin estimate of coverage error optimal bw;
      $\bcedpi$---analog for estimate of the bias; The implementation of
      \citet{ccf15} is used for all four bws. $\hhrmse{2}$, $\hhrmse{6}$---RMSE
      optimal bw, assuming $\SC=2$, and $\SC=6$, respectively.
      $\hfgrot$---\citet{fg96} rule of thumb; $\hhrmse{\Mrot}$---RMSE optimal
      bw, using rule-of-thumb for $\SC$. 50,000 Monte Carlo draws.
    \end{TableNotes}
    \begin{longtable}{@{}lll rllrl c rllrl ll@{}}
      \caption{Monte Carlo simulation: baseline DGP}\label{tab:mc2}\\
      &&&\multicolumn{5}{c@{}}{$\SC=2$}&&\multicolumn{5}{c}{$\SC=6$}\\
      \cmidrule(rl){4-8}\cmidrule(rl){10-14}
      &Method & Bandwidth & Bias& SE &   $E[h]$&   Cov&   RL&&Bias& SE &   $E_{m}[h]$&   Cov&   RL\\
      \midrule
      \endfirsthead%
      \caption*{Monte Carlo simulation: baseline DGP (continued)}\\
      &&&\multicolumn{5}{c@{}}{$\SC=2$}&&\multicolumn{5}{c}{$\SC=6$}\\
      \cmidrule(rl){4-8}\cmidrule(rl){10-14}
      &Method & Bandwidth & Bias& SE &   $E[h]$&   Cov&   RL&&Bias& SE &   $E_{m}[h]$&   Cov&   RL\\
      \midrule
      \endhead%
      \endfoot%
      \insertTableNotes%
      \endlastfoot%

      \multicolumn{2}{@{}l}{Design 1}\\
      \cmidrule(r){1-2} \phantom{a}
      &RBC         &$h=\hmsedpi$, $b=\bmsedpi$  & 0.063& 0.035& 0.75& 55.6& 0.73&&  0.157& 0.036& 0.62&  0.1& 0.61\\
      &RBC         &$h=b=\hmsedpi$              & 0.025& 0.042& 0.75& 93.1& 0.88&&  0.042& 0.047& 0.62& 89.1& 0.78\\
      &RBC         &$h=\hcedpi$, $b=\bcedpi$    & 0.030& 0.041& 0.45& 85.8& 0.85&&  0.059& 0.045& 0.34& 72.4& 0.76\\
      &RBC         &$h=b=\hhrmse{2}$            & 0.001& 0.061& 0.36& 94.5& 1.27&&  0.002& 0.061& 0.36& 94.5& 1.01\\
      &RBC         &$h=b=\hhrmse{6}$            & 0.000& 0.076& 0.23& 94.2& 1.58&&  0.000& 0.075& 0.23& 94.2& 1.26\\
      &RBC         &$h=b=\hhrmse{\Mrot}$        & 0.000& 0.078& 0.22& 93.9& 1.64&&  0.000& 0.097& 0.14& 93.4& 1.63\\
      &Conventional&$\hfgrot$                   & 0.032& 0.036& 0.56& 76.6& 0.76&&  0.049& 0.046& 0.31& 77.4& 0.77\\
      &Conventional&$\hcedpi$                   & 0.029& 0.039& 0.45& 85.2& 0.81&&  0.058& 0.044& 0.34& 72.3& 0.74\\
      &FLCI, $\SC=2$ &$\hhrmse{2}$              & 0.021& 0.043& 0.36& 94.9& 1.00&&  0.065& 0.043& 0.36& 75.2& 0.80\\
      &FLCI, $\SC=6$ &$\hhrmse{6}$              & 0.009& 0.054& 0.23& 96.6& 1.25&&  0.028& 0.053& 0.23& 94.7& 1.00\\
      &FLCI, $\SC=\Mrot$&$\hhrmse{\Mrot}$       & 0.008& 0.056& 0.22& 95.6& 1.29&&  0.010& 0.069& 0.14& 96.3& 1.30\\
      \multicolumn{2}{@{}l}{Design 2}\\
      \cmidrule(r){1-2}
      &RBC         &$h=\hmsedpi$, $b=\bmsedpi$  & 0.043& 0.035& 0.77& 75.9& 0.72&&  0.129& 0.035& 0.77&  4.6& 0.58\\
      &RBC         &$h=b=\hmsedpi$              & 0.026& 0.041& 0.77& 90.9& 0.87&&  0.077& 0.042& 0.77& 53.0& 0.70\\
      &RBC         &$h=\hcedpi$, $b=\bcedpi$    & 0.028& 0.040& 0.49& 87.4& 0.83&&  0.074& 0.041& 0.44& 54.1& 0.69\\
      &RBC         &$h=b=\hhrmse{2}$            & 0.002& 0.061& 0.36& 94.5& 1.27&&  0.006& 0.061& 0.36& 94.4& 1.01\\
      &RBC         &$h=b=\hhrmse{6}$            & 0.000& 0.076& 0.23& 94.2& 1.58&&  0.000& 0.075& 0.23& 94.2& 1.26\\
      &RBC         &$h=b=\hhrmse{\Mrot}$        & 0.001& 0.068& 0.30& 94.0& 1.43&&  0.000& 0.083& 0.20& 93.8& 1.38\\
      &Conventional&$\hfgrot$                   & 0.032& 0.032& 0.78& 74.4& 0.67&&  0.073& 0.040& 0.44& 53.0& 0.66\\
      &Conventional&$\hcedpi$                   & 0.028& 0.037& 0.49& 85.9& 0.78&&  0.076& 0.039& 0.44& 50.1& 0.66\\
      &FLCI, $\SC=2$ &$\hhrmse{2}$              & 0.020& 0.043& 0.36& 95.1& 1.00&&  0.061& 0.043& 0.36& 78.1& 0.80\\
      &FLCI, $\SC=6$ &$\hhrmse{6}$              & 0.009& 0.054& 0.23& 96.6& 1.25&&  0.028& 0.053& 0.23& 94.7& 1.00\\
      &FLCI, $\SC=\Mrot$&$\hhrmse{\Mrot}$       & 0.013& 0.048& 0.30& 94.3& 1.13&&  0.020& 0.059& 0.20& 94.3& 1.10\\
      \multicolumn{2}{@{}l}{Design 3}\\
      \cmidrule(r){1-2}
      &RBC         &$h=\hmsedpi$, $b=\bmsedpi$  & -0.043& 0.035& 0.77& 75.7& 0.72&& -0.123& 0.035& 0.74&  9.9& 0.59\\
      &RBC         &$h=b=\hmsedpi$              & -0.024& 0.042& 0.77& 90.8& 0.87&& -0.066& 0.043& 0.74& 60.3& 0.71\\
      &RBC         &$h=\hcedpi$, $b=\bcedpi$    & -0.026& 0.040& 0.49& 88.1& 0.83&& -0.063& 0.043& 0.43& 64.2& 0.71\\
      &RBC         &$h=b=\hhrmse{2}$            & -0.002& 0.061& 0.36& 94.5& 1.27&& -0.007& 0.061& 0.36& 94.4& 1.01\\
      &RBC         &$h=b=\hhrmse{6}$            &  0.000& 0.076& 0.23& 94.2& 1.58&&  0.000& 0.075& 0.23& 94.2& 1.26\\
      &RBC         &$h=b=\hhrmse{\Mrot}$        &  0.000& 0.074& 0.25& 94.2& 1.54&&  0.000& 0.092& 0.16& 93.6& 1.54\\
      &Conventional&$\hfgrot$                   & -0.032& 0.033& 0.72& 74.7& 0.69&& -0.065& 0.042& 0.39& 62.0& 0.70\\
      &Conventional&$\hcedpi$                   & -0.028& 0.037& 0.49& 85.7& 0.78&& -0.074& 0.040& 0.43& 52.0& 0.66\\
      &FLCI, $\SC=2$ &$\hhrmse{2}$              & -0.020& 0.043& 0.36& 95.0& 1.00&& -0.060& 0.043& 0.36& 78.1& 0.80\\
      &FLCI, $\SC=6$ &$\hhrmse{6}$              & -0.009& 0.054& 0.23& 96.5& 1.25&& -0.027& 0.053& 0.23& 94.7& 1.00\\
      &FLCI, $\SC=\Mrot$&$\hhrmse{\Mrot}$       & -0.010& 0.052& 0.25& 95.6& 1.22&& -0.013& 0.065& 0.16& 96.1& 1.22\\
    \end{longtable}
  \end{ThreePartTable}
\end{landscape}

\begin{landscape}
  \renewcommand{\arraystretch}{1.1} %
  \footnotesize
  \begin{ThreePartTable}
    \begin{TableNotes}
    \item \emph{Legend:} SE---average standard error; $E[h]$---average (over
      Monte Carlo draws) bandwidth; Cov---coverage of CIs (in \%); RL---relative
      (to optimal FLCI) length.
    \item \emph{Bandwidth descriptions:} $\hmsedpi$---plugin estimate of
      pointwise MSE optimal bandwidth (bw); $\bmsedpi$---analog for estimate of
      the bias; $\hcedpi$---plugin estimate of coverage error optimal bw;
      $\bcedpi$---analog for estimate of the bias; The implementation of
      \citet{ccf15} is used for all four bws. $\hhrmse{2}$, $\hhrmse{6}$---RMSE
      optimal bw, assuming $\SC=2$, and $\SC=6$, respectively.
      $\hfgrot$---\citet{fg96} rule of thumb; $\hhrmse{\Mrot}$---RMSE optimal
      bw, using rule-of-thumb for $\SC$. 50,000 Monte Carlo draws.
    \end{TableNotes}
    \begin{longtable}{@{}lll rllrl c rllrl ll@{}}
      \caption{Monte Carlo simulation: beta distribution
        for $x_{i}$}\label{tab:mc3}\\
      &&&\multicolumn{5}{c@{}}{$\SC=2$}&&\multicolumn{5}{c}{$\SC=6$}\\
      \cmidrule(rl){4-8}\cmidrule(rl){10-14}
      &Method & Bandwidth & Bias& SE &   $E[h]$&   Cov&   RL&&Bias& SE &   $E_{m}[h]$&   Cov&   RL\\
      \midrule
      \endfirsthead%
      \caption*{Monte Carlo simulation: beta distribution
      for $x_{i}$ (continued)}\\
      &&&\multicolumn{5}{c@{}}{$\SC=2$}&&\multicolumn{5}{c}{$\SC=6$}\\
      \cmidrule(rl){4-8}\cmidrule(rl){10-14}
      &Method & Bandwidth & Bias& SE &   $E[h]$&   Cov&   RL&&Bias& SE &   $E_{m}[h]$&   Cov&   RL\\
      \midrule
      \endhead%
      \endfoot%
      \insertTableNotes%
      \endlastfoot%

      \multicolumn{2}{@{}l}{Design 1}\\
      \cmidrule(r){1-2} \phantom{a}
      &RBC         &$h=\hmsedpi$, $b=\bmsedpi$  &0.030& 0.037& 0.56& 85.6& 0.83&&  0.056& 0.041& 0.43& 64.8& 0.74 \\
      &RBC         &$h=b=\hmsedpi$              &0.009& 0.044& 0.56& 93.7& 0.98&&  0.009& 0.050& 0.43& 91.7& 0.92 \\
      &RBC         &$h=\hcedpi$, $b=\bcedpi$    &0.009& 0.044& 0.38& 93.1& 0.99&&  0.011& 0.049& 0.29& 92.6& 0.90 \\
      &RBC         &$h=b=\hhrmse{2}$            &0.001& 0.054& 0.36& 94.6& 1.21&&  0.003& 0.054& 0.37& 94.6& 0.98 \\
      &RBC         &$h=b=\hhrmse{6}$            &0.000& 0.068& 0.23& 94.3& 1.53&&  0.000& 0.068& 0.23& 94.4& 1.24 \\
      &RBC         &$h=b=\hhrmse{\Mrot}$        &0.000& 0.073& 0.21& 94.1& 1.62&&  0.000& 0.089& 0.14& 93.8& 1.61 \\
      &Conventional&$\hfgrot$                   &0.025& 0.038& 0.53& 85.6& 0.84&&  0.038& 0.045& 0.29& 83.8& 0.82 \\
      &Conventional&$\hcedpi$                   &0.019& 0.040& 0.38& 90.3& 0.90&&  0.038& 0.045& 0.29& 82.4& 0.81 \\
      &FLCI, $\SC=2$ &$\hhrmse{2}$              &0.019& 0.041& 0.36& 94.7& 1.00&&  0.058& 0.041& 0.37& 77.0& 0.81 \\
      &FLCI, $\SC=6$ &$\hhrmse{6}$              &0.009& 0.050& 0.23& 96.5& 1.23&&  0.025& 0.050& 0.23& 94.7& 1.00 \\
      &FLCI, $\SC=\Mrot$&$\hhrmse{\Mrot}$       &0.007& 0.053& 0.21& 96.1& 1.31&&  0.009& 0.064& 0.14& 96.3& 1.29 \\
      \multicolumn{2}{@{}l}{Design 2}\\
      \cmidrule(r){1-2}
      &RBC         &$h=\hmsedpi$, $b=\bmsedpi$  & 0.027& 0.037& 0.57& 88.0& 0.82&&  0.073& 0.038& 0.53& 49.0& 0.68\\
      &RBC         &$h=b=\hmsedpi$              & 0.013& 0.043& 0.57& 93.2& 0.97&&  0.032& 0.045& 0.53& 84.3& 0.82\\
      &RBC         &$h=\hcedpi$, $b=\bcedpi$    & 0.014& 0.043& 0.40& 92.7& 0.96&&  0.032& 0.045& 0.36& 84.8& 0.83\\
      &RBC         &$h=b=\hhrmse{2}$            & 0.003& 0.054& 0.36& 94.6& 1.21&&  0.007& 0.054& 0.37& 94.5& 0.98\\
      &RBC         &$h=b=\hhrmse{6}$            & 0.000& 0.068& 0.23& 94.3& 1.53&&  0.000& 0.068& 0.23& 94.4& 1.24\\
      &RBC         &$h=b=\hhrmse{\Mrot}$        & 0.001& 0.068& 0.25& 94.2& 1.51&&  0.001& 0.075& 0.20& 94.0& 1.35\\
      &Conventional&$\hfgrot$                   & 0.026& 0.035& 0.70& 85.1& 0.79&&  0.060& 0.039& 0.43& 62.8& 0.71\\
      &Conventional&$\hcedpi$                   & 0.019& 0.039& 0.40& 90.8& 0.88&&  0.050& 0.041& 0.36& 72.2& 0.75\\
      &FLCI, $\SC=2$ &$\hhrmse{2}$              & 0.018& 0.041& 0.36& 94.9& 1.00&&  0.055& 0.041& 0.37& 79.1& 0.81\\
      &FLCI, $\SC=6$ &$\hhrmse{6}$              & 0.009& 0.050& 0.23& 96.5& 1.23&&  0.025& 0.050& 0.23& 94.7& 1.00\\
      &FLCI, $\SC=\Mrot$&$\hhrmse{\Mrot}$       & 0.009& 0.049& 0.25& 95.7& 1.22&&  0.019& 0.054& 0.20& 94.2& 1.09\\
      \multicolumn{2}{@{}l}{Design 3}\\
      \cmidrule(r){1-2}
      &RBC         &$h=\hmsedpi$, $b=\bmsedpi$  & -0.031& 0.037& 0.55& 86.2& 0.83&& -0.070& 0.039& 0.49& 52.9& 0.71\\
      &RBC         &$h=b=\hmsedpi$              & -0.012& 0.044& 0.55& 93.9& 0.98&& -0.024& 0.047& 0.49& 89.4& 0.85\\
      &RBC         &$h=\hcedpi$, $b=\bcedpi$    & -0.011& 0.044& 0.39& 92.9& 0.99&& -0.018& 0.049& 0.31& 91.3& 0.89\\
      &RBC         &$h=b=\hhrmse{2}$            & -0.002& 0.054& 0.36& 94.6& 1.21&& -0.007& 0.054& 0.37& 94.5& 0.98\\
      &RBC         &$h=b=\hhrmse{6}$            &  0.000& 0.068& 0.23& 94.3& 1.53&&  0.000& 0.068& 0.23& 94.3& 1.24\\
      &RBC         &$h=b=\hhrmse{\Mrot}$        &  0.000& 0.072& 0.22& 94.2& 1.60&&  0.000& 0.085& 0.15& 93.9& 1.54\\
      &Conventional&$\hfgrot$                   & -0.025& 0.036& 0.65& 85.3& 0.80&& -0.051& 0.041& 0.37& 72.1& 0.75\\
      &Conventional&$\hcedpi$                   & -0.018& 0.040& 0.39& 91.3& 0.89&& -0.040& 0.044& 0.31& 81.6& 0.79\\
      &FLCI, $\SC=2$ &$\hhrmse{2}$              & -0.018& 0.041& 0.36& 95.1& 1.00&& -0.054& 0.041& 0.37& 79.6& 0.81\\
      &FLCI, $\SC=6$ &$\hhrmse{6}$              & -0.008& 0.050& 0.23& 96.6& 1.23&& -0.024& 0.050& 0.23& 94.8& 1.00\\
      &FLCI, $\SC=\Mrot$&$\hhrmse{\Mrot}$       & -0.007& 0.052& 0.22& 96.1& 1.29&& -0.011& 0.061& 0.15& 96.3& 1.23\\
    \end{longtable}
  \end{ThreePartTable}
\end{landscape}

\begin{landscape}
  \renewcommand{\arraystretch}{1.1} %
  \footnotesize
  \begin{ThreePartTable}
    \begin{TableNotes}
    \item \emph{Legend:} SE---average standard error; $E[h]$---average (over
      Monte Carlo draws) bandwidth; Cov---coverage of CIs (in \%); RL---relative
      (to optimal FLCI) length.
    \item \emph{Bandwidth descriptions:} $\hmsedpi$---plugin estimate of
      pointwise MSE optimal bandwidth (bw); $\bmsedpi$---analog for estimate of
      the bias; $\hcedpi$---plugin estimate of coverage error optimal bw;
      $\bcedpi$---analog for estimate of the bias; The implementation of
      \citet{ccf15} is used for all four bws. $\hhrmse{2}$, $\hhrmse{6}$---RMSE
      optimal bw, assuming $\SC=2$, and $\SC=6$, respectively.
      $\hfgrot$---\citet{fg96} rule of thumb; $\hhrmse{\Mrot}$---RMSE optimal
      bw, using rule-of-thumb for $\SC$. 50,000 Monte Carlo draws.
    \end{TableNotes}
    \begin{longtable}{@{}lll rllrl c rllrl ll@{}}
      \caption{Monte Carlo simulation: heteroskedastic errors}\label{tab:mc4}\\
      &&&\multicolumn{5}{c@{}}{$\SC=2$}&&\multicolumn{5}{c}{$\SC=6$}\\
      \cmidrule(rl){4-8}\cmidrule(rl){10-14}
      &Method & Bandwidth & Bias& SE &   $E[h]$&   Cov&   RL&&Bias& SE &   $E_{m}[h]$&   Cov&   RL\\
      \midrule
      \endfirsthead%
      \caption*{Monte Carlo simulation: heteroskedastic errors (continued)}\\
      &&&\multicolumn{5}{c@{}}{$\SC=2$}&&\multicolumn{5}{c}{$\SC=6$}\\
      \cmidrule(rl){4-8}\cmidrule(rl){10-14}
      &Method & Bandwidth & Bias& SE &   $E[h]$&   Cov&   RL&&Bias& SE &   $E_{m}[h]$&   Cov&   RL\\
      \midrule
      \endhead%
      \endfoot%
      \insertTableNotes%
      \endlastfoot%

      \multicolumn{2}{@{}l}{Design 1}\\
      \cmidrule(r){1-2} \phantom{a}
      &RBC         &$h=\hmsedpi$, $b=\bmsedpi$  & 0.058& 0.049& 0.69& 78.8& 0.83&&  0.160& 0.050& 0.63&  6.7& 0.70\\
      &RBC         &$h=b=\hmsedpi$              & 0.019& 0.058& 0.69& 94.3& 0.97&&  0.044& 0.060& 0.63& 91.0& 0.84\\
      &RBC         &$h=\hcedpi$, $b=\bcedpi$    & 0.029& 0.054& 0.45& 90.4& 0.91&&  0.065& 0.057& 0.37& 76.3& 0.80\\
      &RBC         &$h=b=\hhrmse{2}$            & 0.003& 0.070& 0.43& 94.5& 1.17&&  0.006& 0.070& 0.42& 94.5& 0.99\\
      &RBC         &$h=b=\hhrmse{6}$            & 0.000& 0.083& 0.28& 94.4& 1.40&&  0.000& 0.084& 0.27& 94.4& 1.18\\
      &RBC         &$h=b=\hhrmse{\Mrot}$        & 0.001& 0.087& 0.27& 94.2& 1.46&&  0.000& 0.105& 0.16& 93.8& 1.47\\
      &Conventional&$\hfgrot$                   & 0.038& 0.045& 0.73& 81.1& 0.76&&  0.074& 0.053& 0.40& 67.6& 0.75\\
      &Conventional&$\hcedpi$                   & 0.028& 0.051& 0.45& 89.8& 0.86&&  0.064& 0.055& 0.37& 75.1& 0.77\\
      &FLCI, $\SC=2$ &$\hhrmse{2}$              & 0.029& 0.052& 0.43& 95.1& 1.00&&  0.083& 0.052& 0.42& 73.6& 0.83\\
      &FLCI, $\SC=6$ &$\hhrmse{6}$              & 0.013& 0.061& 0.28& 97.3& 1.20&&  0.036& 0.062& 0.27& 94.8& 1.00\\
      &FLCI, $\SC=\Mrot$&$\hhrmse{\Mrot}$       & 0.012& 0.064& 0.27& 96.4& 1.25&&  0.013& 0.077& 0.16& 96.9& 1.25\\
      \multicolumn{2}{@{}l}{Design 2}\\
      \cmidrule(r){1-2}
      &RBC         &$h=\hmsedpi$, $b=\bmsedpi$  & 0.040& 0.049& 0.69& 87.2& 0.83&&  0.121& 0.049& 0.69& 29.9& 0.69\\
      &RBC         &$h=b=\hmsedpi$              & 0.022& 0.058& 0.69& 93.5& 0.97&&  0.064& 0.058& 0.69& 79.9& 0.81\\
      &RBC         &$h=\hcedpi$, $b=\bcedpi$    & 0.026& 0.054& 0.46& 91.4& 0.90&&  0.074& 0.054& 0.44& 69.6& 0.76\\
      &RBC         &$h=b=\hhrmse{2}$            & 0.005& 0.069& 0.43& 94.5& 1.16&&  0.014& 0.070& 0.43& 94.0& 0.98\\
      &RBC         &$h=b=\hhrmse{6}$            & 0.000& 0.083& 0.28& 94.4& 1.39&&  0.001& 0.083& 0.27& 94.4& 1.17\\
      &RBC         &$h=b=\hhrmse{\Mrot}$        & 0.003& 0.081& 0.31& 94.5& 1.36&&  0.003& 0.090& 0.24& 93.5& 1.26\\
      &Conventional&$\hfgrot$                   & 0.034& 0.043& 0.85& 84.1& 0.73&&  0.091& 0.047& 0.61& 49.6& 0.66\\
      &Conventional&$\hcedpi$                   & 0.027& 0.050& 0.46& 90.7& 0.85&&  0.076& 0.051& 0.44& 65.0& 0.72\\
      &FLCI, $\SC=2$ &$\hhrmse{2}$              & 0.026& 0.051& 0.43& 95.6& 1.00&&  0.076& 0.052& 0.43& 77.7& 0.83\\
      &FLCI, $\SC=6$ &$\hhrmse{6}$              & 0.013& 0.061& 0.28& 97.3& 1.20&&  0.037& 0.061& 0.27& 94.8& 1.00\\
      &FLCI, $\SC=\Mrot$&$\hhrmse{\Mrot}$       & 0.015& 0.060& 0.31& 96.5& 1.18&&  0.029& 0.066& 0.24& 92.8& 1.08\\
      \multicolumn{2}{@{}l}{Design 3}\\
      \cmidrule(r){1-2}
      &RBC         &$h=\hmsedpi$, $b=\bmsedpi$  & -0.040& 0.049& 0.69& 87.2& 0.83&& -0.118& 0.049& 0.69& 33.2& 0.69\\
      &RBC         &$h=b=\hmsedpi$              & -0.020& 0.058& 0.69& 93.4& 0.97&& -0.058& 0.058& 0.69& 81.1& 0.81\\
      &RBC         &$h=\hcedpi$, $b=\bcedpi$    & -0.024& 0.054& 0.46& 91.4& 0.90&& -0.066& 0.055& 0.44& 74.3& 0.77\\
      &RBC         &$h=b=\hhrmse{2}$            & -0.005& 0.069& 0.43& 94.5& 1.17&& -0.014& 0.070& 0.42& 93.9& 0.98\\
      &RBC         &$h=b=\hhrmse{6}$            &  0.000& 0.083& 0.28& 94.4& 1.39&& -0.001& 0.084& 0.27& 94.3& 1.17\\
      &RBC         &$h=b=\hhrmse{\Mrot}$        & -0.002& 0.084& 0.29& 94.4& 1.41&& -0.001& 0.099& 0.19& 93.9& 1.38\\
      &Conventional&$\hfgrot$                   & -0.035& 0.044& 0.82& 83.2& 0.73&& -0.085& 0.049& 0.53& 56.2& 0.69\\
      &Conventional&$\hcedpi$                   & -0.026& 0.050& 0.46& 90.4& 0.85&& -0.075& 0.052& 0.44& 65.1& 0.72\\
      &FLCI, $\SC=2$ &$\hhrmse{2}$              & -0.026& 0.051& 0.43& 95.6& 1.00&& -0.075& 0.052& 0.42& 78.2& 0.83\\
      &FLCI, $\SC=6$ &$\hhrmse{6}$              & -0.013& 0.061& 0.28& 97.2& 1.20&& -0.037& 0.062& 0.27& 94.7& 1.00\\
      &FLCI, $\SC=\Mrot$&$\hhrmse{\Mrot}$       & -0.013& 0.062& 0.29& 96.8& 1.22&& -0.018& 0.072& 0.19& 96.5& 1.18\\
    \end{longtable}
  \end{ThreePartTable}
\end{landscape}

\begin{landscape}
  \renewcommand{\arraystretch}{1.1} %
  \footnotesize
  \begin{ThreePartTable}
    \begin{TableNotes}
    \item \emph{Legend:} SE---average standard error; $E[h]$---average (over
      Monte Carlo draws) bandwidth; Cov---coverage of CIs (in \%); RL---relative
      (to optimal FLCI) length.
    \item \emph{Bandwidth descriptions:} $\hmsedpi$---plugin estimate of
      pointwise MSE optimal bandwidth (bw); $\bmsedpi$---analog for estimate of
      the bias; $\hcedpi$---plugin estimate of coverage error optimal bw;
      $\bcedpi$---analog for estimate of the bias; The implementation of
      \citet{ccf15} is used for all four bws. $\hhrmse{2}$, $\hhrmse{6}$---RMSE
      optimal bw, assuming $\SC=2$, and $\SC=6$, respectively.
      $\hfgrot$---\citet{fg96} rule of thumb; $\hhrmse{\Mrot}$---RMSE optimal
      bw, using rule-of-thumb for $\SC$. 50,000 Monte Carlo draws.
    \end{TableNotes}
    \begin{longtable}{@{}lll rllrl c rllrl ll@{}}
      \caption{Monte Carlo simulation: heteroskedastic errors and beta distribution
      for $x_{i}$}\label{tab:mc5}\\
      &&&\multicolumn{5}{c@{}}{$\SC=2$}&&\multicolumn{5}{c}{$\SC=6$}\\
      \cmidrule(rl){4-8}\cmidrule(rl){10-14}
      &Method & Bandwidth & Bias& SE &   $E[h]$&   Cov&   RL&&Bias& SE &   $E_{m}[h]$&   Cov&   RL\\
      \midrule
      \endfirsthead%
      \caption*{Monte Carlo simulation: heteroskedastic errors and beta distribution
      for $x_{i}$ (continued)}\\
      &&&\multicolumn{5}{c@{}}{$\SC=2$}&&\multicolumn{5}{c}{$\SC=6$}\\
      \cmidrule(rl){4-8}\cmidrule(rl){10-14}
      &Method & Bandwidth & Bias& SE &   $E[h]$&   Cov&   RL&&Bias& SE &   $E_{m}[h]$&   Cov&   RL\\
      \midrule
      \endhead%
      \endfoot%
      \insertTableNotes%
      \endlastfoot%

      \multicolumn{2}{@{}l}{Design 1}\\
      \cmidrule(r){1-2} \phantom{a}
      &RBC         &$h=\hmsedpi$, $b=\bmsedpi$  & 0.027& 0.050& 0.50& 90.8& 0.90&&  0.062& 0.052& 0.44& 72.1& 0.79\\
      &RBC         &$h=b=\hmsedpi$              & 0.006& 0.059& 0.50& 94.5& 1.05&&  0.011& 0.062& 0.44& 93.0& 0.94\\
      &RBC         &$h=\hcedpi$, $b=\bcedpi$    & 0.009& 0.057& 0.37& 94.0& 1.03&&  0.015& 0.060& 0.31& 92.4& 0.92\\
      &RBC         &$h=b=\hhrmse{2}$            & 0.003& 0.062& 0.44& 94.6& 1.11&&  0.008& 0.062& 0.43& 94.4& 0.95\\
      &RBC         &$h=b=\hhrmse{6}$            & 0.000& 0.075& 0.27& 94.5& 1.35&&  0.000& 0.076& 0.26& 94.5& 1.16\\
      &RBC         &$h=b=\hhrmse{\Mrot}$        & 0.001& 0.083& 0.23& 94.3& 1.50&&  0.001& 0.096& 0.16& 94.0& 1.46\\
      &Conventional&$\hfgrot$                   & 0.029& 0.048& 0.64& 88.3& 0.86&&  0.054& 0.052& 0.36& 78.5& 0.80\\
      &Conventional&$\hcedpi$                   & 0.018& 0.052& 0.37& 92.4& 0.94&&  0.043& 0.055& 0.31& 84.0& 0.83\\
      &FLCI, $\SC=2$ &$\hhrmse{2}$              & 0.025& 0.049& 0.44& 94.9& 1.00&&  0.073& 0.050& 0.43& 76.0& 0.85\\
      &FLCI, $\SC=6$ &$\hhrmse{6}$              & 0.012& 0.057& 0.27& 97.1& 1.18&&  0.033& 0.058& 0.26& 94.7& 1.00\\
      &FLCI, $\SC=\Mrot$&$\hhrmse{\Mrot}$       & 0.008& 0.063& 0.23& 97.0& 1.31&&  0.013& 0.070& 0.16& 96.8& 1.24\\
      \multicolumn{2}{@{}l}{Design 2}\\
      \cmidrule(r){1-2}
      &RBC         &$h=\hmsedpi$, $b=\bmsedpi$  & 0.024& 0.050& 0.50& 91.8& 0.89&&  0.069& 0.050& 0.49& 69.8& 0.76\\
      &RBC         &$h=b=\hmsedpi$              & 0.010& 0.058& 0.50& 94.3& 1.05&&  0.025& 0.059& 0.49& 91.3& 0.90\\
      &RBC         &$h=\hcedpi$, $b=\bcedpi$    & 0.012& 0.057& 0.37& 93.8& 1.02&&  0.032& 0.058& 0.36& 89.0& 0.88\\
      &RBC         &$h=b=\hhrmse{2}$            & 0.006& 0.061& 0.44& 94.5& 1.10&&  0.017& 0.062& 0.44& 93.6& 0.94\\
      &RBC         &$h=b=\hhrmse{6}$            & 0.001& 0.075& 0.27& 94.5& 1.35&&  0.001& 0.075& 0.27& 94.5& 1.15\\
      &RBC         &$h=b=\hhrmse{\Mrot}$        & 0.001& 0.081& 0.25& 94.3& 1.46&&  0.002& 0.085& 0.22& 94.0& 1.29\\
      &Conventional&$\hfgrot$                   & 0.028& 0.047& 0.73& 89.6& 0.84&&  0.072& 0.048& 0.57& 64.1& 0.74\\
      &Conventional&$\hcedpi$                   & 0.018& 0.052& 0.37& 92.7& 0.93&&  0.050& 0.053& 0.36& 80.8& 0.80\\
      &FLCI, $\SC=2$ &$\hhrmse{2}$              & 0.023& 0.049& 0.44& 95.4& 1.00&&  0.068& 0.049& 0.44& 79.3& 0.85\\
      &FLCI, $\SC=6$ &$\hhrmse{6}$              & 0.012& 0.057& 0.27& 97.1& 1.18&&  0.034& 0.057& 0.27& 94.8& 1.00\\
      &FLCI, $\SC=\Mrot$&$\hhrmse{\Mrot}$       & 0.009& 0.061& 0.25& 97.1& 1.27&&  0.023& 0.063& 0.22& 94.8& 1.12\\
      \multicolumn{2}{@{}l}{Design 3}\\
      \cmidrule(r){1-2}
      &RBC         &$h=\hmsedpi$, $b=\bmsedpi$  & -0.027& 0.050& 0.50& 91.0& 0.90&& -0.071& 0.051& 0.47& 69.0& 0.78\\
      &RBC         &$h=b=\hmsedpi$              & -0.008& 0.059& 0.50& 94.5& 1.05&& -0.020& 0.060& 0.47& 92.8& 0.92\\
      &RBC         &$h=\hcedpi$, $b=\bcedpi$    & -0.011& 0.057& 0.37& 93.8& 1.03&& -0.021& 0.060& 0.32& 91.6& 0.92\\
      &RBC         &$h=b=\hhrmse{2}$            & -0.005& 0.061& 0.44& 94.6& 1.10&& -0.016& 0.062& 0.44& 93.7& 0.94\\
      &RBC         &$h=b=\hhrmse{6}$            &  0.000& 0.075& 0.27& 94.5& 1.35&&  0.000& 0.076& 0.27& 94.5& 1.15\\
      &RBC         &$h=b=\hhrmse{\Mrot}$        &  0.000& 0.083& 0.23& 94.3& 1.49&&  0.000& 0.092& 0.17& 94.2& 1.40\\
      &Conventional&$\hfgrot$                   & -0.027& 0.047& 0.70& 89.4& 0.85&& -0.066& 0.050& 0.49& 69.4& 0.76\\
      &Conventional&$\hcedpi$                   & -0.017& 0.052& 0.37& 93.0& 0.93&& -0.043& 0.054& 0.32& 84.5& 0.82\\
      &FLCI, $\SC=2$ &$\hhrmse{2}$              & -0.022& 0.049& 0.44& 95.5& 1.00&& -0.067& 0.049& 0.44& 79.9& 0.85\\
      &FLCI, $\SC=6$ &$\hhrmse{6}$              & -0.011& 0.057& 0.27& 97.2& 1.18&& -0.032& 0.057& 0.27& 94.9& 1.00\\
      &FLCI, $\SC=\Mrot$&$\hhrmse{\Mrot}$       & -0.008& 0.062& 0.23& 97.1& 1.30&& -0.014& 0.068& 0.17& 96.8& 1.21\\
    \end{longtable}
  \end{ThreePartTable}
\end{landscape}

\begin{landscape}
  \renewcommand{\arraystretch}{1.1} %
  \footnotesize
  \begin{ThreePartTable}
    \begin{TableNotes}
    \item \emph{Legend:} SE---average standard error; $E[h]$---average (over
      Monte Carlo draws) bandwidth; Cov---coverage of CIs (in \%); RL---relative
      (to optimal FLCI) length.
    \item \emph{Bandwidth descriptions:} $\hmsedpi$---plugin estimate of
      pointwise MSE optimal bandwidth (bw); $\bmsedpi$---analog for estimate of
      the bias; $\hcedpi$---plugin estimate of coverage error optimal bw;
      $\bcedpi$---analog for estimate of the bias; The implementation of
      \citet{ccf15} is used for all four bws. $\hhrmse{2}$, $\hhrmse{6}$---RMSE
      optimal bw, assuming $\SC=2$, and $\SC=6$, respectively.
      $\hfgrot$---\citet{fg96} rule of thumb; $\hhrmse{\Mrot}$---RMSE optimal
      bw, using rule-of-thumb for $\SC$. 50,000 Monte Carlo draws.
    \end{TableNotes}
    \begin{longtable}{@{}lll rllrl c rllrl ll@{}}
      \caption{Monte Carlo simulation: log-normal errors}\label{tab:mc6}\\
      &&&\multicolumn{5}{c@{}}{$\SC=2$}&&\multicolumn{5}{c}{$\SC=6$}\\
      \cmidrule(rl){4-8}\cmidrule(rl){10-14}
      &Method & Bandwidth & Bias& SE &   $E[h]$&   Cov&   RL&&Bias& SE &   $E_{m}[h]$&   Cov&   RL\\
      \midrule
      \endfirsthead%
      \caption*{Monte Carlo simulation: log-normal errors (continued)}\\
      &&&\multicolumn{5}{c@{}}{$\SC=2$}&&\multicolumn{5}{c}{$\SC=6$}\\
      \cmidrule(rl){4-8}\cmidrule(rl){10-14}
      &Method & Bandwidth & Bias& SE &   $E[h]$&   Cov&   RL&&Bias& SE &   $E_{m}[h]$&   Cov&   RL\\
      \midrule
      \endhead%
      \endfoot%
      \insertTableNotes%
      \endlastfoot%

      \multicolumn{2}{@{}l}{Design 1}\\
      \cmidrule(r){1-2} \phantom{a}
      &RBC         &$h=\hmsedpi$, $b=\bmsedpi$  & 0.062& 0.034& 0.73& 57.3& 0.73&&  0.151& 0.035& 0.60&  0.2& 0.62\\
      &RBC         &$h=b=\hmsedpi$              & 0.022& 0.041& 0.73& 94.5& 0.88&&  0.036& 0.045& 0.60& 91.5& 0.78\\
      &RBC         &$h=\hcedpi$, $b=\bcedpi$    & 0.042& 0.037& 0.55& 83.0& 0.79&&  0.111& 0.037& 0.50& 18.8& 0.66\\
      &RBC         &$h=b=\hhrmse{2}$            & 0.001& 0.058& 0.35& 91.1& 1.24&&  0.003& 0.057& 0.35& 91.5& 1.01\\
      &RBC         &$h=b=\hhrmse{6}$            & 0.000& 0.070& 0.23& 89.7& 1.52&&  0.000& 0.070& 0.23& 89.6& 1.23\\
      &RBC         &$h=b=\hhrmse{\Mrot}$        & 0.000& 0.072& 0.22& 89.3& 1.56&&  0.000& 0.087& 0.14& 87.6& 1.54\\
      &Conventional&$\hfgrot$                   & 0.032& 0.034& 0.55& 78.7& 0.74&&  0.048& 0.044& 0.31& 82.2& 0.77\\
      &Conventional&$\hcedpi$                   & 0.041& 0.034& 0.55& 81.2& 0.73&&  0.107& 0.035& 0.50& 16.4& 0.61\\
      &FLCI, $\SC=2$ &$\hhrmse{2}$              & 0.021& 0.041& 0.35& 96.2& 1.00&&  0.062& 0.041& 0.35& 79.1& 0.81\\
      &FLCI, $\SC=6$ &$\hhrmse{6}$              & 0.009& 0.051& 0.23& 95.5& 1.23&&  0.027& 0.050& 0.23& 96.2& 1.00\\
      &FLCI, $\SC=\Mrot$&$\hhrmse{\Mrot}$       & 0.007& 0.052& 0.22& 94.6& 1.27&&  0.010& 0.064& 0.14& 94.8& 1.27\\
      \multicolumn{2}{@{}l}{Design 2}\\
      \cmidrule(r){1-2}
      &RBC         &$h=\hmsedpi$, $b=\bmsedpi$  &  0.042& 0.033& 0.76& 80.5& 0.72&&  0.127& 0.033& 0.76&  2.5& 0.59\\
      &RBC         &$h=b=\hmsedpi$              &  0.024& 0.040& 0.76& 93.3& 0.86&&  0.073& 0.040& 0.76& 52.2& 0.70\\
      &RBC         &$h=\hcedpi$, $b=\bcedpi$    &  0.033& 0.036& 0.56& 89.8& 0.79&&  0.097& 0.037& 0.55& 19.0& 0.64\\
      &RBC         &$h=b=\hhrmse{2}$            &  0.002& 0.057& 0.35& 91.3& 1.24&&  0.006& 0.057& 0.35& 91.9& 1.01\\
      &RBC         &$h=b=\hhrmse{6}$            &  0.000& 0.070& 0.23& 89.6& 1.52&&  0.000& 0.070& 0.23& 89.7& 1.23\\
      &RBC         &$h=b=\hhrmse{\Mrot}$        &  0.002& 0.063& 0.29& 90.3& 1.37&&  0.000& 0.075& 0.19& 89.0& 1.32\\
      &Conventional&$\hfgrot$                   &  0.032& 0.030& 0.76& 77.4& 0.66&&  0.072& 0.037& 0.43& 53.3& 0.66\\
      &Conventional&$\hcedpi$                   &  0.034& 0.033& 0.56& 87.8& 0.72&&  0.099& 0.034& 0.55& 11.4& 0.59\\
      &FLCI, $\SC=2$ &$\hhrmse{2}$              &  0.019& 0.041& 0.35& 96.4& 1.00&&  0.059& 0.041& 0.35& 83.5& 0.81\\
      &FLCI, $\SC=6$ &$\hhrmse{6}$              &  0.009& 0.051& 0.23& 95.6& 1.23&&  0.027& 0.051& 0.23& 96.5& 1.00\\
      &FLCI, $\SC=\Mrot$&$\hhrmse{\Mrot}$       &  0.013& 0.046& 0.29& 94.7& 1.11&&  0.019& 0.055& 0.19& 94.9& 1.08\\
      \multicolumn{2}{@{}l}{Design 3}\\
      \cmidrule(r){1-2}
      &RBC         &$h=\hmsedpi$, $b=\bmsedpi$  & -0.043& 0.034& 0.76& 67.1& 0.73&& -0.121& 0.034& 0.72& 12.3& 0.60\\
      &RBC         &$h=b=\hmsedpi$              & -0.024& 0.040& 0.76& 83.9& 0.86&& -0.065& 0.041& 0.72& 56.1& 0.72\\
      &RBC         &$h=\hcedpi$, $b=\bcedpi$    & -0.030& 0.037& 0.55& 78.7& 0.80&& -0.077& 0.039& 0.51& 45.8& 0.69\\
      &RBC         &$h=b=\hhrmse{2}$            & -0.002& 0.057& 0.36& 90.7& 1.23&& -0.006& 0.057& 0.36& 90.1& 1.00\\
      &RBC         &$h=b=\hhrmse{6}$            &  0.000& 0.069& 0.23& 89.7& 1.50&&  0.000& 0.069& 0.23& 89.7& 1.22\\
      &RBC         &$h=b=\hhrmse{\Mrot}$        &  0.000& 0.068& 0.25& 89.6& 1.48&&  0.000& 0.083& 0.16& 88.1& 1.46\\
      &Conventional&$\hfgrot$                   & -0.031& 0.033& 0.71& 69.3& 0.72&& -0.063& 0.040& 0.39& 54.9& 0.71\\
      &Conventional&$\hcedpi$                   & -0.033& 0.034& 0.55& 74.2& 0.73&& -0.093& 0.035& 0.51& 27.8& 0.61\\
      &FLCI, $\SC=2$ &$\hhrmse{2}$              & -0.020& 0.041& 0.36& 89.9& 1.00&& -0.059& 0.041& 0.36& 70.1& 0.81\\
      &FLCI, $\SC=6$ &$\hhrmse{6}$              & -0.009& 0.050& 0.23& 93.0& 1.23&& -0.027& 0.050& 0.23& 88.4& 1.00\\
      &FLCI, $\SC=\Mrot$&$\hhrmse{\Mrot}$       & -0.010& 0.050& 0.25& 91.8& 1.22&& -0.012& 0.060& 0.16& 91.6& 1.21\\
    \end{longtable}
  \end{ThreePartTable}
\end{landscape}

\begin{landscape}
  \renewcommand{\arraystretch}{1.1} %
  \footnotesize
  \begin{ThreePartTable}
    \begin{TableNotes}
    \item \emph{Legend:} SE---average standard error; $E[h]$---average (over
      Monte Carlo draws) bandwidth; Cov---coverage of CIs (in \%); RL---relative
      (to optimal FLCI) length.
    \item \emph{Bandwidth descriptions:} $\hmsedpi$---plugin estimate of
      pointwise MSE optimal bandwidth (bw); $\bmsedpi$---analog for estimate of
      the bias; $\hcedpi$---plugin estimate of coverage error optimal bw;
      $\bcedpi$---analog for estimate of the bias; The implementation of
      \citet{ccf15} is used for all four bws. $\hhrmse{2}$, $\hhrmse{6}$---RMSE
      optimal bw, assuming $\SC=2$, and $\SC=6$, respectively.
      $\hfgrot$---\citet{fg96} rule of thumb; $\hhrmse{\Mrot}$---RMSE optimal
      bw, using rule-of-thumb for $\SC$. 50,000 Monte Carlo draws.
    \end{TableNotes}
    \begin{longtable}{@{}lll rllrl c rllrl ll@{}}
      \caption{Monte Carlo simulation: log-normal errors and beta distribution
      for $x_{i}$}\label{tab:mc7}\\
      &&&\multicolumn{5}{c@{}}{$\SC=2$}&&\multicolumn{5}{c}{$\SC=6$}\\
      \cmidrule(rl){4-8}\cmidrule(rl){10-14}
      &Method & Bandwidth & Bias& SE &   $E[h]$&   Cov&   RL&&Bias& SE &   $E_{m}[h]$&   Cov&   RL\\
      \midrule
      \endfirsthead%
      \caption*{Monte Carlo simulation: log-normal errors and beta distribution
      for $x_{i}$ (continued)}\\
      &&&\multicolumn{5}{c@{}}{$\SC=2$}&&\multicolumn{5}{c}{$\SC=6$}\\
      \cmidrule(rl){4-8}\cmidrule(rl){10-14}
      &Method & Bandwidth & Bias& SE &   $E[h]$&   Cov&   RL&&Bias& SE &   $E_{m}[h]$&   Cov&   RL\\
      \midrule
      \endhead%
      \endfoot%
      \insertTableNotes%
      \endlastfoot%

      \multicolumn{2}{@{}l}{Design 1}\\
      \cmidrule(r){1-2} \phantom{a}
      &RBC         &$h=\hmsedpi$, $b=\bmsedpi$  & 0.027& 0.035& 0.55& 88.4& 0.82&&  0.049& 0.039& 0.41& 65.1& 0.73\\
      &RBC         &$h=b=\hmsedpi$              & 0.006& 0.041& 0.55& 91.7& 0.96&&  0.007& 0.047& 0.41& 87.0& 0.90\\
      &RBC         &$h=\hcedpi$, $b=\bcedpi$    & 0.011& 0.041& 0.46& 92.2& 0.94&&  0.015& 0.045& 0.41& 88.9& 0.85\\
      &RBC         &$h=b=\hhrmse{2}$            & 0.001& 0.052& 0.36& 91.7& 1.19&&  0.004& 0.051& 0.36& 91.9& 0.98\\
      &RBC         &$h=b=\hhrmse{6}$            & 0.000& 0.064& 0.22& 90.1& 1.48&&  0.000& 0.064& 0.22& 90.2& 1.21\\
      &RBC         &$h=b=\hhrmse{\Mrot}$        & 0.000& 0.067& 0.21& 89.8& 1.56&&  0.000& 0.081& 0.13& 88.5& 1.54\\
      &Conventional&$\hfgrot$                   & 0.024& 0.035& 0.52& 87.9& 0.82&&  0.037& 0.043& 0.28& 87.9& 0.81\\
      &Conventional&$\hcedpi$                   & 0.026& 0.036& 0.46& 91.3& 0.82&&  0.064& 0.037& 0.41& 51.5& 0.70\\
      &FLCI, $\SC=2$ &$\hhrmse{2}$              & 0.019& 0.039& 0.36& 96.2& 1.00&&  0.055& 0.039& 0.36& 81.0& 0.82\\
      &FLCI, $\SC=6$ &$\hhrmse{6}$              & 0.008& 0.047& 0.22& 95.7& 1.22&&  0.024& 0.047& 0.22& 96.1& 1.00\\
      &FLCI, $\SC=\Mrot$&$\hhrmse{\Mrot}$       & 0.006& 0.049& 0.21& 95.0& 1.28&&  0.009& 0.059& 0.13& 94.9& 1.27\\
      \multicolumn{2}{@{}l}{Design 2}\\
      \cmidrule(r){1-2}
      &RBC         &$h=\hmsedpi$, $b=\bmsedpi$  & 0.024& 0.035& 0.56& 91.4& 0.81&&  0.067& 0.036& 0.52& 49.0& 0.68\\
      &RBC         &$h=b=\hmsedpi$              & 0.010& 0.041& 0.56& 92.2& 0.95&&  0.028& 0.042& 0.52& 81.6& 0.81\\
      &RBC         &$h=\hcedpi$, $b=\bcedpi$    & 0.015& 0.040& 0.47& 92.9& 0.93&&  0.037& 0.042& 0.45& 81.9& 0.79\\
      &RBC         &$h=b=\hhrmse{2}$            & 0.002& 0.052& 0.36& 91.8& 1.19&&  0.007& 0.051& 0.36& 92.4& 0.98\\
      &RBC         &$h=b=\hhrmse{6}$            & 0.000& 0.064& 0.22& 90.0& 1.48&&  0.000& 0.064& 0.22& 90.1& 1.21\\
      &RBC         &$h=b=\hhrmse{\Mrot}$        & 0.001& 0.062& 0.25& 90.4& 1.44&&  0.000& 0.069& 0.19& 89.5& 1.31\\
      &Conventional&$\hfgrot$                   & 0.026& 0.033& 0.69& 88.9& 0.77&&  0.059& 0.037& 0.43& 64.9& 0.70\\
      &Conventional&$\hcedpi$                   & 0.023& 0.036& 0.47& 92.5& 0.82&&  0.068& 0.036& 0.45& 48.7& 0.68\\
      &FLCI, $\SC=2$ &$\hhrmse{2}$              & 0.018& 0.039& 0.36& 96.4& 1.00&&  0.053& 0.039& 0.36& 84.5& 0.82\\
      &FLCI, $\SC=6$ &$\hhrmse{6}$              & 0.008& 0.047& 0.22& 95.8& 1.22&&  0.024& 0.047& 0.22& 96.4& 1.00\\
      &FLCI, $\SC=\Mrot$&$\hhrmse{\Mrot}$       & 0.009& 0.046& 0.25& 95.3& 1.19&&  0.018& 0.050& 0.19& 95.0& 1.08\\
      \multicolumn{2}{@{}l}{Design 3}\\
      \cmidrule(r){1-2}
      &RBC         &$h=\hmsedpi$, $b=\bmsedpi$  & -0.032& 0.035& 0.55& 77.2& 0.82&& -0.068& 0.038& 0.47& 50.6& 0.72\\
      &RBC         &$h=b=\hmsedpi$              & -0.015& 0.041& 0.55& 88.9& 0.95&& -0.025& 0.045& 0.47& 85.4& 0.85\\
      &RBC         &$h=\hcedpi$, $b=\bcedpi$    & -0.014& 0.041& 0.44& 87.8& 0.96&& -0.022& 0.045& 0.37& 85.4& 0.86\\
      &RBC         &$h=b=\hhrmse{2}$            & -0.002& 0.051& 0.36& 91.0& 1.18&& -0.006& 0.051& 0.36& 90.5& 0.97\\
      &RBC         &$h=b=\hhrmse{6}$            &  0.000& 0.064& 0.23& 90.1& 1.47&&  0.000& 0.064& 0.23& 90.1& 1.21\\
      &RBC         &$h=b=\hhrmse{\Mrot}$        &  0.000& 0.066& 0.22& 89.7& 1.53&&  0.000& 0.077& 0.15& 88.7& 1.47\\
      &Conventional&$\hfgrot$                   & -0.025& 0.035& 0.64& 77.5& 0.81&& -0.050& 0.040& 0.36& 63.6& 0.76\\
      &Conventional&$\hcedpi$                   & -0.023& 0.036& 0.44& 82.3& 0.84&& -0.055& 0.038& 0.37& 62.4& 0.73\\
      &FLCI, $\SC=2$ &$\hhrmse{2}$              & -0.018& 0.039& 0.36& 90.1& 1.00&& -0.053& 0.039& 0.36& 71.1& 0.82\\
      &FLCI, $\SC=6$ &$\hhrmse{6}$              & -0.008& 0.047& 0.23& 93.1& 1.22&& -0.024& 0.047& 0.23& 88.6& 1.00\\
      &FLCI, $\SC=\Mrot$&$\hhrmse{\Mrot}$       & -0.007& 0.049& 0.22& 92.6& 1.27&& -0.011& 0.056& 0.15& 91.9& 1.22\\
    \end{longtable}
  \end{ThreePartTable}
\end{landscape}

\begin{landscape}
  \renewcommand{\arraystretch}{1.1} %
  \footnotesize
  \begin{ThreePartTable}
    \begin{TableNotes}
    \item \emph{Legend:} SE---average standard error; $E[h]$---average (over
      Monte Carlo draws) bandwidth; Cov---coverage of CIs (in \%); RL---relative
      (to optimal FLCI) length.
    \item \emph{Bandwidth descriptions:} $\hmsedpi$---plugin estimate of
      pointwise MSE optimal bandwidth (bw); $\bmsedpi$---analog for estimate of
      the bias; $\hcedpi$---plugin estimate of coverage error optimal bw;
      $\bcedpi$---analog for estimate of the bias; The implementation of
      \citet{ccf15} is used for all four bws. $\hhrmse{2}$, $\hhrmse{6}$---RMSE
      optimal bw, assuming $\SC=2$, and $\SC=6$, respectively.
      $\hfgrot$---\citet{fg96} rule of thumb; $\hhrmse{\Mrot}$---RMSE optimal
      bw, using rule-of-thumb for $\SC$. 50,000 Monte Carlo draws.
    \end{TableNotes}
    \begin{longtable}{@{}lll rllrl c rllrl ll@{}}
      \caption{Monte Carlo simulation: $\sd(u_{i})=1/4$}\label{tab:mc8}\\
      &&&\multicolumn{5}{c@{}}{$\SC=2$}&&\multicolumn{5}{c}{$\SC=6$}\\
      \cmidrule(rl){4-8}\cmidrule(rl){10-14}
      &Method & Bandwidth & Bias& SE &   $E[h]$&   Cov&   RL&&Bias& SE &   $E_{m}[h]$&   Cov&   RL\\
      \midrule
      \endfirsthead%
      \caption*{Monte Carlo simulation: $\sd(u_{i})=1/4$ (continued)}\\
      &&&\multicolumn{5}{c@{}}{$\SC=2$}&&\multicolumn{5}{c}{$\SC=6$}\\
      \cmidrule(rl){4-8}\cmidrule(rl){10-14}
      &Method & Bandwidth & Bias& SE &   $E[h]$&   Cov&   RL&&Bias& SE &   $E_{m}[h]$&   Cov&   RL\\
      \midrule
      \endhead%
      \endfoot%
      \insertTableNotes%
      \endlastfoot%

      \multicolumn{2}{@{}l}{Design 1}\\
      \cmidrule(r){1-2} \phantom{a}
      &RBC         &$h=\hmsedpi$, $b=\bmsedpi$  & 0.058& 0.018& 0.68&  4.5& 0.64&&  0.116& 0.020& 0.49&  0.0& 0.57\\
      &RBC         &$h=b=\hmsedpi$              & 0.019& 0.022& 0.68& 90.2& 0.80&&  0.017& 0.026& 0.49& 91.2& 0.76\\
      &RBC         &$h=\hcedpi$, $b=\bcedpi$    & 0.024& 0.022& 0.39& 77.4& 0.78&&  0.041& 0.025& 0.28& 60.8& 0.72\\
      &RBC         &$h=b=\hhrmse{2}$            & 0.000& 0.035& 0.27& 94.3& 1.26&&  0.000& 0.035& 0.28& 94.3& 1.01\\
      &RBC         &$h=b=\hhrmse{6}$            & 0.000& 0.043& 0.18& 93.8& 1.57&&  0.000& 0.043& 0.18& 93.8& 1.26\\
      &RBC         &$h=b=\hhrmse{\Mrot}$        & 0.000& 0.045& 0.16& 93.7& 1.64&&  0.000& 0.056& 0.11& 93.0& 1.62\\
      &Conventional&$\hfgrot$                   & 0.022& 0.021& 0.38& 76.3& 0.77&&  0.028& 0.026& 0.24& 78.9& 0.77\\
      &Conventional&$\hcedpi$                   & 0.023& 0.021& 0.39& 76.9& 0.76&&  0.040& 0.024& 0.28& 61.0& 0.71\\
      &FLCI, $\SC=2$ &$\hhrmse{2}$              & 0.013& 0.025& 0.27& 94.7& 1.00&&  0.038& 0.025& 0.28& 73.9& 0.80\\
      &FLCI, $\SC=6$ &$\hhrmse{6}$              & 0.005& 0.031& 0.18& 96.5& 1.25&&  0.016& 0.031& 0.18& 94.5& 1.00\\
      &FLCI, $\SC=\Mrot$&$\hhrmse{\Mrot}$       & 0.004& 0.032& 0.16& 96.2& 1.30&&  0.006& 0.040& 0.11& 96.2& 1.29\\
      \multicolumn{2}{@{}l}{Design 2}\\
      \cmidrule(r){1-2}
      &RBC         &$h=\hmsedpi$, $b=\bmsedpi$  & 0.043& 0.017& 0.77& 28.8& 0.63&&  0.128& 0.017& 0.76&  0.0& 0.51\\
      &RBC         &$h=b=\hmsedpi$              & 0.026& 0.021& 0.77& 76.5& 0.75&&  0.075& 0.021& 0.76&  5.5& 0.61\\
      &RBC         &$h=\hcedpi$, $b=\bcedpi$    & 0.026& 0.020& 0.47& 70.2& 0.73&&  0.061& 0.022& 0.37& 24.8& 0.64\\
      &RBC         &$h=b=\hhrmse{2}$            & 0.000& 0.035& 0.27& 94.3& 1.26&&  0.001& 0.035& 0.28& 94.3& 1.01\\
      &RBC         &$h=b=\hhrmse{6}$            & 0.000& 0.043& 0.18& 93.8& 1.57&&  0.000& 0.043& 0.18& 93.8& 1.25\\
      &RBC         &$h=b=\hhrmse{\Mrot}$        & 0.000& 0.038& 0.24& 93.5& 1.38&&  0.000& 0.047& 0.15& 93.6& 1.37\\
      &Conventional&$\hfgrot$                   & 0.029& 0.018& 0.57& 58.5& 0.66&&  0.048& 0.023& 0.32& 46.0& 0.67\\
      &Conventional&$\hcedpi$                   & 0.027& 0.019& 0.47& 66.1& 0.69&&  0.062& 0.021& 0.37& 21.8& 0.62\\
      &FLCI, $\SC=2$ &$\hhrmse{2}$              & 0.012& 0.025& 0.27& 94.8& 1.00&&  0.039& 0.024& 0.28& 73.6& 0.80\\
      &FLCI, $\SC=6$ &$\hhrmse{6}$              & 0.005& 0.031& 0.18& 96.5& 1.25&&  0.016& 0.030& 0.18& 94.6& 1.00\\
      &FLCI, $\SC=\Mrot$&$\hhrmse{\Mrot}$       & 0.009& 0.027& 0.24& 92.9& 1.09&&  0.011& 0.033& 0.15& 95.4& 1.10\\
      \multicolumn{2}{@{}l}{Design 3}\\
      \cmidrule(r){1-2}
      &RBC         &$h=\hmsedpi$, $b=\bmsedpi$  & -0.042& 0.017& 0.76& 32.3& 0.63&& -0.107& 0.018& 0.63&  1.7& 0.54\\
      &RBC         &$h=b=\hmsedpi$              & -0.023& 0.021& 0.76& 77.5& 0.76&& -0.048& 0.023& 0.63& 45.4& 0.67\\
      &RBC         &$h=\hcedpi$, $b=\bcedpi$    & -0.024& 0.021& 0.46& 75.0& 0.75&& -0.046& 0.023& 0.35& 49.5& 0.68\\
      &RBC         &$h=b=\hhrmse{2}$            &  0.000& 0.035& 0.27& 94.3& 1.26&& -0.001& 0.035& 0.28& 94.4& 1.01\\
      &RBC         &$h=b=\hhrmse{6}$            &  0.000& 0.043& 0.18& 93.8& 1.57&&  0.000& 0.043& 0.18& 93.8& 1.25\\
      &RBC         &$h=b=\hhrmse{\Mrot}$        &  0.000& 0.042& 0.19& 93.8& 1.54&&  0.000& 0.053& 0.12& 93.3& 1.53\\
      &Conventional&$\hfgrot$                   & -0.026& 0.019& 0.49& 65.2& 0.69&& -0.041& 0.024& 0.29& 58.7& 0.70\\
      &Conventional&$\hcedpi$                   & -0.026& 0.019& 0.46& 66.5& 0.69&& -0.057& 0.022& 0.35& 30.4& 0.63\\
      &FLCI, $\SC=2$ &$\hhrmse{2}$              & -0.012& 0.025& 0.27& 94.7& 1.00&& -0.038& 0.025& 0.28& 73.9& 0.80\\
      &FLCI, $\SC=6$ &$\hhrmse{6}$              & -0.005& 0.031& 0.18& 96.4& 1.25&& -0.016& 0.031& 0.18& 94.5& 1.00\\
      &FLCI, $\SC=\Mrot$&$\hhrmse{\Mrot}$       & -0.006& 0.030& 0.19& 95.8& 1.22&& -0.007& 0.037& 0.12& 96.2& 1.23\\
    \end{longtable}
  \end{ThreePartTable}
\end{landscape}

\begin{landscape}
  \renewcommand{\arraystretch}{1.1} %
  \footnotesize
  \begin{ThreePartTable}
    \begin{TableNotes}
    \item \emph{Legend:} SE---average standard error; $E[h]$---average (over
      Monte Carlo draws) bandwidth; Cov---coverage of CIs (in \%); RL---relative
      (to optimal FLCI) length.
    \item \emph{Bandwidth descriptions:} $\hmsedpi$---plugin estimate of
      pointwise MSE optimal bandwidth (bw); $\bmsedpi$---analog for estimate of
      the bias; $\hcedpi$---plugin estimate of coverage error optimal bw;
      $\bcedpi$---analog for estimate of the bias; The implementation of
      \citet{ccf15} is used for all four bws. $\hhrmse{2}$, $\hhrmse{6}$---RMSE
      optimal bw, assuming $\SC=2$, and $\SC=6$, respectively.
      $\hfgrot$---\citet{fg96} rule of thumb; $\hhrmse{\Mrot}$---RMSE optimal
      bw, using rule-of-thumb for $\SC$. 50,000 Monte Carlo draws.
    \end{TableNotes}
    \begin{longtable}{@{}lll rllrl c rllrl ll@{}}
      \caption{Monte Carlo simulation: smooth DGP with $\lambda=40$}\label{tab:mc9}\\
      &&&\multicolumn{5}{c@{}}{$\SC=2$}&&\multicolumn{5}{c}{$\SC=6$}\\
      \cmidrule(rl){4-8}\cmidrule(rl){10-14}
      &Method & Bandwidth & Bias& SE &   $E[h]$&   Cov&   RL&&Bias& SE &   $E_{m}[h]$&   Cov&   RL\\
      \midrule
      \endfirsthead%
      \caption*{Monte Carlo simulation: smooth DGP, with $\lambda=40$ (continued)}\\
      &&&\multicolumn{5}{c@{}}{$\SC=2$}&&\multicolumn{5}{c}{$\SC=6$}\\
      \cmidrule(rl){4-8}\cmidrule(rl){10-14}
      &Method & Bandwidth & Bias& SE &   $E[h]$&   Cov&   RL&&Bias& SE &   $E_{m}[h]$&   Cov&   RL\\
      \midrule
      \endhead%
      \endfoot%
      \insertTableNotes%
      \endlastfoot%

      \multicolumn{2}{@{}l}{Design 1}\\
      \cmidrule(r){1-2} \phantom{a}
      &RBC         &$h=\hmsedpi$, $b=\bmsedpi$  & 0.062& 0.035& 0.74& 57.7& 0.73&&  0.151& 0.036& 0.61&  0.2& 0.61\\
      &RBC         &$h=b=\hmsedpi$              & 0.024& 0.042& 0.74& 93.3& 0.88&&  0.039& 0.047& 0.61& 90.1& 0.78\\
      &RBC         &$h=\hcedpi$, $b=\bcedpi$    & 0.029& 0.041& 0.46& 86.1& 0.85&&  0.059& 0.045& 0.34& 72.6& 0.76\\
      &RBC         &$h=b=\hhrmse{2}$            & 0.001& 0.061& 0.36& 94.5& 1.27&&  0.003& 0.061& 0.36& 94.5& 1.01\\
      &RBC         &$h=b=\hhrmse{6}$            & 0.000& 0.076& 0.23& 94.2& 1.58&&  0.000& 0.075& 0.23& 94.2& 1.26\\
      &RBC         &$h=b=\hhrmse{\Mrot}$        & 0.000& 0.078& 0.22& 93.9& 1.63&&  0.000& 0.097& 0.14& 93.4& 1.63\\
      &Conventional&$\hfgrot$                   & 0.032& 0.036& 0.57& 77.0& 0.75&&  0.050& 0.046& 0.32& 76.9& 0.77\\
      &Conventional&$\hcedpi$                   & 0.028& 0.039& 0.46& 85.7& 0.80&&  0.057& 0.044& 0.34& 72.8& 0.74\\
      &FLCI, $\SC=2$ &$\hhrmse{2}$              & 0.021& 0.043& 0.36& 95.0& 1.00&&  0.063& 0.043& 0.36& 76.2& 0.80\\
      &FLCI, $\SC=6$ &$\hhrmse{6}$              & 0.009& 0.054& 0.23& 96.6& 1.25&&  0.027& 0.053& 0.23& 94.7& 1.00\\
      &FLCI, $\SC=\Mrot$&$\hhrmse{\Mrot}$       & 0.008& 0.055& 0.22& 95.6& 1.29&&  0.010& 0.069& 0.14& 96.3& 1.29\\
      \multicolumn{2}{@{}l}{Design 2}\\
      \cmidrule(r){1-2}
      &RBC         &$h=\hmsedpi$, $b=\bmsedpi$  &  0.041& 0.035& 0.77& 77.4& 0.72&&  0.124& 0.035& 0.77&  5.4& 0.58\\
      &RBC         &$h=b=\hmsedpi$              &  0.024& 0.042& 0.77& 91.4& 0.87&&  0.072& 0.042& 0.77& 58.0& 0.70\\
      &RBC         &$h=\hcedpi$, $b=\bcedpi$    &  0.026& 0.040& 0.49& 88.1& 0.83&&  0.071& 0.041& 0.44& 56.4& 0.69\\
      &RBC         &$h=b=\hhrmse{2}$            &  0.002& 0.061& 0.36& 94.5& 1.27&&  0.007& 0.061& 0.36& 94.4& 1.01\\
      &RBC         &$h=b=\hhrmse{6}$            &  0.000& 0.076& 0.23& 94.2& 1.58&&  0.000& 0.075& 0.23& 94.2& 1.26\\
      &RBC         &$h=b=\hhrmse{\Mrot}$        &  0.002& 0.068& 0.30& 94.0& 1.43&&  0.000& 0.083& 0.20& 93.8& 1.38\\
      &Conventional&$\hfgrot$                   &  0.030& 0.032& 0.78& 76.0& 0.67&&  0.071& 0.040& 0.44& 54.7& 0.66\\
      &Conventional&$\hcedpi$                   &  0.027& 0.037& 0.49& 86.7& 0.77&&  0.072& 0.039& 0.44& 52.5& 0.66\\
      &FLCI, $\SC=2$ &$\hhrmse{2}$              &  0.019& 0.043& 0.36& 95.3& 1.00&&  0.058& 0.043& 0.36& 80.0& 0.80\\
      &FLCI, $\SC=6$ &$\hhrmse{6}$              &  0.009& 0.054& 0.23& 96.6& 1.25&&  0.027& 0.053& 0.23& 94.8& 1.00\\
      &FLCI, $\SC=\Mrot$&$\hhrmse{\Mrot}$       &  0.013& 0.048& 0.30& 94.5& 1.13&&  0.019& 0.059& 0.20& 94.4& 1.10\\
      \multicolumn{2}{@{}l}{Design 3}\\
      \cmidrule(r){1-2}
      &RBC         &$h=\hmsedpi$, $b=\bmsedpi$  & -0.041& 0.035& 0.77& 77.0& 0.72&& -0.119& 0.035& 0.74& 11.0& 0.59\\
      &RBC         &$h=b=\hmsedpi$              & -0.023& 0.042& 0.77& 91.3& 0.87&& -0.064& 0.042& 0.74& 62.4& 0.71\\
      &RBC         &$h=\hcedpi$, $b=\bcedpi$    & -0.025& 0.040& 0.49& 88.6& 0.83&& -0.061& 0.043& 0.43& 66.1& 0.71\\
      &RBC         &$h=b=\hhrmse{2}$            & -0.002& 0.061& 0.36& 94.5& 1.27&& -0.007& 0.061& 0.36& 94.3& 1.01\\
      &RBC         &$h=b=\hhrmse{6}$            &  0.000& 0.076& 0.23& 94.2& 1.58&& -0.001& 0.075& 0.23& 94.2& 1.26\\
      &RBC         &$h=b=\hhrmse{\Mrot}$        & -0.001& 0.074& 0.25& 94.2& 1.54&&  0.000& 0.092& 0.16& 93.6& 1.54\\
      &Conventional&$\hfgrot$                   & -0.030& 0.033& 0.72& 75.9& 0.69&& -0.062& 0.042& 0.39& 63.4& 0.70\\
      &Conventional&$\hcedpi$                   & -0.027& 0.037& 0.49& 86.4& 0.78&& -0.071& 0.040& 0.43& 53.9& 0.66\\
      &FLCI, $\SC=2$ &$\hhrmse{2}$              & -0.019& 0.043& 0.36& 95.1& 1.00&& -0.058& 0.043& 0.36& 79.8& 0.80\\
      &FLCI, $\SC=6$ &$\hhrmse{6}$              & -0.009& 0.054& 0.23& 96.5& 1.25&& -0.027& 0.053& 0.23& 94.7& 1.00\\
      &FLCI, $\SC=\Mrot$&$\hhrmse{\Mrot}$       & -0.010& 0.052& 0.25& 95.7& 1.22&& -0.013& 0.065& 0.16& 96.1& 1.22\\
    \end{longtable}
  \end{ThreePartTable}
\end{landscape}

\begin{figure}[p]
  \centering%
  \input{./figure_s1a.tex}
  \input{./figure_s1b.tex}
  \caption{Optimal equivalent kernels for Taylor class $\FSY{p}(\SC)$ on the
    interior, and in the boundary, rescaled to be supported on
    $[0,1]$ on the boundary and $[-1,1]$ in the interior.}\label{fig:sy}
\end{figure}

\begin{figure}[p]
  \centering%
  \input{./figure_s2a.tex}%
  \input{./figure_s2b.tex}
  \caption{Optimal equivalent kernels for Hölder class $\FHol{2}(\SC)$ on the
    interior, and in the boundary, rescaled to be supported on
    $[0,1]$ on the boundary and $[-1,1]$ in the interior.}\label{fig:hol2}
\end{figure}